\documentclass{article}

\usepackage{arxiv}

\usepackage[utf8]{inputenc} % allow utf-8 input
\usepackage[T1]{fontenc}    % use 8-bit T1 fonts
\usepackage[hidelinks]{hyperref}       % hyperlinks
\usepackage{url}            % simple URL typesetting
\usepackage{amssymb}
\usepackage{makecell}

\usepackage{authblk}
\usepackage{booktabs}       % professional-quality tables
\usepackage{hyperref}
\usepackage{amsfonts}       % blackboard math symbols
\usepackage{nicefrac}       % compact symbols for 1/2, etc.
\usepackage{graphicx}
\usepackage{doi}
\usepackage{float}
\usepackage{stfloats}
\usepackage{subcaption}
\usepackage{textcomp, gensymb}
\usepackage{enumitem}
\usepackage{multirow}
\usepackage{xcolor}
\usepackage{pgfplots}
\usepackage{wrapfig}
\usepackage{rotating}
\usepackage{todonotes}
\usepackage[ruled]{algorithm2e}
\usepackage{amsmath}
\usepackage{bbm} % \mathbbm{1}
\usepackage[most]{tcolorbox}
\usepackage{xparse}
\usepackage{cleveref}

\usepackage{amsthm}

\newtheorem{theorem}{Theorem}
\newtheorem{proposition}{Proposition}
\newtheorem{corollary}{Corollary}
\newtheorem{definition}{Definition}
\newtheorem{remark}{Remark}

\newtheorem{innercustomthm}{Theorem}
\newenvironment{customthm}[1]
  {\renewcommand\theinnercustomthm{#1}\innercustomthm}
  {\endinnercustomthm}
\newtheorem{innercustomprop}{Proposition}
\newenvironment{customprop}[1]
  {\renewcommand\theinnercustomprop{#1}\innercustomprop}
  {\endinnercustomprop}

\usepackage{array}
\newcolumntype{x}[1]{>{\centering\arraybackslash\hspace{0pt}}p{#1}}

\pgfplotsset{compat=1.18}

\usepackage{tikz}
\usepackage{tikz-qtree}
\usetikzlibrary{trees}
\usetikzlibrary{mindmap}
\usetikzlibrary{positioning}

\newcommand{\changed}[1]{{#1}}

\newcounter{llmprompt}

\NewEnviron{llmprompt}[3][0.9\linewidth]{%
  \refstepcounter{llmprompt}%
  \def\prompttitle{Prompt~\thellmprompt}%
  \if\relax\detokenize{#2}\relax
  \else
    \edef\prompttitle{Prompt~\thellmprompt: #2}%
  \fi
  \begin{tcolorbox}[float=!ht,
    width=#1,
    colback=gray!10,
    colframe=black,
    fonttitle=\bfseries,
    title=\prompttitle,
    sharp corners,
    boxrule=0.5pt,
    arc=2mm,
    left=4pt,
    right=4pt,
    top=4pt,
    bottom=4pt,
    listing only,
    listing options={
      basicstyle=\ttfamily\small,
      breaklines=false,
      breakautoindent=false,
      columns=fullflexible,
      keepspaces=true,
      showstringspaces=false,
      aboveskip=0pt,
      belowskip=0pt,
      frame=none,
      rulecolor=\color{white},
      rulesepcolor=\color{white}
    },
    label={#3}
  ]
  \BODY
  \end{tcolorbox}
}

\tcbset{
  myresultbox/.style={
    colback=white,       % no background color
    colframe=gray,      % black frame
    fonttitle=\bfseries, % bold title
    title=Result~\thetcbcounter,
    boxrule=0.5pt,       % thin line
    sharp corners,
    enhanced,
  }
}

\newtcolorbox[auto counter, number within=section]{resultbox}[1][]{
  myresultbox, #1
}

\tcbset{
  myobservationbox/.style={
    colback=white,       % no background color
    colframe=gray,      % black frame
    fonttitle=\bfseries, % bold title
    title=Observation~\thetcbcounter,
    boxrule=0.5pt,       % thin line
    sharp corners,
    enhanced,
  }
}

\newtcolorbox[auto counter, number within=section]{observationbox}[1][]{
  myobservationbox, #1
}

\begin{document}

\title{Statistical Detection of Positional Bias in LLM-Based Recommender Systems: Methodologies and Empirical Evidence}
\newcommand{\shorttitle}{Statistical Detection of Positional Bias in LLM-Based Recommender Systems}

\author{Carlos Pel\'aez-Gonz\'alez$^{*,1}$, Francisco Herrera$^{1}$, Andr\'es Herrera-Poyatos$^{2}$}

\affil{$^{1}$Department of Computer Science and Artificial Intelligence, Andalusian Institute of Data Science and Computational Intelligence (DaSCI), University of Granada, Spain. \\ 
$^{2}$ Department of Algebra, University of Granada, Spain.
$^{*}$Corresponding author. \\ Emails: \texttt{\{carlosprog, andreshp\}@ugr.es}, \texttt{herrera@decsai.ugr.es}}

\date{\today}

\maketitle              % typeset the header of the contribution

\begin{abstract}
Large Language Models (LLMs) are increasingly used as recommendation generators in conversational recommender systems (RS) via retrieval-augmented generation (RAG). This flexible design introduces new sources of structural bias. This paper studies \emph{positional bias} in \changed{top-1 and single-turn} LLM-based RS: systematic dependence of the recommended item on the order in which candidates are presented in the prompt. We propose two complementary detection methodologies: (i) \textbf{POB-ICM}, based on contingency matrix analysis and Chi-square independence testing, and (ii) \textbf{POB-UNBP}, a sample-efficient approach derived from a uniformity null model under permutation-based evaluation, enabling stochastic decoding with formal statistical guarantees. Using a controlled RAG testbed \changed{for product recommendation, built on a real product catalog, a fixed prompt template and} 2000 \changed{template-generated} user queries \changed{with} 20 permutations each (40{,}000 instances), both methodologies reject the null hypothesis of positional independence and uniformity in all evaluated state-of-the-art models. We introduce three complementary metrics — permutation sensitivity, positional bias, and strong positional bias — to quantify this effect\changed{, the latter targeting cases where the position displaces the recommendation towards clearly less relevant items rather than only choosing among comparably relevant ones}. A robustness analysis across candidate-set size and decoding temperature shows that the bias is structural and model-dependent, not an artifact of deterministic decoding or candidate-set size. These results establish positional bias as a measurable, reproducible property of prompt-mediated recommendation pipelines \changed{in the evaluated setting} and motivate routine auditing and governance controls (e.g., order randomization, standardized prompting) for deployed LLM-based RS. Source code is publicly released\footnote{\url{https://github.com/GoComputing/Positional-Bias.git}}.
\end{abstract}

\keywords {Large language models \and conversational recommender system \and positional bias}

\section{Introduction}

% Introduction to RS
Recommender systems (RS) are a core component of modern digital platforms, supporting users in navigating and filtering large volumes of information. By modeling user preferences and item characteristics, RS aims to deliver relevant and personalized recommendations across domains such as e-commerce, media consumption, and information retrieval. Despite substantial progress, RS remain susceptible to a range of biases that can affect both the reliability and the consistency of their outputs. Previous work has documented several forms of bias—including exposure bias, popularity bias, and positional bias—that can emerge at different stages of the recommendation pipeline and propagate through system components, ultimately shaping the user experience in unintended ways \cite{chen_bias_2023}. Understanding, measuring, and mitigating such biases is, therefore, essential for ensuring transparent and reliable recommendation behavior.

% Introduction to LLMs and capabilities
In parallel, Large Language Models (LLMs) have recently emerged as a powerful paradigm for building intelligent systems. Trained in large-scale corpora using self-supervised objectives, LLMs exhibit a wide range of emergent capabilities, including question answering, reasoning, and decision-making \cite{raiaan_review_2024}. This has led to a shift from traditional model training toward prompt-based interaction, where system behavior is largely determined by how information is presented in natural language form. As a consequence, LLMs are increasingly being adopted as central components in RS, particularly in conversational settings.

% Introduction to jailbreaking
% Beyond reliability concerns, sensitivity to non-semantic prompt structure has security implications for LLM-based RS. Prompt-based attacks ("jailbreaks") often exploit systematic model vulnerabilities that arise from generalization gaps, competing objectives, or robustness failures, enabling outputs that violate intended constraints. Recent work proposes a domain-based taxonomy of jailbreak vulnerabilities that emphasizes underlying model deficiencies rather than surface prompt templates~\cite{pelaez2025domain}. In this perspective, positional bias can be viewed as a measurable form of vulnerability to prompt-structure, motivating routine auditing of ordering effects in prompt-mediated recommendation pipelines.

% Introduction to LLM-based RS
A common design for LLM-based RS is based on a retrieval-augmented generation (RAG) architecture \cite{deldjoo2024review}. In this setting, a retriever first computes semantic similarity scores between a user query and a collection of candidate items using an embedding model. The top-$k$ most relevant items are then selected and inserted into a prompt, together with the user query, which is provided to the LLM. Based on this prompt, the LLM produces a recommendation, typically by selecting or describing one of the candidate items. This architecture decouples retrieval from generation while allowing the LLM to reason over a constrained candidate set.

% Importance of RS analysis
The increasing adoption of LLMs within information systems has been identified as a major research frontier in the Information Systems (IS) community. Recent work emphasizes that generative AI systems introduce new forms of socio-technical risk, requiring systematic analysis of their behavioral properties, reliability, and governance when deployed in real-world organizational settings~\cite{storey2025genai}. In this context, understanding how design choices in prompt-mediated systems affect model behavior is critical to advance both theory and practice in LLM-enabled information systems.

% Challenges on the analysis of LLMs as RS
However, LLM-based RS operate largely as black boxes. Although the retriever stage is often well understood and explicitly optimized, the final recommendation decision is made by the LLM through implicit reasoning over a sequential text prompt. As a result, unintended behaviors may arise that are not directly related to user preferences or item relevance. These behaviors are not necessarily associated with socially sensitive attributes such as gender or race but instead stem from subtle interactions between prompt structure and model internals. Ideally, a recommendation should depend only on the semantic relationship between the user query and the candidate items. In practice, additional latent variables may influence the outcome.

% Introduction to positional bias
One such variable is the position of candidate items within the prompt. Because text is inherently sequential, the retrieved items must be presented in a specific order. This ordering introduces positional information that is not semantically meaningful, but can nevertheless affect the LLM’s output. When the probability of recommending an item depends on its position in the input list rather than solely on its relevance, the system exhibits positional bias. Although positional bias has been extensively studied in traditional RS and user interfaces, its manifestation and measurement in LLM-based RS remain unstudied.

% What answers this paper
In this work, we focus on the measurement and statistical characterization of positional bias in top-1 and single-turn LLM-based RS, where the model recommends a single item from the retrieved candidates in one conversational turn. Rather than proposing mitigation strategies, our objective is to rigorously determine whether positional bias is present, how it can be detected, and how strong its effects are under controlled experimental conditions. To structure this investigation, we formalize the following research questions:

\begin{itemize}
\item[RQ1:] Does output selection in LLM-based RS depend on the position of candidate items within the input prompt?

\item[RQ2:]  Is any observed dependence between item position and model output statistically significant under permutation-based testing?

\item[RQ3:] How large is the positional bias, and does its magnitude vary across models?

\end{itemize}

% Introduction to POB-ICM methodology
To address these questions, we propose two complementary statistical methodologies for positional bias detection. The first, Positional Bias detection based on Independence of the Contingency Matrix (POB-ICM), models the relationship between input positions and output selections as a contingency table and evaluates whether these variables are statistically independent. Under the null hypothesis of no positional bias, the distribution of recommended items should be independent of their input positions. Violations of this assumption are assessed using a query-level permutation test of independence that allows both detection and interpretation of positional bias.

% Introduction to POB-UNBP methodology
The second methodology, Positional Bias detection based on Uniformity assuming Non-Biased Positions (POB-UNBP), is designed to reduce sample complexity by leveraging a stronger null assumption. Specifically, in the absence of positional bias, on expectation over all the possible queries, the probability of recommending the item at a given position is uniform. By repeatedly permuting the order of candidate items and analyzing the resulting recommendation distributions, this method enables efficient detection of systematic positional preferences while requiring fewer samples than full contingency-based testing. Together, these methodologies provide a principled and statistically grounded framework for evaluating positional bias in LLM-based RS.

% Our main contributions
Using these approaches, we conducted an empirical study on state-of-the-art LLMs embedded within a RAG-based RS for product recommendation, built on a real product catalog with template-generated user queries and a fixed prompt template. Our experiments reveal statistically significant positional bias, demonstrating that item order can systematically influence recommendation outcomes even when item relevance is kept constant. We complement these statistical tests with three proposed metrics. Permutation sensitivity metric evaluates how the model changes its decision on permutation changes, so it is an upper bound of positional bias. Positional bias metric measures the total variation distance between the uniform distribution and the distribution of positions of items selected by the RS, as deviation from the uniform distribution indicates a systematic preference for certain positions. Finally, the strong positional bias metric builds on relevance sets, the query-dependent sets of candidate items that the model treats as comparably relevant. Relevance sets decompose the overall positional bias detected by our methodologies into soft positional bias, where the position determines which item of the relevance set is recommended, and strong positional bias, where the position displaces the recommendation outside the relevance set, which this metric quantifies. In summary, the main contributions of this paper are as follows:

\begin{enumerate}
 \item We introduce POB-ICM, a statistically grounded methodology for detecting and interpreting overall positional bias in LLM-based RS based on contingency matrix independence.
 \item We propose POB-UNBP, a sample-efficient positional bias detection methodology based on uniformity assumptions under the absence of positional bias, which allows for stochastic decoding, and complement it with two quantification metrics (permutation sensitivity and positional bias) that disentangle generic permutation sensitivity from structural positional preference.
 \item We decompose overall positional bias into soft and strong positional bias through query-dependent relevance sets, quantify the strong component, and show through counterfactual analyses that it is not explained by the identifiers assigned to the items and, for six of the eight evaluated models, not by stochastic decoding alone. The formal detection of each component separately is left as future work.
 \item We implement and evaluate an end-to-end experimental framework to assess positional bias in LLM-based RS using a RAG architecture in a single-turn product recommendation setting, including a systematic robustness analysis across candidate-set size $k$ and decoding temperature $t$ on eight state-of-the-art LLMs, showing that positional bias is a structural, model-dependent property rather than an artifact of deterministic decoding or candidate-set size.
\end{enumerate}

% Results
The empirical results indicate that positional bias is a systematic and statistically detectable factor in LLM-based recommender systems implemented under an RAG pipeline, within the evaluated setting. Across permutation-based evaluations, both proposed methodologies (POB-ICM and POB-UNBP) reject the null hypothesis that recommendation outcomes are independent of the positional placement of candidate items within the prompt in all evaluated models, and the resulting selection distributions exhibit structured, non-uniform positional bias with certain prompt positions associated with reliably elevated selection probabilities.

The remainder of this paper is organized as follows. In~\autoref{sec:relatedwork} we review related work on positional bias and related phenomena in RS and language models. In \autoref{sec:methodologies} we describe the system design, formal modeling, and the proposed POB-ICM and POB-UNBP methodologies. The experimental setup and empirical results are presented in ~\autoref{sec:results}. Finally, ~\autoref{sec:conclusion_futurework} discusses the implications of our findings, limitations, and directions for future work.

\section{Related Work}
\label{sec:relatedwork}

% Introduction to LLM-based RS
Conversational RS generate personalized recommendations by interpreting user preferences expressed in natural language. Early conversational RS relied on rule-based or traditional machine learning approaches, while more recent systems increasingly incorporate deep learning models. With the advent of Large Language Models (LLMs), recommendation pipelines have undergone a substantial shift, as LLMs enable direct reasoning over textual user queries and item descriptions. Previous work has demonstrated the applicability of LLMs to a variety of recommendation tasks, including rating prediction, top-$N$ recommendation, and conversational recommendation scenarios~\cite{li-etal-2024-large}. These systems can generate recommendations as item identifiers or as natural-language responses, leveraging the linguistic and reasoning capabilities of LLMs.

% Introduction to RAG-based RS
A common instantiation of LLM-based RS adopts a retrieval-augmented generation (RAG) architecture. In this setting, a retriever selects a set of $K$ candidate items from a larger database based on semantic similarity to the user query, and these candidates are subsequently provided to the LLM for recommendation generation~\cite{wang_zero-shot_2023}. In the top-$N$ recommendation, the model selects or ranks multiple items (with $K > N$), while in the special case of the Top-1 recommendation, the LLM is required to select a single item from the candidate set. The latter setting is particularly relevant for conversational interfaces, where concise recommendations are often preferred, and is the setting studied in this work, in which a single item is recommended in a single conversational turn.

% Introduction to unfair and structural biases
Bias in RS has been extensively studied and is commonly categorized as \emph{unfair biases} and \emph{structural biases}. Unfair biases involve systematic disparities that affect sensitive user groups, such as those defined by gender, race, or age~\cite{gallegos_bias_2024, navigli2023biases}. In contrast, structural biases arise from the design of the system itself, including data preprocessing choices, model architecture, or prompt formulation. In LLM-based RS, structural biases can emerge from seemingly innocuous prompt design decisions, such as wording or synonym selection, which can influence the model output independently of user intent~\cite{gallegos_bias_2024}.

% Positional bias introduction and empirical evidence
Beyond recommender systems, structural bias in LLMs has been studied across several application domains. We define positional bias as a special case of structural bias in which a model's output depends systematically on the sequential position of an input element such as a candidate item, an answer option, or a retrieved passage, rather than solely on its semantic content. Because language models process text as ordered sequences, any system that presents multiple elements in a list implicitly encodes positional information that the model may exploit, independently of relevance or quality. In the context of \textit{multiple-choice question answering}, it has been shown that the order in which choices are listed can affect model answers~\cite{pezeshkpour-2024}. LLMs also exhibit selection bias by assigning more probability mass to specific option ID tokens regardless of content~\cite{zheng_large_2024}. Additional studies have demonstrated that even the symbols used to delimit answer options can introduce systematic biases in model output~\cite{yang-etal-2025-option}. In \textit{long-context question answering and retrieval-augmented generation}, \cite{liu-etal-2024-lost} demonstrated that LLM performance follows a U-shaped curve as a function of where relevant information appears in the context. They named this phenomenon as ``lost in the middle'' effect, with \cite{hsieh-etal-2024-found} tracing this degradation to intrinsic U-shaped attention patterns in transformer models independent of content relevance. In the domain of \textit{LLM-as-a-Judge evaluation}, \cite{zheng-etal-2023-judging} documented that strong LLMs used as evaluators exhibit positional bias when judging response quality, a finding confirmed at scale by \cite{shi-etal-2025-judging} across 15 judge models and over 150,000 evaluation instances. In \textit{LLM-based recommender systems}, \cite{hou_large_2024} observed that recommendation performance deteriorates when highly relevant items appear in later positions of the candidate list. These findings provide empirical indicators of positional bias in different areas. However, across all these domains, these characterize positional bias through performance degradation metrics and does not provide a hypothesis-testing framework that formally isolates position as a statistically independent variable.

% Justification on why positional bias in LLM-based RS
Top-1 RS based on LLM can be naturally framed as a multiple-choice decision problem: a set of $K$ candidate items is presented to the model, and the LLM must select the item that best satisfies the user's preferences. However, unlike standard multiple-choice tasks, recommendation scenarios often admit multiple plausible answers rather than a single correct option. This distinction complicates the direct transfer of positional bias analysis from question answering to recommendation settings, and similarly limits the applicability of long-context QA approaches such as the lost-in-the-middle framework~\cite{liu-etal-2024-lost}, which rely on a ground-truth answer to measure positional bias and do not isolate position as a statistically independent variable. Moreover, unlike ranking-based formulations of LLM-based recommendation~\cite{hou_large_2024}, where the model orders all the candidates, a top-1 recommender only reveals the item it selects, so the relevance it assigns to the remaining candidates must be inferred indirectly, which makes it harder to separate positional effects from item relevance. There are nevertheless strong reasons to hypothesize that positional bias affects LLM-based RS. Training data for LLMs include numerous instances of recommendations and ranked lists, which may themselves reflect user interaction biases, such as a tendency to favor items appearing earlier in lists~\cite{collins_position_2018}. Previous surveys have warned that such biases can propagate into learned recommendation behaviors~\cite{chen_bias_2023, wu2024survey}. Empirical observations further suggest that the performance of LLM-based RS deteriorates when highly relevant items appear in later positions within the candidate set~\cite{hou_large_2024}, though these studies do not isolate positional bias from confounders such as retrieval quality or item relevance.

% Limitations of current approaches and how we solve the gap
Despite these observations and intuition, the closest existing approaches to statistically grounded quantification of positional bias in LLMs rely on scalar summary metrics rather than formal hypothesis testing. \cite{zheng_large_2024} propose the standard deviation of per-option recall (RStd) as a measure of selection bias in MCQ tasks. \cite{liusie-etal-2024-teacher} define positional bias as the divergence between the marginalized selection distribution and the uniform distribution, implicitly assuming, as we formally prove in \autoref{sec:uniform_methodology}, that absence of positional bias implies uniformity. \cite{wei-etal-2024-unveiling} introduce the Fluctuation Rate, a label-free metric defined as the fraction of instances where a fixed forward ordering and a single fixed backward ordering yield different answers; because it compares exactly two predetermined permutations rather than the full permutation space, it may not capture the full extent of positional bias. Crucially, none of these approaches frames detection as a hypothesis test or provides statistical significance guarantees, and all are developed for MCQ tasks where a unique correct answer exists. Moreover, positional bias analysis developed for multiple-choice question answering does not directly transfer to recommendation settings, where multiple candidates may be plausibly relevant and a single ground-truth label may not exist. To address these gaps, we develop two complementary hypothesis-testing methodologies tailored to single-turn Top-1 LLM-based recommendation under controlled permutations: \textbf{POB-ICM}, which detects and interprets positional dependence via contingency-matrix independence testing, and \textbf{POB-UNBP}, which provides a sample-efficient alternative via a uniformity-based null model with formal statistical guarantees. Furthermore, since several candidates may be comparably relevant, we distinguish soft positional bias, where the position only determines which of these candidates is recommended, from strong positional bias, where the position displaces the recommendation towards clearly less relevant items, and, besides the overall positional bias, we quantify its strong component. Together, these methods enable principled detection and quantification of positional bias independently of item quality and retrieval performance.

\section{Methodology for positional bias detection in LLM-based RS}
\label{sec:methodologies}

LLM-based RS expose candidate items to the model as an ordered list within a text prompt. Let $\{p(1), p(2), \ldots, p(k)\}$ denote the marginal probability distribution over prompt positions, where $p(i)$ is the probability that the item appearing at position $i$ is recommended, averaged over queries and random permutations of the candidate set (formally defined in~\autoref{sec:mathematical_modeling}). In the absence of positional bias, a model should recommend each item based solely on its relevance to the user query, regardless of its position in the prompt. As we formally prove in~\autoref{sec:uniform_methodology}, this condition implies that $p(i)$ follows a uniform distribution: when items are presented in random order, no prompt position is systematically preferred, and thus $p(1) = p(2) = \cdots = p(k) = 1/k$.

To detect violations of this property, we propose two complementary statistical methodologies, POB-ICM and POB-UNBP, which trade off interpretability against sample efficiency. We first describe the RAG-based RS testbed used as our experimental framework in~\autoref{sec:recommender_implementation}, establish the formal probabilistic notation in~\autoref{sec:mathematical_modeling}, and then introduce each methodology in~\autoref{sec:independency_methodology} and~\autoref{sec:uniform_methodology} respectively.

\subsection{An implementation of an LLM-based RS for positional bias analysis}
\label{sec:recommender_implementation}

To empirically study positional bias under controlled conditions, we implement an LLM-based RS following a RAG architecture. This implementation serves as a testbed that allows systematic manipulation of item order while keeping all other components fixed.
\begin{figure}
    \centering
    \includegraphics[width=0.5\linewidth]{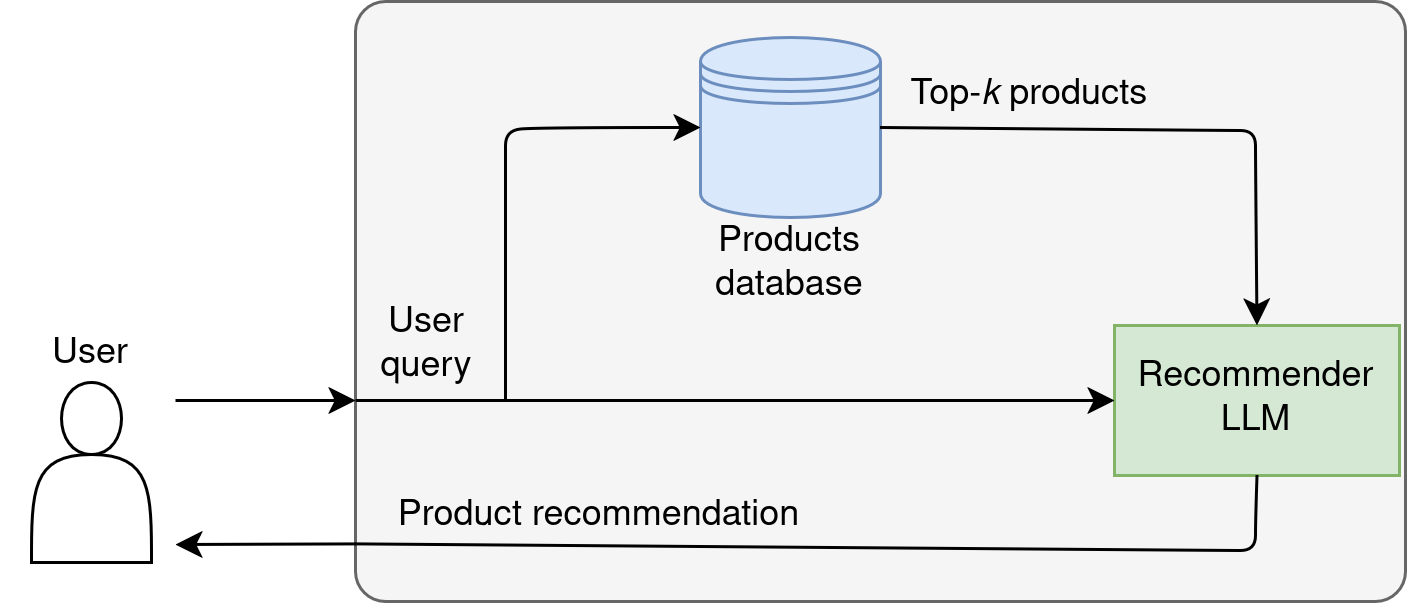}
    \caption{LLM-based RS architecture used in this paper.}
    \label{fig:usecase}
\end{figure}

% Description of our use case
The implemented RS based on LLM, illustrated in \autoref{fig:usecase}, follows an RAG architecture. It takes as input a user query and is used to semantically match several products from a database, sorting them using the distance between the query and the product titles. Then, this query, the titles and the descriptions of the products are fed to an LLM used to recommend one of these products. As a result, the RS generates a recommendation selecting the position of one of the related products. The specific steps are described below.

\begin{enumerate}
  \item An embedding is computed using a pre-trained model, namely, the Beijing Academy of Artificial Intelligence (BAAI) general embedding model~\cite{xiao_c-pack_2024}.
  \item The cosine distance is computed along each precomputed embedding from a product database, specifically from the products' titles. These products are sorted using the cosine distance, as it represents the semantic similarity of the products with the user query.
  \item The top-$k$ products are selected, where $k$ is an hyperparameter of the system. 
  \item The selected products (titles and descriptions) and the user query are used to construct a prompt, following the template illustrated in~\autoref{fig:recommender_prompt}. This prompt asks the model to generate a recommendation given the products and the user query. A JSON format is required to automatically parse the answer. The response is a recommendation for a specific product, and the chatbot generates a conversation for the user.
\end{enumerate}

\begin{llmprompt}[\linewidth]{Recommender system template prompt}{fig:recommender_prompt}
You are a shop assistant that recommends the user the given shop articles based on his/her preferences.\newline

Here are some articles in my shop:\newline

Title 1: \{title\}\newline
Description 1: \{description\}\newline

...\newline

Title \textit{k}: \{title\}\newline
Description \textit{k}: \{description\}\newline

The user query is \{query\}. Based on the user query and the provided articles, recommend one of these articles in a chat way. Think step by step why you choose your selected article.\newline

At the end, please provide your answer using a JSON format with the fields "article\_number" (int), "article\_title" (str) and "recommendation" (str). Provide a single JSON as described.
\end{llmprompt}

The above implementation defines an end-to-end mapping from a user query and a set of retrieved items to a single recommended item. To enable a rigorous statistical analysis of positional bias, we next formalize this pipeline as a function that maps a query instance and an item permutation to a selected output position. 
This formalization allows us to define explicit null hypotheses and apply statistical tests that isolate the effect of item position from other components of the system.

\subsection{Mathematical notation and probability distributions of the LLM-based RS}
\label{sec:mathematical_modeling}

Let \(S_k\) be the set of all permutations of $\{1,2,\ldots,k\}$. Let $\mathcal{D}$ be a probability distribution over possible user queries for the recommendation system and the associated top-$k$ products, that is, we sample from $\mathcal{D}$ tuples of the form $(Q,P_1, P_2, \ldots, P_k)$, where $Q$ is a user query and $P_1, P_2, \ldots, P_k$ are the top-$k$ products selected within the RAG architecture for query $Q$. For an instance $X \sim \mathcal{D}$ and a permutation $\sigma \in S_k$, we use~\autoref{fig:recommender_prompt} with the items listed in the order $P_{\sigma(1)}, P_{\sigma(2)}, \dots, P_{\sigma(k)}$, obtaining a recommended item by the LLM. We denote by $R(X, \sigma)$ the position (in the prompt) of the item recommended. For example, if $R(X, \sigma) = 1$, this means that the LLM recommends the first item in the prompt, which corresponds to $P_{\sigma(1)}$. That is, in the original list $P_1, P_2, \ldots, P_k$ of top-$k$ products, the recommended item actually has index $\sigma(1)$. 

In our analysis we sample $\sigma$ from the uniform distribution over $S_k$, denoted $U(S_k)$. The use of random permutations enables the evaluation of positional bias on the recommendation $R$, consistent with the principles of permutation testing~\cite{good2013permutation}. For example, the probability that the LLM will recommend the first item to appear in the prompt is
\begin{equation*}
%p_1 = \frac{1}{|\mathcal{D}| k!}  \sum_{X \in \mathcal{D}} \sum_{\sigma \in S_k} \mathbf{1}[R(X,\sigma)=1],
p(i) = \operatorname{Pr}_{X \sim \mathcal{D},\sigma \sim U(S_k)}\left[R(X, \sigma) = i\right] = \mathbb{E}_{X \sim \mathcal{D},\sigma \sim U(S_k)}\mathbf{1}\left[R(X, \sigma) = i\right],
\end{equation*}
where $\mathbf{1}(A)$ is $1$ if the event $A$ holds, and $\mathbf{1}(A) = 0$ otherwise. Generally, the probability that the element at position $i$ in the prompt is the one recommended in our experiment is
\begin{equation}
p(i) = \mathbb{E}_{X \sim \mathcal{D},\sigma \sim U(S_k)}\mathbf{1}\left[R(X, \sigma) = i\right].
\label{eq:probability_distribution}
\end{equation}
Obviously, we have $p(1)+p(2)+\cdots+p(k) = 1$. If the LLM-based RS does not suffer from positional bias, we would have $p(1) = p(2) = \cdots = p(k) = 1/k$, which will be exploited  in one of our proposed methodologies.

We design two complementary methodologies to test the null hypothesis of positional independence. 
The first methodology, POB-ICM, is described in \autoref{sec:independency_methodology} and evaluates dependence through a contingency matrix and statistical independence tests. 
Although robust and interpretable, this approach requires a large number of samples. 
To address this limitation, we propose a second methodology, POB-UNBP, described in ~\autoref{sec:uniform_methodology}, which tests positional bias through uniformity assumptions and requires substantially fewer samples, at the cost of reduced interpretability.

\subsection{POB-ICM: POsitional Bias detection methodology based on the Independency of the Contingency Matrix}\label{sec:independency_methodology}

% Definition of the method
In the absence of positional bias, applying a permutation to the top-$k$ items in the LLM query should not statistically affect the final recommendation. Concretely, the choice of $\sigma \sim U(S_k)$ should be independent of the final recommendation. Building on this observation, we present a methodology that employs the Chi-square test for independence~\cite{agresti2013categorical}, whose $p$-value is calibrated with a query-level permutation test to account for the dependence among the permutations evaluated on the same query. In addition to showing statistical significance, we assess the strength of the detected association using three different metrics described in \autoref{sec:metrics}. Implementing this test requires building a contingency matrix for the following two random variables: the selected product and the random shuffling applied. First, we formalize these random variables. 

Recall that in our experiment we start with an instance $X= (Q, P_1, P_2, \ldots, P_k) \sim \mathcal{D}$, choose $\sigma \in S_k$ uniformly at random and determine $R(X, \sigma) = j$, that is, the recommended product is $P_{\sigma(j)}$. We define the random variable ``original index'', denoted $Y$, as $\sigma(j)$. Note that $Y$ is actually a function of $X$ and $\sigma$, so we can write $Y = f(X, \sigma)$ for $X \sim \mathcal{D}$ and $\sigma \sim U(S_k)$. Moreover, $Y$ yields an integer between $1$ and $k$. Lack of positional bias implies that the random variables $Y$ and $\sigma$ are conditionally independent given $X$, that is, $\operatorname{Pr}[Y = i \,\vert\, \sigma, X] = \operatorname{Pr}[Y= i \,\vert\, X]$ for all $i \in \{1,2,\ldots,k\}$ and $\sigma \in S_k$.

We build the contingency table of $Y$ and $\sigma$ as follows. For ease of notation, we assume that we have indexed all the permutations in $S_k$, obtaining a list $\sigma_1, \sigma_2, \ldots, \sigma_{k!}$. The $i$-th row in this table correspond to $Y  = i$, that is, the case when the recommended product has position $i$ in the original order (before shuffling). The $j$-th column corresponds to the permutation $\sigma_j$. Let $X_1, X_2, \ldots, X_t$ be samples from the distribution $\mathcal{D}$. For $i \in \{1,2,\ldots,k\}$ and $\sigma_j \in S_k$, the $i,j$ entry of this contingency matrix is formally defined as
\begin{equation}\label{eq:contingency_matrix}
    C_{i,j} = \sum_{\ell = 1}^t
    \mathbf{1} \left[
    f(X_\ell, \sigma_j) = i
    \right] = \sum_{\ell = 1}^t
    \mathbf{1} \left[
    \sigma_j(R(X_\ell, \sigma_j)) = i
    \right].
\end{equation}
That is, $C_{i,j}$ counts how many instances $(Q, P_1, P_2, \ldots, P_k)$ of the dataset sampled from $\mathcal{D}$ satisfy that the recommended item, when the prompt presents the items following the order $\sigma_j$, is the one in position $i$ in the original order $P_1, P_2, \ldots, P_k$. Note that if there is no positional bias, then all columns of the matrix $C$ should be the same since the item recommended is the same independently of the permutation $\sigma$, that is, $\sigma(R(X, \sigma)) = \varphi(R(X, \varphi))$ for any two permutations $\varphi, \sigma \in S_k$.

The null-hypothesis is that the variables $Y$ and $\sigma$ are conditionally independent given $X$. We compute the Pearson Chi-square statistic $\chi^2_{\mathrm{obs}}$ of the contingency matrix $C$. Since $C$ aggregates several permutations of the same instance $X_\ell$, its entries are not independent observations and the asymptotic Chi-square distribution may yield miscalibrated $p$-values. Hence, we calibrate the $p$-value with a query-level (stratified) permutation test. Under the null-hypothesis, the recommended item $Y$ does not depend on $\sigma$, so, for each instance $X_\ell$, the values of $Y$ are exchangeable among the permutations applied to $X_\ell$. We generate $B$ contingency matrices $C^{(1)}, \ldots, C^{(B)}$ by randomly shuffling the values of $Y$ among the permutations of each instance, never mixing values across different instances, and compute their Chi-square statistics $\chi^2_{(1)}, \ldots, \chi^2_{(B)}$. This empirical distribution preserves both the query-level dependence and the marginals of $C$. The Monte Carlo $p$-value is $p = \big(1 + \sum_{b=1}^{B} \mathbf{1}[\chi^2_{(b)} \geq \chi^2_{\mathrm{obs}}]\big)/(1 + B)$, with $B = 2000$ in our experiments. If the null-hypothesis is rejected, there is a dependency between the product selection and the position of the products, indicating a positional bias.

The generation of the contingency matrix described in~\autoref{eq:contingency_matrix} requires the generation of $k!$ permutations for each of the $t$ samples from the distribution~$\mathcal{D}$. In practice, a subset of permutations may be used, since the permutation test above does not rely on the asymptotic Chi-square approximation and thus does not require a minimum number of counts per cell. We can display the contingency matrix via a heatmap, which allows the user to inspect the positional bias induced by the LLM, see \autoref{fig:contingency_matrix}. However, this matrix might be overwhelming due to its size. To simplify this visualization, we also use a matrix where each row represents the position of a given product \textit{before} the shuffling, and the columns represent the position of this product \textit{after} the shuffling. Then each cell $i,j$ averages the number of times the product $i$ was selected when $\sigma$ placed it at position $j$, that is,
\begin{equation} \label{eq:bias_analysis}
   A_{i,j} = \frac{1}{c_{i,j}}
             \sum_{X \in \mathcal{D}}
             \sum_{\sigma \in \Sigma_X : \sigma(j) = i}
             \mathbf{1}[R(X, \sigma) = j],
\end{equation}
where $\Sigma_X$ is the set of permutations evaluated for instance $X$, and $c_{i,j}$ is the average number of permutations per instance in $\Sigma_X$ that place the product $i$ at position $j$. If all possible permutations are evaluated, $\Sigma_X = S_k$ and $c_{i,j} = (k-1)!$.

We could interpret the matrix $A_{i,j}$ by columns. Each column represents the number of times that a recommender LLM selects a product in the position $j$ after moving the products to this position. If a column is more likely to be selected, this would mean that there is a positional bias with a positive tendency. In the same way, if a column is less likely to be selected, then there is a negative positional bias towards this position.

\subsection{POB-UNBP: Positional Bias detection methodology based on Uniformity assuming Non-Biased Positions}\label{sec:uniform_methodology}

The positional bias evaluation methodology described in the previous section requires many samples in order to be robust. Thus, its evaluation on closed-weights models could be too expensive. In order to alleviate this, a second methodology is proposed based on the uniformity of the probability distribution generated by the model. In this section we work with the probability distribution $p(i)$ given in \autoref{eq:probability_distribution}. 

First, we mathematically prove that, under the assumption of absence of positional bias, $p(i)$ follows a uniform distribution. Secondly, we generalize this result to the stochastic setting, where the RS output is understood as a probability distribution on the set of items in the prompt. Here  we allow for a small tolerance of positional bias and show that, despite this tolerance and the stochastic setting, $p(i)$ is very close to the uniform distribution. As a consequence, positional bias can be statistically tested by empirically sampling from $p(i)$ and checking if sampled data comes from a probability distribution close to the uniform distribution.
 
\subsubsection{Zero-temperature setting}
 
We note that in the absence of positional bias, given an instance $X = (Q, P_1, P_2, \ldots, P_k)$, for any two permutations $\sigma, \varphi \in S_k$, the recommended item is the same. Recall that we denote by $R(X, \sigma)$ the position of the recommended item within the prompt, and that this prompt contains the items in the order $P_{\sigma(1)}, P_{\sigma(2)}, \ldots, P_{\sigma(k)}$. Hence, the recommended item is $P_{\sigma(R(X, \sigma))}$, and, by assuming absence of positional bias, we have $P_{\sigma(R(X, \sigma))} = P_{\varphi(R(X, \varphi))}$. That is, $\sigma(R(X, \sigma)) = \varphi(R(X, \varphi))$ for all permutations $\sigma, \varphi \in S_k$. Equivalently, there exists a map $g$ from instances to $\{1, 2, \ldots, k\}$ such that $\sigma(R(X, \sigma)) = g(X)$ for all $\sigma \in S_k$.
 
\begin{proposition}\label{prop:deterministic}
    In the absence of positional bias, $p(i \mid X = x) = 1/k$ for all $i \in \{1, 2, \ldots, k\}$ and all instances $x$.
\end{proposition}
\begin{proof}
    Fix an instance $X = x$. Since $\sigma(R(x, \sigma)) = g(x)$ for all $\sigma \in S_k$, we have $R(x, \sigma) = \sigma^{-1}(g(x))$. Therefore,
    \begin{equation*}
    \begin{split}
        p(i \mid X = x) &= \mathbb{E}_{\sigma \sim U(S_k)} \mathbf{1}[R(x, \sigma) = i] \\
        &= \mathbb{E}_{\sigma \sim U(S_k)} \mathbf{1}[\sigma^{-1}(g(x)) = i] \\
        &= \mathbb{E}_{\sigma \sim U(S_k)} \mathbf{1}[\sigma(i) = g(x)] \\
        &= \frac{1}{k},
    \end{split}
    \end{equation*}
    where the last equality holds because, for $\sigma$ drawn uniformly from $S_k$, the value $\sigma(i)$ is uniformly distributed over $\{1, 2, \ldots, k\}$.
\end{proof}

 % Description of statistical test and power (Chi-square goodness-of-fit test)
Once we have determined that the probability distribution $p$ defined in~\autoref{eq:probability_distribution} should be close to a uniform distribution under the absence of positional bias, we can evaluate positional bias by testing whether $p$ deviates from the uniform distribution. However, the notion of absence of positional bias used in \autoref{prop:deterministic} assumes that, for each instance, a single item is recommended regardless of the permutation applied. This excludes both decoding at non-zero temperature, where the recommended item may vary even for a fixed prompt, and the case in which several candidates are comparably relevant and the model may spread its recommendations among them. Hence, we describe the statistical test in the stochastic setting below, where absence of positional bias is defined on the distribution of the recommended item, which may be spread over several items as long as it does not depend on the ordering, and is further relaxed to tolerate small positional effects. This setting includes the deterministic one as a particular case. This test requires fewer samples compared to the methodology proposed in~\autoref{sec:independency_methodology}, as shown in \autoref{sec:simulation_study}.

\subsubsection{Stochastic setting}

In certain settings, the LLM may operate with non-zero temperature, so items recommended will vary each time we query the same recommendation problem. In such a setting, the LLM recommends item $P_j$ with some probability $\mathrm{Pr}(P_j \mid X, \sigma)$ that depends on the query and items $X = (Q, P_1, P_2, \ldots, P_k)$ and the ordering $\sigma$ of the items in the prompt. Ideally, in the absence of positional bias, this probability distribution over items should not depend on $\sigma$, that is, we should have $\mathrm{Pr}(P_j \mid X, \sigma) = \mathrm{Pr}(P_j \mid X, \varphi)$ for any permutations $\sigma, \varphi \in S_k$. However, in practice, such definition seems difficult to fully attain by current LLM architectures. Therefore, we introduce the concept of $\varepsilon$-absence of positional bias, which only requires the distributions  $\mathrm{Pr}(\cdot \mid X, \sigma)$ and $\mathrm{Pr}(\cdot \mid X, \varphi)$ to be reasonably similar.

\begin{definition}\label{def:eps-absence}
    We say that there is $\varepsilon$-absence of positional bias if, for all instances $X = x$ and all permutations $\sigma, \varphi \in S_k$,
    \begin{equation*}
        \mathrm{TV}\!\left(\mathrm{Pr}(\cdot \mid x, \sigma),\; \mathrm{Pr}(\cdot \mid x, \varphi)\right) \leq \varepsilon,
    \end{equation*}
    where $\mathrm{TV}(\pi, \mu) = \frac{1}{2} \sum_{j=1}^k |\pi(j) - \mu(j)|$ denotes the total variation distance.
\end{definition}

\begin{theorem}\label{thm:stochastic}
    Under $\varepsilon$-absence of positional bias, we have $\mathrm{TV}(p(\cdot), U_k) \le \varepsilon$, where $U_k$ denotes the uniform distribution over $\{1, 2, \ldots, k\}$. In fact, for all instances $x$, $\mathrm{TV}\!\left(p(\cdot \mid x),\; U_k\right) \leq \varepsilon$.
\end{theorem}

The proof of \autoref{thm:stochastic} is given in \autoref{app:proof_theorem}.
\begin{remark}
    \autoref{prop:deterministic} is a particular case of  \autoref{thm:stochastic} with $\varepsilon = 0$. Indeed, $\mathrm{TV}(p(\cdot \mid x), U_k) = 0$ implies $p(i \mid x) = 1/k$ for all $i$. Moreover, here we still allow for stochastic answers of the RS.
\end{remark}

\autoref{thm:stochastic} establishes that $\varepsilon$-absence of positional bias implies $\mathrm{TV}(p(\cdot), U_k) \le \varepsilon$. The null hypothesis naturally associated with $\varepsilon$-absence is therefore the interval hypothesis
\begin{equation}\label{eq:eps_null}
    H_0^{\varepsilon}: \mathrm{TV}(p(\cdot), U_k) \le \varepsilon
    \qquad \text{versus} \qquad
    H_1^{\varepsilon}: \mathrm{TV}(p(\cdot), U_k) > \varepsilon,
\end{equation}
and, by contraposition, rejecting $H_0^{\varepsilon}$ provides evidence against $\varepsilon$-absence of positional bias. Testing exact uniformity of $p$ corresponds to the special case $\varepsilon = 0$. Since exact absence of positional bias is rarely attainable in practice, and a test of exact uniformity may flag arbitrarily small deviations when the sample is large, we now describe a test of $H_0^{\varepsilon}$ for arbitrary $\varepsilon \ge 0$.

The test relies on the following characterization of the total variation distance to the uniform distribution. Writing $x_+ = \max(x, 0)$, we have
\begin{equation}\label{eq:tv_positive_part}
    \mathrm{TV}(p(\cdot), U_k)
    = \max_{A \subseteq \{1, 2, \ldots, k\}} \left( p(A) - \frac{|A|}{k} \right)
    = \sum_{i=1}^{k} \left( p(i) - \frac{1}{k} \right)_{\!+},
\end{equation}
where $p(A) = \sum_{i \in A} p(i)$, and the maximum is attained at $A^{*} = \{ i : p(i) > 1/k \}$. Writing $\mu_i = p(i) - 1/k$ for each prompt position $i$, it therefore suffices to obtain simultaneous lower confidence bounds for the $k$ scalar quantities $\mu_1, \ldots, \mu_k$.

Let $X_1, \ldots, X_n$ be instances sampled independently from our data distribution $\mathcal{D}$ (see \autoref{sec:mathematical_modeling}). For each $X_j$ we sample $\ell$ permutations $\sigma_{j,1}, \ldots, \sigma_{j,\ell}$ uniformly and independently from $S_k$, and determine $R(X_j, \sigma_{j,t})$ for all $t$. We define
\begin{equation*}
    \hat{q}_j(i) = \frac{1}{\ell} \left\lvert \left\{ t \in \{1, 2, \ldots, \ell\} : R(X_j, \sigma_{j,t}) = i \right\} \right\rvert,
\end{equation*}
that is, $\hat{q}_j(i)$ is the fraction of evaluated permutations for instance $X_j$ in which the LLM selected the item at position $i$ in the prompt.\footnote{When some generations are discarded due to parsing errors or refusals (see \autoref{sec:experimental_setup}), $\hat{q}_j(i)$ is computed over the valid generations of $X_j$, and the analysis below is understood conditionally on their number.} For each position $i \in \{1, 2, \ldots, k\}$, we define the instance-level statistics
\begin{equation*}
    Z_j(i) = \hat{q}_j(i) - \frac{1}{k}, \qquad j = 1, 2, \ldots, n.
\end{equation*}
Since the instances are sampled independently and the permutations are drawn independently of each other and of the instances, the variables $Z_1(i), \ldots, Z_n(i)$ are independent and identically distributed. Moreover, $\mathbb{E}[\hat{q}_j(i)] = p(i)$ by \autoref{eq:probability_distribution}, so their common mean is $\mathbb{E}[Z_j(i)] = \mu_i = p(i) -1/k$.

Fix a significance level $\alpha \in (0,1)$; in our experiments $\alpha = 0.05$. Let $\bar{Z}(i)$ and $s_i$ denote the sample mean and standard deviation of $Z_1(i), \ldots, Z_n(i)$. For each position $i$ we define the one-sided lower confidence bound
\begin{equation}\label{eq:position_lower_bound}
    L_i = \bar{Z}(i) - t_{1-\alpha/k,\, n-1} \, \frac{s_i}{\sqrt{n}},
\end{equation}
where $t_{1-\alpha/k,\,n-1}$ denotes the $(1-\alpha/k)$-quantile of the Student's $t$ distribution with $n-1$ degrees of freedom, and we aggregate these bounds into
\begin{equation}\label{eq:tv_lower_bound}
    L = \sum_{i=1}^{k} \left( L_i \right)_{\!+}.
\end{equation}
The quantities $s_i$, $L_i$ and $L$ depend on $n$, which we omit from the notation.

\begin{proposition}\label{prop:lower_bound}
    Fix $k \ge 2$, $\ell \ge 1$ and $\alpha \in (0,1)$, and let $L$ be the bound defined in~\autoref{eq:tv_lower_bound}, computed from $n$ independent instances. Then
    \begin{equation*}
        \liminf_{n \to \infty} \; \Pr\!\left[\, L \le \mathrm{TV}(p(\cdot), U_k) \,\right] \;\ge\; 1 - \alpha .
    \end{equation*}
\end{proposition}

The proof of \autoref{prop:lower_bound} is given in \autoref{app:proof_theorem}.

\begin{corollary}\label{cor:test_level}
    Fix $k \ge 2$, $\ell \ge 1$ and $\alpha \in (0,1)$, and let $L$ be the bound defined in~\autoref{eq:tv_lower_bound}, computed from $n$ independent instances. For every $\varepsilon \ge 0$, the test that rejects $H_0^{\varepsilon}$ whenever $L > \varepsilon$ has asymptotic level $\alpha$, that is, $\limsup_{n \to \infty} \Pr[L > \varepsilon] \le \alpha$ whenever $\mathrm{TV}(p(\cdot), U_k) \le \varepsilon$.
\end{corollary}
\begin{proof}
    Under $H_0^{\varepsilon}$, the event $L > \varepsilon$ implies $L > \mathrm{TV}(p(\cdot), U_k)$, whose probability has $\limsup$ at most $\alpha$ by \autoref{prop:lower_bound}.
\end{proof}

By \autoref{cor:test_level}, $\varepsilon$-absence of positional bias is rejected at level $\alpha$ for every $\varepsilon < L$, so that no tolerance needs to be fixed in advance. Only $k$ statistics are computed to determine $L$, rather than the $2^k$ subsets appearing in~\autoref{eq:tv_positive_part}, so the procedure is computationally efficient for every candidate-set size considered in this work. Beyond the aggregate bound $L$, the individual bounds $L_i$ are informative in their own right: a position $i$ with $L_i > 0$ is over-selected with respect to the uniform baseline, at simultaneous confidence level $1-\alpha$. We report $L$ as a percentage in \autoref{sec:results_analysis}.

%The Bonferroni correction in~\autoref{eq:position_lower_bound} is conservative, as it disregards the dependence between the $k$ statistics. A less conservative alternative is to replace $t_{1-\alpha/k,\,n-1}$ by the $(1-\alpha)$-quantile of the bootstrap distribution of $\max_{i} \sqrt{n}\,(\bar{Z}^{*}(i) - \bar{Z}(i)) / s^{*}_i$, where $\bar{Z}^{*}(i)$ and $s^{*}_i$ are computed on a resample of the $n$ instances drawn with replacement (max-$t$ bootstrap~\cite{westfall1993resampling, romano2005exact}). In our setting, however, the two choices yield nearly identical critical values, since the statistics $Z_j(1), \ldots, Z_j(k)$ are only weakly correlated. For instance, for $n = 2000$ and $\alpha = 0.05$, the Bonferroni critical values are $2.33$ for $k = 5$ and $2.81$ for $k = 20$, while the bootstrap yields [2.34] and [2.73], respectively. We therefore report results based on~\autoref{eq:position_lower_bound}.

The construction above is valid for any number of permutations per instance $\ell \ge 1$, including $\ell = 1$, since $\mathbb{E}[\hat{q}_j(i)] = p(i)$ for every $\ell$, and hence $\bar{Z}(i)$ is always an unbiased estimator of $\mu_i$. The choice of $\ell$ affects only the variance of the instance-level statistics, which decomposes as
\begin{equation*}
    \operatorname{Var}(Z_j(i)) = \operatorname{Var}_X\!\left( p(i \mid X) \right)
    + \frac{1}{\ell} \, \mathbb{E}_X\!\left[ \operatorname{Var}_{\sigma}\!\left( \mathbf{1}\left[ R(X, \sigma) = i \right] \mid X \right) \right].
\end{equation*}
Only the second, within-instance term decreases with $\ell$, while the first, between-instance term is irreducible. Consequently, increasing $\ell$ yields diminishing returns, and beyond a moderate value the power of the test is governed by the number of instances $n$. In the simulation study of \autoref{sec:simulation_study}, the between-instance term vanishes, since $p(i \mid X)$ does not depend on $X$ in the simulated recommender, and the required number of instances decreases roughly as $1/\ell$ (\Cref{tab:min_samples}); in real data, where the relevance of the candidates varies across queries, the gains from increasing $\ell$ are expected to saturate earlier. In our experiments we use $\ell = 20$, matching the query-permutation dataset described in \autoref{sec:datasets}.

\section{Empirical results and analysis of positional bias}
\label{sec:results}

This section presents the empirical evaluation and statistical analysis conducted to assess positional bias in LLM-based RS. The experiments are designed to test the null hypotheses of positional independence and uniformity formalized in ~\autoref{sec:methodologies}, using the POB-ICM and POB-UNBP detection methodologies introduced in ~\autoref{sec:independency_methodology} and~\autoref{sec:uniform_methodology}, respectively.

We first describe the construction of the evaluation datasets in ~\autoref{sec:datasets}, including the generation of user queries, candidate item sets, and their permutations. Secondly, ~\autoref{sec:experimental_setup} details the experimental setup and implementation choices, which strictly follow the RS design introduced in ~\autoref{sec:recommender_implementation}. Finally, ~\autoref{sec:results_analysis} reports the empirical results and analyzes the observed positional bias in the different recommended LLMs, highlighting both statistical significance and qualitative behavioral patterns.

\subsection{Evaluation dataset} \label{sec:datasets}

The dataset used in this study is derived by integrating and transforming domain-specific data sources. The main external source is the \textit{Amazon ML Challenge 2023} dataset. This dataset, in conjunction with user queries automatically generated, forms the basis of the final RAG-query dataset.

\textit{User-query dataset}. This dataset contains $2000$ different user queries for RS. Each query asks for a specific type of product, stating the user preferences. See \autoref{fig:user_queries_examples} for examples of such queries. The user-query dataset can be generated using our source code\footnote{\url{https://github.com/GoComputing/Positional-Bias.git}}. To construct this dataset, we start with $56$ distinct query templates, which ask the RS for a generic recommendation, using different wording and user preferences in each. Each template has placeholders for a product, up to three product attributes (such as quality, brand, feature, style, color, size, material, or condition), and optionally a use case or feature. We consider $194$ distinct products names, which can be grouped in $18$ categories: electronics, fashion, home appliances, beauty, sports, automotive, books, toys, groceries, office supplies, pet supplies, fitness, gaming, travel, outdoor equipment, health, baby products, and jewelry. Our user-query dataset has been constructed by sampling a category, a product within the category, a query template, and compatible attributes uniformly at random (discarding the repetitions), ensuring that each selected attribute is compatible with the chosen product category. Of the $69{,}071{,}332$ total possible queries, we sampled $2000$ in order to minimize the running time of the experiments while having a large enough dataset to achieve statistical significance in our results.

\begin{figure}
    \begin{tcolorbox}[colback=gray!5!white, colframe=gray, sharp corners, title={Examples of queries}]
    \centering
    \begin{itemize}
        \item Show me \underline{reputable brand} \underline{fitness trackers}.
        \item What's the best \underline{modern} \underline{portable speakers} \underline{for students}?
        \item \underline{humidifiers} that is \underline{large} and \underline{new}.
        \item What's a good \underline{compact} \underline{suspension trainers} \underline{for gym}?.
    \end{itemize}
\end{tcolorbox}
    \caption{Several user queries examples included in our dataset. These are generated from templates, where the underlined words indicate a placeholder.}
    \label{fig:user_queries_examples}
\end{figure}

\textit{Products dataset}. This dataset contains $50{,}000$ distinct products, each with a name/title and a description, which we have sampled from the \textit{Amazon ML Challenge 2023 dataset}\footnote{Available at \url{https://www.kaggle.com/datasets/ashisparida/amazon-ml-challenge-2023}}. Examples of such products are shown in \autoref{fig:products_examples}. Before sampling, we have discarded all products with description shorter than $40$ characters or empty values in the product type ID. The full Amazon ML Challenge 2023 dataset contains over $2.9$ million products spanning $9{,}919$ distinct browse node categories. Its domain coverage is broad, encompassing a wide variety of product types such as electronics, clothing, home appliances, food, movies, dietary supplements, and personal care products, among others. While our evaluation is therefore restricted to the product recommendation domain, the diversity of product types within this domain lends robustness to our findings and reduces the risk of conclusions being specific to a narrow product niche.

\textit{RAG-query dataset}. This dataset contains $2000$ different queries. Each one of them consists of a user query and $5$ products from the product dataset. We have constructed this data set running a RAG scheme on the user-query data set and the product dataset. That is, for each user query, the top $5$ products ($K=5$) are selected from the product dataset, as described in~\autoref{sec:recommender_implementation}. Thus, we obtain a tuple ($Q$, $P_1$, $P_2$, $P_3$, $P_4$, $P_5$), where $Q$ is the user query and $P_1, \ldots, P_5$ are the top-$5$ products. For the robustness analysis across candidate-set sizes (\autoref{sec:results_unbp}), analogous datasets are built from the same user queries with the top-$k$ products, for $k \in \{10, 15, 20\}$. In an LLM-based RS, this RAG-query is sent to the LLM (by filling in a prompt template with the query and the $5$ products), to get a recommendation. Our hypothesis is that LLMs exhibit a systematic preference for certain positions within the candidate list, independently of the relevance of the products.

\begin{figure}
    \begin{tcolorbox}[colback=gray!5!white, colframe=gray, sharp corners, title={Products titles associated with "Find me an eco-friendly protein snack." query}]
    \centering
    \begin{itemize}
        \item Green Sense Organic Multigrain Flour | Multigrain Atta | Gluten Free | Helps in Weight Loss, High in Protein - 500gm x 1 | Pack of 1
        \item Go Vegan California Roasted Pistachios, 400 Gram | Pista, Lightly Salted - (Pack of 2 -200g Each)
        \item Go Vegan Stone Candy - 400g, Rock Candy, Milk Stone Chocolate [Jar Pack]
        \item Nutritional Designs Sugar-Free Liquid Collagen \& Whey Protein Supplement - Non-GMO Drink - Promotes Healthy Skin \& Hair for Men \& Women (Neutral)
        \item Dev.Pro. Date and Nuts Drops, Healthy Snack with High Fiber, Gluten Free, Vegan - Cocoa Coconut 40g, (Pack of 5)
    \end{itemize}
\end{tcolorbox}
    \caption{Titles of the selected products for the query "Find me an eco-friendly protein snack". These are sorted by the cosine similarity to the query using BAAI embedding model.}
    \label{fig:products_examples}
\end{figure}

\textit{Query-permutation dataset}. This dataset contains, for each tuple ($Q$, $P_1$, $P_2$, $P_3$, $P_4$, $P_5$) in the RAG-query dataset, $20$ independently and uniformly sampled permutations of the $5$ candidate products. Thus, this dataset contains $2000 \cdot 20 = 40{,}000$ instances. In our experiments, each LLM analyzed is evaluated on this dataset as explained in~\autoref{sec:experimental_setup}. The resulting query-permutation dataset yields sufficient sample size to support the tests employed in ~\autoref{sec:results_analysis}, whose $p$-values are calibrated with a query-level permutation test and thus do not rely on asymptotic expected-count conditions (\autoref{sec:independency_methodology}). The same procedure is applied to the datasets with $k \in \{10, 15, 20\}$, each also yielding $40{,}000$ instances. Unless stated otherwise, all results are based on these sampled configurations. In addition, to better illustrate the relevance signal of the retriever, for $k = 5$ we also use an exhaustive configuration, in which all $5! = 120$ permutations are evaluated for a subset of $730$ queries of the RAG-query dataset.

\subsection{Positional bias metrics}
\label{sec:metrics}

Quantifying positional bias requires distinguishing between \emph{permutation sensitivity} and positional bias itself, as noted by~\cite{liusie-etal-2024-teacher}. Permutation sensitivity captures the instability of the model under different orderings of the same candidate set, a \emph{local} property that may change on a per-query basis. Positional bias, by contrast, refers to a systematic tendency to favor certain prompt positions regardless of item relevance, a \emph{global} property reflecting consistent patterns across queries. Crucially, permutation sensitivity does not necessarily imply positional bias, but positional bias necessarily induces some degree of permutation sensitivity. To capture both aspects, we introduce the following metrics.

\paragraph{Permutation sensitivity (PS).} We quantify permutation sensitivity as the probability that the recommended item differs from the item most frequently recommended for the same instance across permutations:
\begin{equation}
    \mathrm{PS} = 100 \cdot \mathbb{E}_{X \sim \mathcal{D}}\left[ \Pr_{\sigma \sim U(S_k)}\!\left[\sigma\!\left(R(X,\sigma)\right) \neq m(X)\right] \right],
    \label{eq:permutation_sensitivity}
\end{equation}
where $m(X) = \arg\max_{j} \Pr_{\sigma \sim U(S_k)}[\sigma(R(X,\sigma)) = j]$ is the most frequently recommended item for $X$. PS provides a directly interpretable practical impact metric: it is the fraction of cases in which the top recommendation changes solely due to item reordering.

\paragraph{Positional bias (PB).} We quantify positional bias as the total variation distance between the position-selection distribution $p(\cdot)$ and the uniform distribution $U_k$, expressed as a percentage:
\begin{equation}
    \mathrm{PB} = 100 \cdot \mathrm{TV}\!\left(p(\cdot),\, U_k\right).
    \label{eq:positional_bias_metric}
\end{equation}
To quantify the statistical uncertainty of this estimate, we construct a $(1-\alpha)$ confidence interval via percentile bootstrap resampling at the query level. Let $\widehat{\mathrm{PB}}^*_1, \ldots, \widehat{\mathrm{PB}}^*_B$ denote the PB values computed from $B$ bootstrap resamples. The confidence interval is then
\begin{equation}
    \mathrm{CI}_{1-\alpha}(\mathrm{PB}) =
    \left[
    \widehat{\mathrm{PB}}^*_{(\lfloor B\alpha/2\rfloor)},\;
    \widehat{\mathrm{PB}}^*_{(\lceil B(1-\alpha/2)\rceil)}
    \right],
    \label{eq:pb_confidence_interval}
\end{equation}
where $\widehat{\mathrm{PB}}^*_{(j)}$ denotes the $j$-th order statistic of the bootstrap distribution. In our experiments we set $\alpha = 0.05$, yielding 95\% confidence intervals. Note that $\mathrm{PB} \geq 0$ by definition of total variation distance, with $\mathrm{PB} = 0$ only under exact uniformity; in practice the confidence interval is bounded away from zero.

Together, PS measures the \emph{practical instability} of recommendations under reordering, while PB quantifies the \emph{structural positional preference} of the model. The latter is meaningful as evidence of positional bias only when the null hypothesis of uniformity has been rejected by the POB-UNBP methodology; otherwise, deviations from uniformity may reflect sampling variability alone.

\begin{remark}\label{rem:ps_upper_bound}
    Permutation sensitivity is an upper bound of positional bias, $\mathrm{PB} \le \mathrm{PS}$. Indeed, the proof of \autoref{thm:stochastic} (\autoref{app:proof_theorem}) only uses that $r_{\varphi_0}$ is a probability distribution over items, so it shows that $\mathrm{TV}(p(\cdot \mid x), U_k) \le \mathbb{E}_{\sigma}\left[\mathrm{TV}(r_\sigma, \mu)\right]$ for any instance $x$ and any distribution $\mu$ over items. Taking $\mu$ as the point mass at $m(x)$, we have $\mathrm{TV}(r_\sigma, \mu) = 1 - r_\sigma(m(x))$, and hence $\mathrm{TV}(p(\cdot \mid x), U_k) \le \Pr_{\sigma}\left[\sigma(R(x,\sigma)) \neq m(x)\right]$. Taking expectations over $X \sim \mathcal{D}$ and using the convexity of the total variation distance yields the claim.
\end{remark}

\paragraph{Strong positional bias (sPB).} The metrics PS and PB in \autoref{eq:permutation_sensitivity}--\autoref{eq:positional_bias_metric}, as well as the POB-ICM and POB-UNBP methodologies, treat any dependence of the recommended item on its position as positional bias, regardless of the relevance of the candidates involved. Such a dependence can manifest in two distinct ways. When several candidates are comparably relevant for the user query, the position of the items may determine which of them is recommended; we refer to this as \emph{soft positional bias}. This form of bias does not degrade the relevance of the recommendation, but it still makes the exposure of comparably relevant items depend on the prompt construction rather than on their content. Alternatively, the position may lead the model to recommend an item that is clearly less relevant than the alternatives; we refer to this as \emph{strong positional bias}. Positional bias, as detected by POB-ICM and POB-UNBP and quantified by PB, comprises both forms; we refer to this combination as \emph{overall positional bias}. To make this distinction operational, we introduce a query-dependent relevance set $\mathcal{R}(X) \subseteq \{1, 2, \ldots, k\}$, identifying the candidate items considered comparably relevant for instance $X$. For each instance, we count how often each candidate item is recommended across its sampled permutations.

\begin{definition}\label{def:relevance_set}
    Given a threshold $\tau \in (0, 1]$, the relevance set $\mathcal{R}(X)$ of an instance $X$ is the smallest set of its most frequently recommended items that together account for at least a fraction $\tau$ of its recommendations.
\end{definition}

The threshold $\tau$ is calibrated from the data. The recommendation frequencies of an instance, sorted in decreasing order, usually show a natural break between the items the model treats as plausibly interchangeable and those it rarely selects regardless of position. For each instance, we locate this break as the largest gap between consecutive sorted frequencies and compute the fraction of recommendations captured by the items above it; $\tau$ is the average of this fraction across instances. Using a common threshold, instead of locating the break separately in each instance, makes the relevance set less sensitive to the noise of the frequencies of individual instances, which are estimated from a limited number of permutations.

We then define strong positional bias as the proportion of recommendations that fall outside the relevance set:
\begin{equation}
    \mathrm{sPB} = 100 \cdot \Pr_{X \sim \mathcal{D}, \sigma \sim U(S_k)}\left[\sigma(R(X,\sigma)) \notin \mathcal{R}(X)\right].
    \label{eq:strong_positional_bias_metric}
\end{equation}
To avoid calibrating the threshold on the same instances on which sPB is evaluated, we use $5$-fold cross-validation over instances. For each fold, $\tau$ is calibrated on the other four folds, and the relevance sets and sPB are computed on the held-out fold. The reported sPB is the average over the five folds, so each instance is evaluated exactly once. Note that sPB and PB are not on a comparable scale, and sPB may exceed PB even though it targets a component of the overall positional bias. The calibrated values of $\tau$ for each setting are reported in \autoref{app:spb_ci}, and they are stable across folds.

\subsection{Experimental setup and implementation details}\label{sec:experimental_setup}

In this section we include relevant details of our experimental setup and implementation. 

\begin{enumerate}
    \item For the RAG scheme, we use the BAAI general embedding model~\cite{xiao_c-pack_2024} as a semantic match between user queries and product titles. For each user query, the top-$5$ most similar products ($K=5$)  are selected, leading to our RAG-query dataset.
    
    \item When evaluating a tuple ($Q$, $\sigma$, $P_1$, $P_2$, $P_3$, $P_4$, $P_5$) from the query-permutation dataset, the final prompt sent to the LLM uses the prompt template given in~\autoref{fig:recommender_prompt}, filled in with the titles and descriptions of the top-$5$ selected products, and the user query. The top-$5$ products appear in this prompt following the order $\sigma$, that is, first the product $P_{\sigma(1)}$, then the product $P_{\sigma(2)}$, etc. 
    
    \item  Several recommender LLMs are tested\footnote{Tested models are in FP16 precision, except for GPT-OSS and LLaMA 3.1 70B}, including LLaMA 3.1 (8B and 70B)~\cite{grattafiori_llama_2024}, LLaMA 4 Scout~\cite{meta2024llama4}, GPT-OSS (20B and 120B)~\cite{gptoss}, and DeepSeek R1 (8B LLaMA distill, 7B Qwen distill and 32B Qwen distill)~\cite{deepseekr1}. The temperature of the recommender model is kept to zero, so the results are deterministic. Thus, if there were no positional bias, the recommender should output the same recommendation for a given user query and top-$5$ products, independently of $\sigma$.
    
    %\item The prompt used in our experiments asks the recommender LLM ($R$ function using our notation) to output the position of the recommended product, obtaining thus a number between $1$ and $5$ ($k=5$), denoted $R(Q, P_1, \ldots, P_5, \sigma)$. Then, we recover back the original position by applying $\sigma(j)$. For example, if the model recommends the product in the $j=3$rd position with the product shuffled using the indices $\left[4, 2, 1, 3, 5\right]$, then the recommended product is $P_1$ (as the first product was moved to the third position after shuffling), $j = R(Q, \sigma, P_1, \ldots, P_5) = 3$, and $\sigma^{-1}(j) = 1$).
    
    \item A \textit{Recommender parser} has been implemented to correct the JSON format. As LLM might fail in generating a JSON following correct syntax, this algorithm fixes the most common mistakes the model makes, including not generating the quotes, not escaping the quotes inside the strings, and so on. As a remark, this recommender model does not always generate a valid recommendation\footnote{A valid recommendation generates a JSON with the specified format and a valid product position.} or the model even refuses to recommend a product if no suitable product is found in the list. We report error rates and refusal rates on \autoref{tab:error_rates}. Since parsing errors are detected and the corresponding samples discarded prior to applying POB-ICM and POB-UNBP, we assess their impact in \autoref{app:failures_order}, where restricting POB-UNBP to the queries with no discarded generations leaves its conclusions unchanged. If failure to produce a valid response is itself influenced by item ordering, this would reflect positional bias already present in the model rather than bias introduced by the methodology.
\end{enumerate}

\begin{table}[t]
    \centering
    \begin{tabular}{c|cccc|cccc}
        \toprule
        \multirow{2}{*}{Model name} & \multicolumn{4}{c|}{Error rate} & \multicolumn{4}{c}{Refusal rate} \\
         & Top-5 & Top-10 & Top-15 & Top-20 & Top-5 & Top-10 & Top-15 & Top-20 \\
        \midrule
        DeepSeek R1 32B & 0.35\% & 0.39\% & 0.24\% & 0.26\% & 4.03\% & 3.89\% & 3.82\% & 4.06\% \\
        DeepSeek R1 7B & 0.81\% & 1.58\% & 1.99\% & 2.46\% & 7.39\% & 5.38\% & 5.14\% & 4.36\% \\
        DeepSeek R1 8B & 0.37\% & 0.40\% & 0.21\% & 0.23\% & 0.01\% & 0.01\% & 0.00\% & 0.00\% \\
        LLaMA 3.1 8B & 0.60\% & 1.99\% & 3.23\% & 4.01\% & 0.01\% & 0.01\% & 0.01\% & 0.01\% \\
        LLaMA 4 Scout & 0.21\% & 0.48\% & 2.05\% & 3.20\% & 0.17\% & 0.06\% & 0.08\% & 0.05\% \\
        GPT OSS 120B & 0.45\% & 0.50\% & 0.45\% & 0.36\% & 0.01\% & 0.02\% & 0.01\% & 0.03\% \\
        GPT-OSS 20B & 1.01\% & 1.07\% & 1.11\% & 1.22\% & 0.04\% & 0.03\% & 0.03\% & 0.02\% \\
        LLaMA 3.1 70B & 0.31\% & 0.41\% & 0.44\% & 0.54\% & 0.12\% & 0.17\% & 0.22\% & 0.26\% \\
        \bottomrule
        \vspace{0mm}
    \end{tabular}
    \caption{Recommendation parsing error rate and refusal rate for each candidate-set size, each computed over the $40{,}000$ samples of the corresponding query-permutation dataset.}
    \label{tab:error_rates}
\end{table}

All of the described implementation details, both for the generation of the dataset and the evaluation of bias, are available in our source code\footnote{\url{https://github.com/GoComputing/Positional-Bias.git}}, which allows reproduction of all of the results presented.

\subsection{Results and analysis}\label{sec:results_analysis}

% Introduce the four experiments
\autoref{tab:tests_results} summarizes the outcome of both proposed methodologies, POB-ICM and POB-UNBP, together with the quantification metrics PS, PB, and sPB for the top-$5$ setting, providing a unified overview of the statistical significance and magnitude of positional bias across all evaluated models. We then examine these results in turn. Subsection~\ref{sec:results_icm} details the POB-ICM results, including the contingency-matrix visualization and the associated bias-inspection analysis, characterizing \emph{how} positional bias manifests for each model. Subsection~\ref{sec:results_unbp} reports the POB-UNBP results, which corroborate these findings with a more sample-efficient test and additionally assess the robustness of positional bias as the candidate-set size $k$ and the decoding temperature $t$ vary. Finally, Subsection~\ref{sec:results_quantification} analyzes the PS, PB, and sPB metrics across these same conditions, disentangling generic permutation sensitivity from the structural, position-driven component of the bias.

\begin{table}
    \centering
    
    \begin{tabular}{x{2.6cm}x{1.7cm}x{2cm}x{1.1cm}x{1.8cm}x{2.8cm}x{2cm}}
        \toprule
        Model name & $p$-value (POB-ICM) & $p$-value (POB-UNBP) & $L$ & Permutation sensitivity (PS) & Positional bias (PB) &  Strong Positional Bias (sPB) \\
        \midrule
        DeepSeek R1 32B & \textbf{<0.0005} & \textbf{7.456e-07} & 0.97{\%} & 21.857 & 2.563, CI: (2.0, 3.2) & 3.1 \\
        DeepSeek R1 7B & \textbf{<0.0005} & \textbf{6.197e-60} & 4.15\% & 33.983 & 5.612, CI: (5.1, 6.3) & 9.8 \\
        DeepSeek R1 8B & \textbf{<0.0005} & \textbf{5.875e-05} & 0.79\% & 29.758 & 2.391, CI: (1.8, 3.0) & 6.7 \\
        LLaMA 3.1 8B & \textbf{<0.0005} & \textbf{4.457e-26} & 3.01\% & 30.395 & 4.187, CI: (3.6, 4.8) & 6.7 \\
        LLaMA 4 Scout & \textbf{<0.0005} & \textbf{1.188e-04} & 0.87\% & 22.393 & 2.886, CI: (2.4, 3.4) & 3.2 \\
        GPT OSS 120B & \textbf{<0.0005} & \textbf{8.879e-30} & 3.14\% & 15.487 & 4.368, CI: (3.8, 5.0) & 1.0 \\
        GPT-OSS 20B & \textbf{<0.0005} & \textbf{4.814e-65} & 4.73\% & 17.709 & 5.944, CI: (5.3, 6.6) & 2.1 \\
        LLaMA 3.1 70B & \textbf{<0.0005} & \textbf{7.216e-24} & 3.40\% & 22.654 & 4.638, CI: (4.0, 5.3) & 3.1 \\
        \bottomrule
        \vspace{0mm}    
    \end{tabular}
    \caption{Results from the proposed bias quantification metrics and proposed statistical methodologies for top-$5$ setting. POB-ICM $p$-values are Monte Carlo estimates from $B = 2000$ permutations, so their smallest attainable value is $1/(B+1) = 0.0005$. $L$ value is the lower bound of $TV(p(\cdot), U_k)$ according to the POB-UNBP test. Confidence intervals of sPB reported at \autoref{app:spb_ci}.}
    \label{tab:tests_results}
\end{table}

\subsubsection{Results and visualization of POB-ICM methodology}
\label{sec:results_icm}

We first apply the POB-ICM methodology from \autoref{sec:independency_methodology}. To do so, the contingency matrix is built from the results generated using~\autoref{eq:contingency_matrix}. This empirical contingency matrix is represented in~\autoref{fig:contingency_matrix} for LLaMA 8B. The statistical test results of the POB-ICM methodology are shown in \autoref{tab:tests_results}.

\begin{resultbox}
Under the POB-ICM methodology, the calibrated Chi-square test rejects the null hypothesis of independence between candidate-item position and the selected output in all evaluated models, with p-values below the standard significance threshold of $0.05$, providing strong evidence that positional dependence is a substantive and reproducible property of the evaluated pipelines. See~\autoref{tab:tests_results} for exact results.
\end{resultbox}

We can further analyze the positional bias through \autoref{fig:contingency_matrix}. The figure includes two configurations: 730 queries with exhaustive permutations (120 per query), and 2000 queries with sampled permutations (20 per query). In the exhaustive configuration, the products in the first position given by the semantic matcher (first row) are more likely to be selected, while the last product is less likely to be selected (last row). Indeed, the first product is the preferred product by the semantic matcher, so it is more likely to also be preferred by the LLM. Note that the semantic matcher only has access to the product title, while the LLM has access to both the title and the description of the product, and thus, the recommendation of the LLM should be more reliable. We stress that this row structure reflects item relevance, not positional bias. Rows are indexed by the original (pre-shuffling) position $Y = \sigma(R(X,\sigma))$, so unequal row totals only indicate that the LLM tends to agree with the retriever on which candidates are more relevant. This does not confound the detection of positional bias, for two reasons. First, the candidates are presented to the LLM in a uniformly random order $\sigma$, and retrieval scores are never included in the prompt, so the retriever's ranking cannot favor any prompt position. Second, the null hypothesis of POB-ICM (independence between $Y$ and $\sigma$) places no restriction on the row totals: positional bias manifests instead as differences between columns, that is, as a dependence of the selected item on the permutation applied. Similarly, the distribution $p(\cdot)$ used by POB-UNBP is indexed by prompt position, and \autoref{prop:deterministic} and \autoref{thm:stochastic} guarantee its uniformity under absence of positional bias for every instance, regardless of how unequal the relevance of its candidates is. In the sampled configuration, this influence of the similarity ranker is less clearly visible, as the reduced number of permutations per query introduces additional noise in the contingency matrix. Nevertheless, the overall positional bias structure remains consistent across both configurations, supporting the robustness of the detected effect.

\begin{figure}
    \centering
    \begin{subfigure}[t]{0.49\linewidth}
        \centering
        \includegraphics[width=\linewidth]{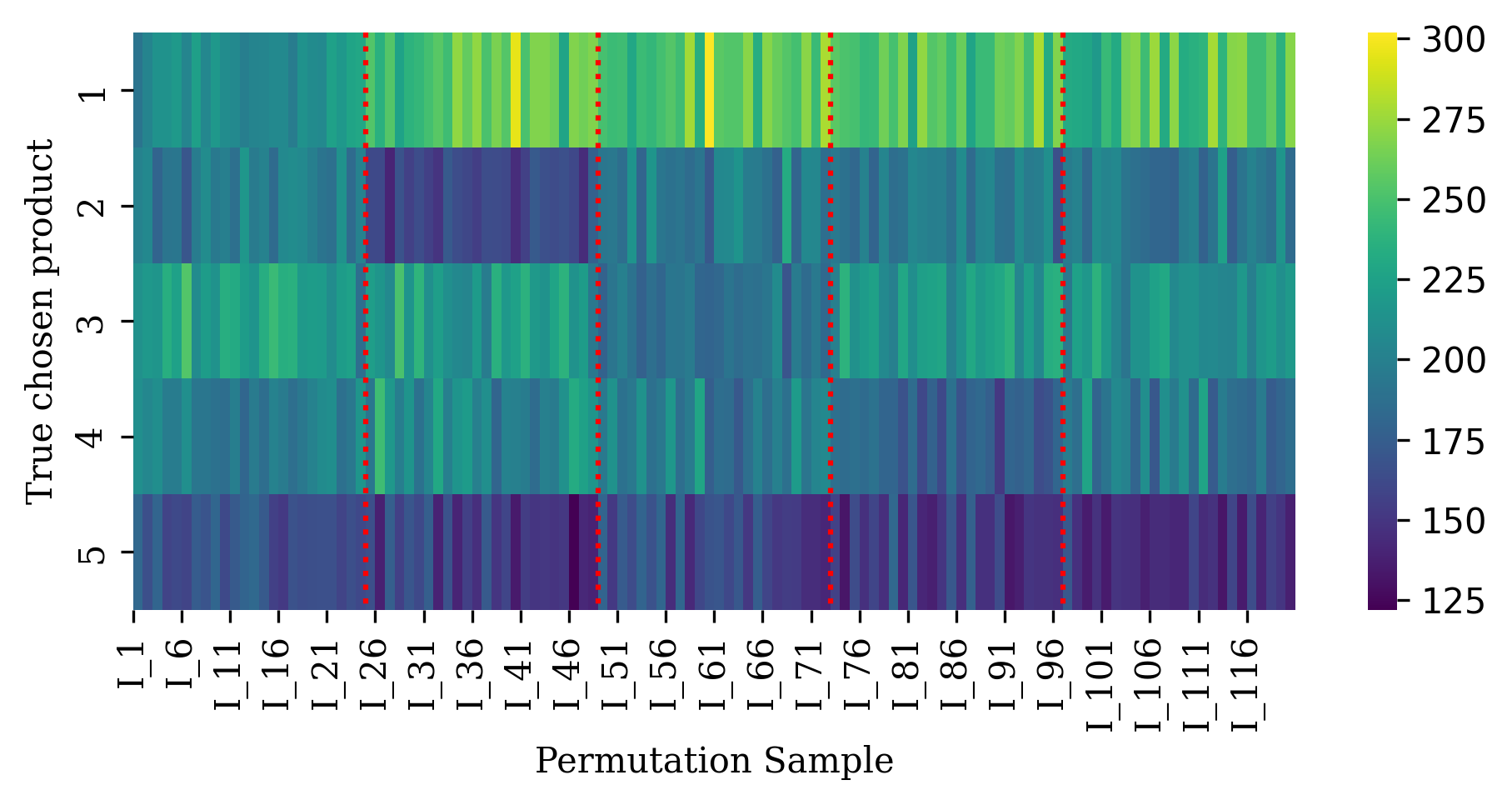}
        \caption{Generated using 730 queries and 120 permutations per query.}
    \end{subfigure}
    \begin{subfigure}[t]{0.49\linewidth}
        \centering
        \includegraphics[width=\linewidth]{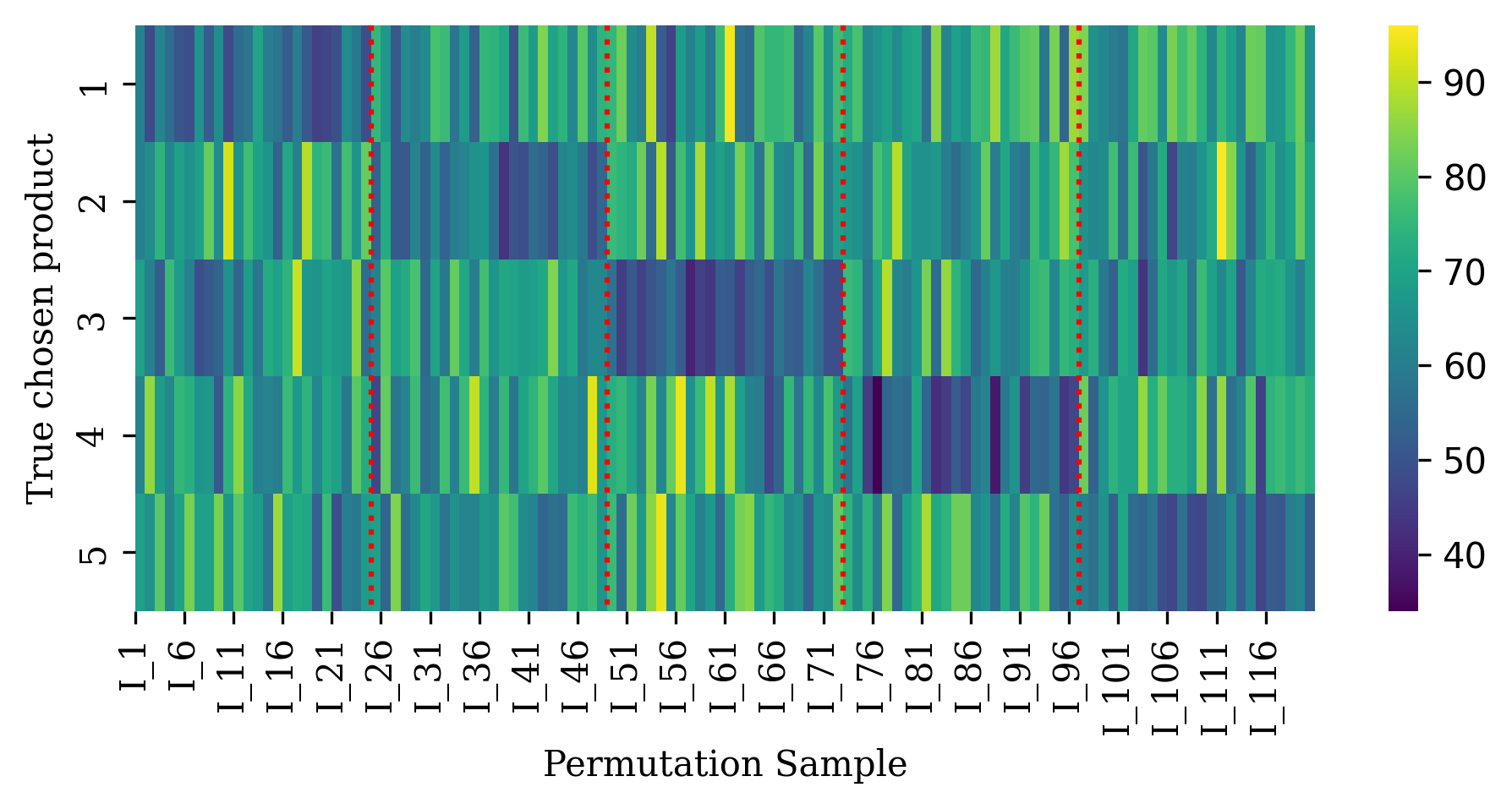}
        \caption{Generated using 2000 queries and 20 sampled permutations per query.}
    \end{subfigure}
    
    \caption{Contingency matrix from \autoref{sec:experimental_setup}, shown as a heatmap. The plot was generated using LLaMA 3.1 8B as recommender. Each column shows the empirical distribution of recommendations conditioned on a permutation, with permutations sorted lexicographically. Red lines indicate when the first, second, third, fourth, and fifth product are shuffled to the first position, respectively. Consistency was ensured by shuffling products according to the permutation sample, generating the recommendation, and then unshuffling the recommendation positions.}
    \label{fig:contingency_matrix}
\end{figure}

%Second, we observe that moving products to the first position reduces the likelihood of them being selected. This is observed by understanding how the abscissa axis is represented. This axis represents the shuffles, where these are lexicographically sorted. This means that each shuffle in a block of $24$ $ $($120 / 5)$ shares the first index. In other words, each of these blocks of shuffles moves the product from a given position to the first position. This given position is determined by the block position in this domain. For example, the leftmost block in the plot (separated by a red line) does not move the first product. The next block moves the second product to the first position, and the last block (the rightmost block) moves the last product to the first position. Given this note, we observe that products moved to the first position are not likely selected by the recommender model.

Positional bias in this setting is better visualized using our proposed \textit{bias inspection matrix}, described in~\autoref{sec:independency_methodology}. This matrix is interpreted by columns. Each column represents products selected when these are moved to the position $j$. This plot is shown in~\autoref{fig:bias_analysis}, generated using the 2000-query configuration with 20 sampled permutations per query. For example, in the plot associated with LLaMA 3.1 8B, we note that products moved to the first position are less likely selected. We also note that products moved to the last position are more likely to be selected. There are several potential explanations for this, as, for example, LLMs pay more attention to tokens that are closer to the end of the prompt, as observed in~\cite{zhang_can_2024}. This is not the same for models from the GPT-OSS family, where products moved to the first position are more likely recommended, so the positional bias is dependent either on the architecture, training, or post-training process. Interestingly, DeepSeek R1 8B and 32B exhibit comparatively weaker positional bias. While the underlying causes are not investigated here, this observation suggests that architectural or training differences may modulate positional sensitivity. Notably, these qualitative patterns are consistent with those observed under the exhaustive 730-query configuration, supporting the robustness of the findings across evaluation settings.

\begin{observationbox}
Using the proposed POB-ICM visualization methodology in \autoref{fig:bias_analysis}, several models exhibit different positional biases. LLaMA 3.1 8B tends to select the last products in the prompt, while models like GPT-OSS and LLaMA 3.1 70B tend to select first products.
\end{observationbox}

\begin{figure}
    \centering
    \includegraphics[width=\linewidth]{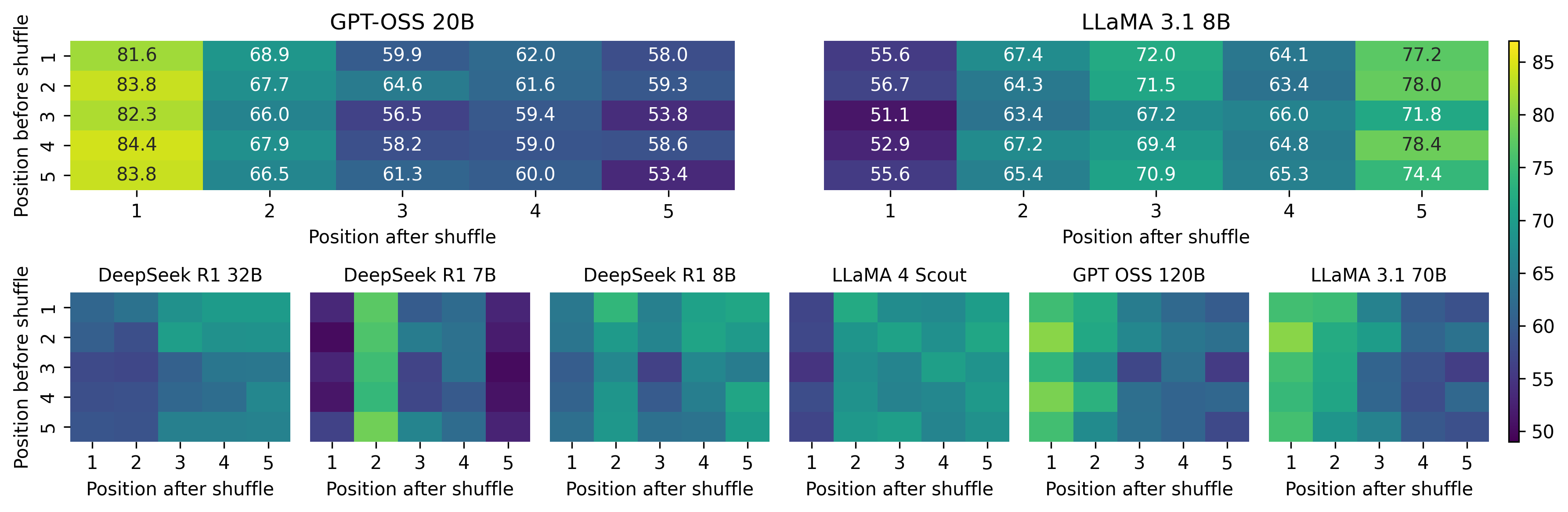}
    \caption{Bias analysis given the~\autoref{eq:bias_analysis} for different recommender LLMs. For each model, a column of its associated heatmap represents the mean of the (non-normalized) counts of selected products by the model when these were moved to the position represented by said column. In the absence of positional bias, all columns would be identical.}
    \label{fig:bias_analysis}
\end{figure}

\subsubsection{Results on POB-UNBP methodology}
\label{sec:results_unbp}

A more optimized methodology, POB-UNBP, is also proposed in \autoref{sec:uniform_methodology}. This methodology mathematically proves approximate uniformity of the probability distribution $p(i)$ under the assumption of absence of positional bias, both in deterministic and stochastic settings. This allows testing whether $p(i)$ deviates from the uniform distribution through the lower confidence bound $L$ on $\mathrm{TV}(p(\cdot), U_k)$ given in~\autoref{eq:tv_lower_bound}. Thus, a much smaller number of queries are required to evaluate the positional bias, at the expense of lower interpretability of the results (see \autoref{sec:simulation_study}). This reduction in sample complexity is enabled by averaging random permutations of the candidate set, a strategy aligned with the classical permutation-based hypothesis testing~\cite{good2013permutation}. Since the permutations of each query are aggregated into the instance-level statistics $Z_j(i)$, which are independent across queries, all $20$ sampled permutations per query are used without violating the independence assumption. Along with $L$, we report the $p$-value of the test in \autoref{cor:test_level} for $\varepsilon = 0$, that is, of the null hypothesis of exact uniformity, defined as the smallest significance level $\alpha$ at which $L > 0$. Since $L > 0$ if and only if $L_i > 0$ for some position $i$, this $p$-value is $\min\left(1,\, k \cdot \min_i p_i\right)$, where $p_i$ is the one-sided $p$-value of the $t$-statistic $\bar{Z}(i)\sqrt{n}/s_i$ with $n-1$ degrees of freedom; see~\autoref{tab:pob_unbp_topk} for exact results.

\begin{table}
    \centering
    
    \begin{tabular}{c|cccccccc}
        \toprule
        & \multicolumn{2}{c}{Top-5} & \multicolumn{2}{c}{Top-10} & \multicolumn{2}{c}{Top-15} & \multicolumn{2}{c}{Top-20} \\
        Model name & $p$-value & $L$ & $p$-value & $L$ & $p$-value & $L$ & $p$-value & $L$ \\
        \midrule
        DeepSeek R1 32B & \textbf{7.456e-07} & 0.97{\%} & \textbf{5.130e-17} & 2.83{\%} & \textbf{2.202e-47} & 4.29{\%} & \textbf{2.098e-48} & 6.78{\%} \\
        DeepSeek R1 7B & \textbf{6.197e-60} & 4.15{\%} & \textbf{9.362e-55} & 5.79{\%} & \textbf{5.966e-79} & 9.51{\%} & \textbf{4.449e-70} & 9.89{\%} \\
        DeepSeek R1 8B & \textbf{5.875e-05} & 0.79{\%} & \textbf{1.150e-11} & 1.96{\%} & \textbf{2.253e-46} & 4.30{\%} & \textbf{2.111e-71} & 7.24{\%} \\
        LLaMA 3.1 8B & \textbf{4.457e-26} & 3.01{\%} & \textbf{4.928e-17} & 3.13{\%} & \textbf{3.379e-16} & 4.08{\%} & \textbf{1.465e-26} & 3.77{\%} \\
        LLaMA 4 Scout & \textbf{1.188e-04} & 0.87{\%} & \textbf{1.438e-60} & 8.30{\%} & \textbf{8.023e-208} & 21.43{\%} & \textbf{1.341e-249} & 30.19{\%} \\
        GPT OSS 120B & \textbf{8.879e-30} & 3.14{\%} & \textbf{2.233e-41} & 3.62{\%} & \textbf{6.968e-48} & 5.10{\%} & \textbf{1.601e-58} & 7.00{\%} \\
        GPT-OSS 20B & \textbf{4.814e-65} & 4.73{\%} & \textbf{5.034e-82} & 5.69{\%} & \textbf{6.788e-96} & 7.78{\%} & \textbf{2.993e-100} & 9.03{\%} \\
        LLaMA 3.1 70B & \textbf{7.216e-24} & 3.40{\%} & \textbf{4.845e-25} & 2.99{\%} & \textbf{2.463e-29} & 4.40{\%} & \textbf{1.045e-41} & 7.04{\%} \\
        \bottomrule
    \end{tabular}
    \caption{POB-UNBP statistical test $p$-value and the computed lower bound $L$ for different top-$K$ values and a fixed temperature $t=0$. Bold values represent statistically detectable positional bias according to the threshold of $0.05$.}
    \label{tab:pob_unbp_topk}
\end{table}

\begin{resultbox}
By the POB-UNBP methodology, positional bias is confirmed in all evaluated models, requiring a smaller quantity of samples to confirm positional bias compared to the POB-ICM methodology (see \autoref{sec:simulation_study}).
\end{resultbox}

% Top-k, k in {5, 10, 15, 20}
The statistical tests confirm non-uniformity in all evaluated models, providing further evidence of positional bias. To assess whether the strength and location of this bias depend on the size of the candidate set, we visualize the results in \autoref{fig:raw_probability_distribution}. Two qualitatively distinct behaviors emerge. DeepSeek R1 variants and LLaMA 3.1 8B exhibit a preference above uniform baseline for middle positions over the very first positions of the prompt, and the location of this preference is itself sensitive to $k$: the bias shifts toward later relative positions for smaller $k$ and toward earlier relative positions as $k$ increases. LLaMA 4 Scout displays a related but substantially more pronounced pattern, where the preference for later positions grows markedly with $k$, indicating that this model is particularly sensitive to candidate-set size. In contrast, GPT-OSS (20B and 120B) and LLaMA 3.1 70B consistently favor the first positions of the prompt, and this tendency remains stable across all evaluated values of $k$, suggesting that positional bias in these models is largely independent of the number of candidates presented.

% Temperature \in {0.2, 0.4, 0.6}

To assess whether positional bias is an artifact of deterministic decoding or persists under stochastic sampling, we repeat the POB-UNBP evaluation at top-$5$ for $t \in \{0.0, 0.2, 0.4, 0.6\}$, where $t=0.0$ corresponds to the deterministic setting reported in \autoref{sec:results_analysis}. We use these temperatures according to the used temperature in the literature~\cite{hou_large_2024}. The corresponding $p$-values are reported in \autoref{tab:pob_unbp_temp}.

\begin{table}
    \centering
    
    \begin{tabular}{c|cccccccc}
        \toprule
        & \multicolumn{2}{c}{$t = 0.0$} & \multicolumn{2}{c}{$t = 0.2$} & \multicolumn{2}{c}{$t = 0.4$} & \multicolumn{2}{c}{$t = 0.6$} \\
        Model name & $p$-value & $L$ & $p$-value & $L$ & $p$-value & $L$ & $p$-value & $L$ \\
        \midrule
        DeepSeek R1 32B & \textbf{7.456e-07} & 0.97{\%} & \textbf{2.001e-09} & 1.14{\%} & \textbf{6.392e-06} & 0.79{\%} & \textbf{6.290e-07} & 1.12{\%} \\
        DeepSeek R1 7B & \textbf{6.197e-60} & 4.15{\%} & \textbf{1.924e-64} & 4.29{\%} & \textbf{7.115e-58} & 4.24{\%} & \textbf{2.998e-66} & 4.40{\%} \\
        DeepSeek R1 8B & \textbf{5.875e-05} & 0.79{\%} & \textbf{8.181e-06} & 0.90{\%} & \textbf{9.225e-08} & 1.21{\%} & \textbf{1.701e-16} & 1.65{\%} \\
        LLaMA 3.1 8B & \textbf{4.457e-26} & 3.01{\%} & \textbf{2.049e-28} & 2.67{\%} & \textbf{4.884e-21} & 2.27{\%} & \textbf{1.925e-17} & 1.97{\%} \\
        LLaMA 4 Scout & \textbf{1.188e-04} & 0.87{\%} & \textbf{1.476e-05} & 0.97{\%} & \textbf{5.367e-05} & 0.97{\%} & \textbf{6.510e-05} & 1.00{\%} \\
        GPT OSS 120B & \textbf{8.879e-30} & 3.14{\%} & \textbf{3.367e-31} & 3.27{\%} & \textbf{4.932e-31} & 3.40{\%} & \textbf{1.528e-32} & 3.24{\%} \\
        GPT-OSS 20B & \textbf{4.814e-65} & 4.73{\%} & \textbf{1.098e-73} & 5.26{\%} & \textbf{1.878e-72} & 5.06{\%} & \textbf{1.475e-69} & 4.93{\%} \\
        LLaMA 3.1 70B & \textbf{7.216e-24} & 3.40{\%} & \textbf{1.690e-28} & 3.63{\%} & \textbf{7.065e-26} & 3.78{\%} & \textbf{6.539e-39} & 4.48{\%} \\
        \bottomrule
    \end{tabular}
    \caption{POB-UNBP statistical test $p$-value and the lower bound $L$ for different decoding $t$ temperatures and a fixed $K=5$. Bold values represent statistically detectable positional bias according to the threshold of $0.05$.}
    \label{tab:pob_unbp_temp}
\end{table}

\begin{resultbox}
Positional bias is detected by the POB-UNBP methodology consistently across decoding temperatures, with $p$-values remaining below the standard significance threshold of $0.05$ for all models at every evaluated value of $t$. The magnitude and significance of the effect do not vary substantially as $t$ increases, indicating that positional bias is not an artifact of deterministic decoding but persists under stochastic sampling.
\end{resultbox}

\autoref{tab:temperature_experiment} holds the position-selection distributions across temperatures for top-$5$. Consistently with the statistical results in \autoref{tab:pob_unbp_temp}, the qualitative shape of the bias remains stable as $t$ varies for all evaluated models. GPT-OSS (20B and 120B) and LLaMA 3.1 70B preferentially select the first positions of the prompt across all temperatures, mirroring the behavior observed in the top-$k$ experiment. LLaMA 3.1 8B, DeepSeek R1 32B, and DeepSeek R1 8B instead favor later positions, while DeepSeek R1 7B exhibits a distinct pattern, concentrating its bias on the second position rather than the first or last. LLaMA 4 Scout shows the complementary behavior, systematically avoiding the first position regardless of temperature. Taken together, these results indicate that positional bias is primarily determined by model identity rather than by decoding stochasticity, reinforcing the conclusion that it reflects a structural property of the underlying model rather than an effect of the sampling procedure.

\begin{figure}[t]
    \centering
    \begin{subfigure}[t]{\linewidth}
        \centering
        \includegraphics[width=\linewidth]{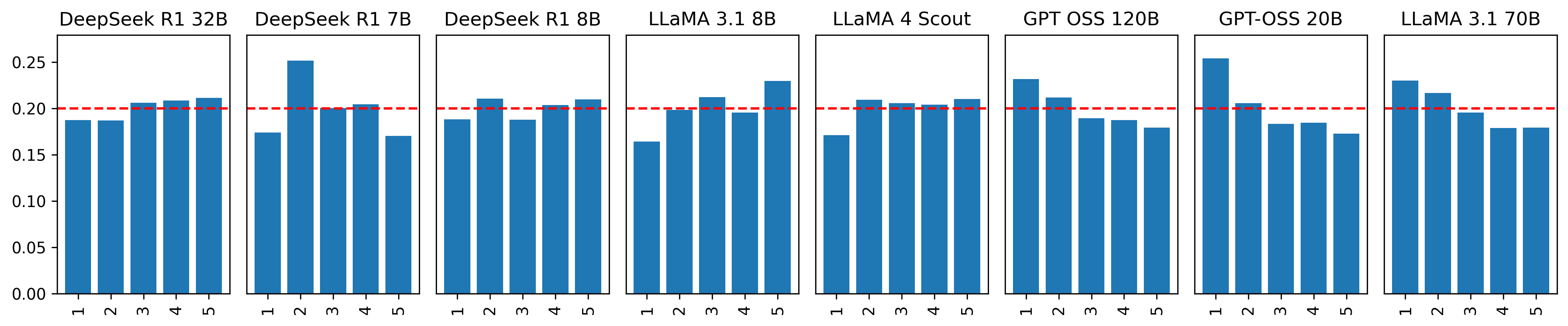}
        \caption{$k = 5$}
    \end{subfigure}
    \begin{subfigure}[t]{\linewidth}
        \centering
        \includegraphics[width=\linewidth]{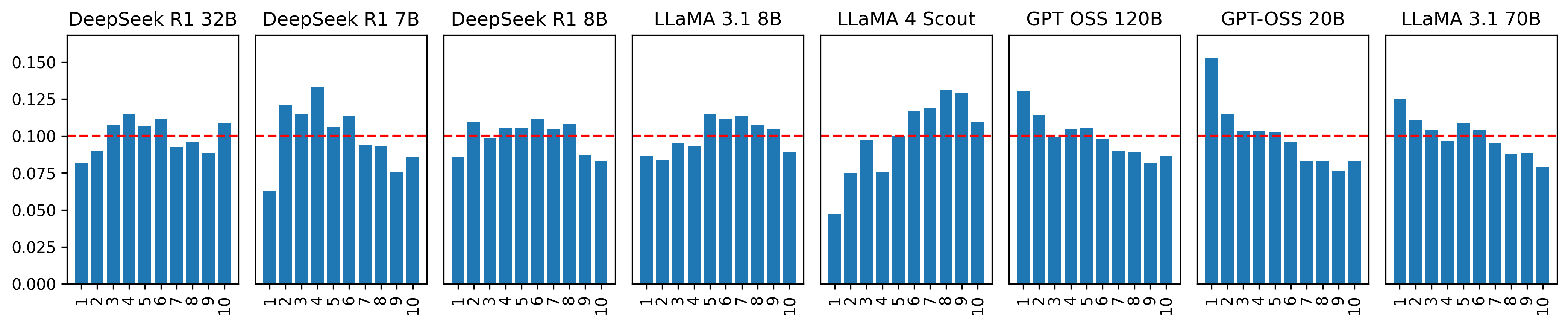}
        \caption{$k = 10$}
    \end{subfigure}
    \begin{subfigure}[t]{\linewidth}
        \centering
        \includegraphics[width=\linewidth]{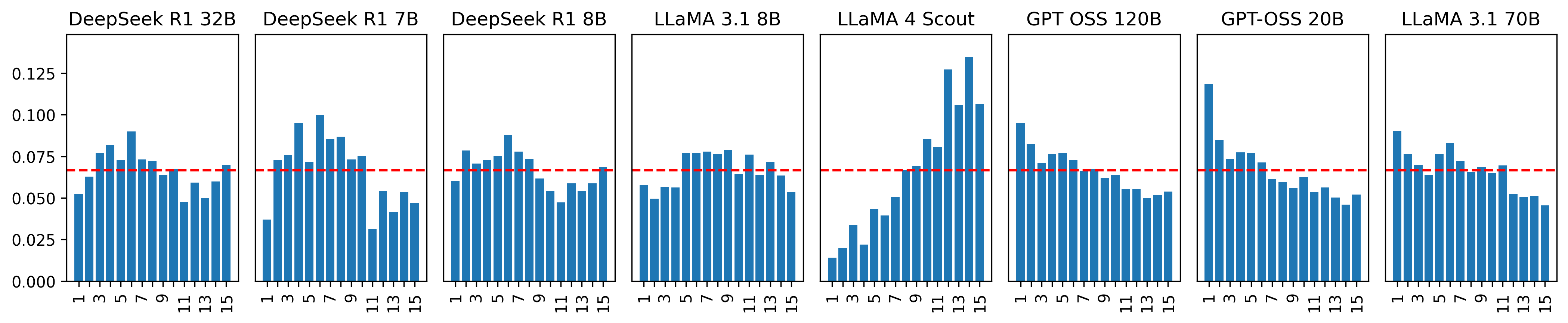}
        \caption{$k = 15$}
    \end{subfigure}
    \begin{subfigure}[t]{\linewidth}
        \centering
        \includegraphics[width=\linewidth]{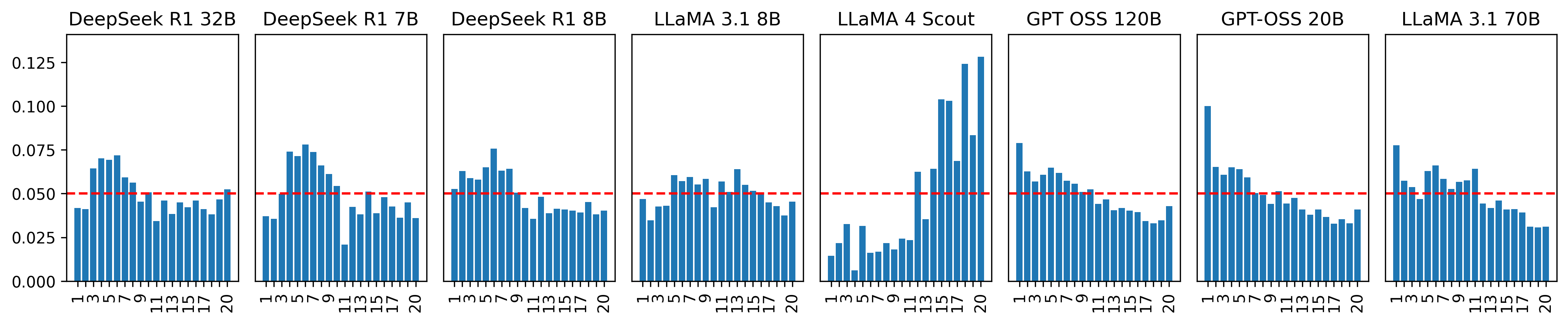}
        \caption{$k = 20$}
    \end{subfigure}
    \caption{Representation of probability distribution p(i) defined in \autoref{eq:probability_distribution} for different $k$ values. Red-dashed lines represent uniform baseline $1/k$.}
    \label{fig:raw_probability_distribution}
\end{figure}

\begin{table}
    \centering
    \begin{tabular}{cc|x{1.5cm}x{1.5cm}x{1.5cm}x{1.5cm}x{1.5cm}}
        \toprule
        Model name & Temperature & p(1) & p(2) & p(3) & p(4) & p(5) \\
        \midrule
        \multirow{4}{*}{DeepSeek R1 32B} & 0.0 & 0.1874 & 0.1869 & 0.2060 & 0.2083 & 0.2114 \\
        & 0.2 & 0.1869 & 0.1868 & 0.2046 & 0.2080 & 0.2137 \\
        & 0.4 & 0.1873 & 0.1879 & 0.2055 & 0.2097 & 0.2097 \\
        & 0.6 & 0.1872 & 0.1852 & 0.2053 & 0.2103 & 0.2119 \\
        \midrule
        \multirow{4}{*}{DeepSeek R1 7B} & 0.0 & 0.1738 & 0.2514 & 0.2004 & 0.2043 & 0.1701 \\
        & 0.2 & 0.1686 & 0.2535 & 0.2012 & 0.2039 & 0.1729 \\
        & 0.4 & 0.1689 & 0.2509 & 0.2012 & 0.2069 & 0.1721 \\
        & 0.6 & 0.1644 & 0.2544 & 0.1994 & 0.2059 & 0.1759 \\
        \midrule
        \multirow{4}{*}{DeepSeek R1 8B} & 0.0 & 0.1883 & 0.2104 & 0.1878 & 0.2037 & 0.2098 \\
        & 0.2 & 0.1898 & 0.2083 & 0.1896 & 0.2001 & 0.2121 \\
        & 0.4 & 0.1847 & 0.2094 & 0.1884 & 0.2029 & 0.2146 \\
        & 0.6 & 0.1847 & 0.2057 & 0.1870 & 0.2006 & 0.2221 \\
        \midrule
        \multirow{4}{*}{LLaMA 3.1 8B} & 0.0 & 0.1644 & 0.1981 & 0.2122 & 0.1956 & 0.2297 \\
        & 0.2 & 0.1629 & 0.2021 & 0.2088 & 0.1973 & 0.2289 \\
        & 0.4 & 0.1679 & 0.2026 & 0.2097 & 0.1960 & 0.2237 \\
        & 0.6 & 0.1748 & 0.1998 & 0.2095 & 0.1951 & 0.2208 \\
        \midrule
        \multirow{4}{*}{LLaMA 4 Scout} & 0.0 & 0.1711 & 0.2092 & 0.2057 & 0.2041 & 0.2099 \\
        & 0.2 & 0.1691 & 0.2105 & 0.2049 & 0.2053 & 0.2102 \\
        & 0.4 & 0.1694 & 0.2097 & 0.2068 & 0.2055 & 0.2085 \\
        & 0.6 & 0.1702 & 0.2097 & 0.2072 & 0.2040 & 0.2090 \\
        \midrule
        \multirow{4}{*}{GPT OSS 120B} & 0.0 & 0.2318 & 0.2119 & 0.1895 & 0.1875 & 0.1793 \\
        & 0.2 & 0.2319 & 0.2122 & 0.1870 & 0.1900 & 0.1790 \\
        & 0.4 & 0.2320 & 0.2137 & 0.1874 & 0.1892 & 0.1777 \\
        & 0.6 & 0.2329 & 0.2112 & 0.1897 & 0.1878 & 0.1784 \\
        \midrule
        \multirow{4}{*}{GPT-OSS 20B} & 0.0 & 0.2538 & 0.2057 & 0.1834 & 0.1843 & 0.1728 \\
        & 0.2 & 0.2574 & 0.2073 & 0.1827 & 0.1813 & 0.1712 \\
        & 0.4 & 0.2574 & 0.2052 & 0.1833 & 0.1810 & 0.1730 \\
        & 0.6 & 0.2549 & 0.2071 & 0.1814 & 0.1838 & 0.1729 \\
        \midrule
        \multirow{4}{*}{LLaMA 3.1 70B} & 0.0 & 0.2299 & 0.2165 & 0.1953 & 0.1789 & 0.1794 \\
        & 0.2 & 0.2307 & 0.2166 & 0.1930 & 0.1807 & 0.1789 \\
        & 0.4 & 0.2289 & 0.2205 & 0.1944 & 0.1790 & 0.1773 \\
        & 0.6 & 0.2363 & 0.2202 & 0.1940 & 0.1760 & 0.1736 \\
        \bottomrule
    \end{tabular}
    \caption{Probability distributions p(i) defined in \autoref{eq:probability_distribution} for different temperature values and top-5 setting. Uniform baseline is $0.2$.}
    \label{tab:temperature_experiment}
\end{table}

\subsubsection{Positional bias quantification}
\label{sec:results_quantification}

Beyond the hypothesis-testing results reported above, we now characterize positional bias through the quantification metrics introduced in \autoref{sec:methodologies}: permutation sensitivity (PS), positional bias (PB), and strong positional bias (sPB). Before reporting how these metrics evolve across $k$ and $t$, we first illustrate the rationale behind the relevance set $\mathcal{R}(X)$ on which sPB is based.

\autoref{fig:relevance_set} shows, for twenty randomly sampled queries under the exhaustive configuration (730 queries, 120 permutations each), the empirical selection count of each candidate item in its original (pre-shuffling) position. These histograms reveal substantial heterogeneity across queries in the number of items the model treats as plausible recommendations. For several queries, a single item is selected overwhelmingly often regardless of the permutation applied, indicating a clear, unambiguous preference largely unaffected by item ordering. For other queries, however, selection counts are split across two or more items, suggesting that the model regards several candidates as comparably relevant to the user query, and that residual instability across permutations within this subset reflects soft rather than strong positional bias. This heterogeneity directly motivates the construction of a query-dependent relevance set $\mathcal{R}(X)$ (\autoref{def:relevance_set}): treating every change in the recommended item as evidence of strong positional bias, without accounting for cases of comparably relevant alternatives, would overstate the practical impact of position on the relevance of the recommendations.

\begin{figure}
    \centering
    \includegraphics[width=\linewidth]{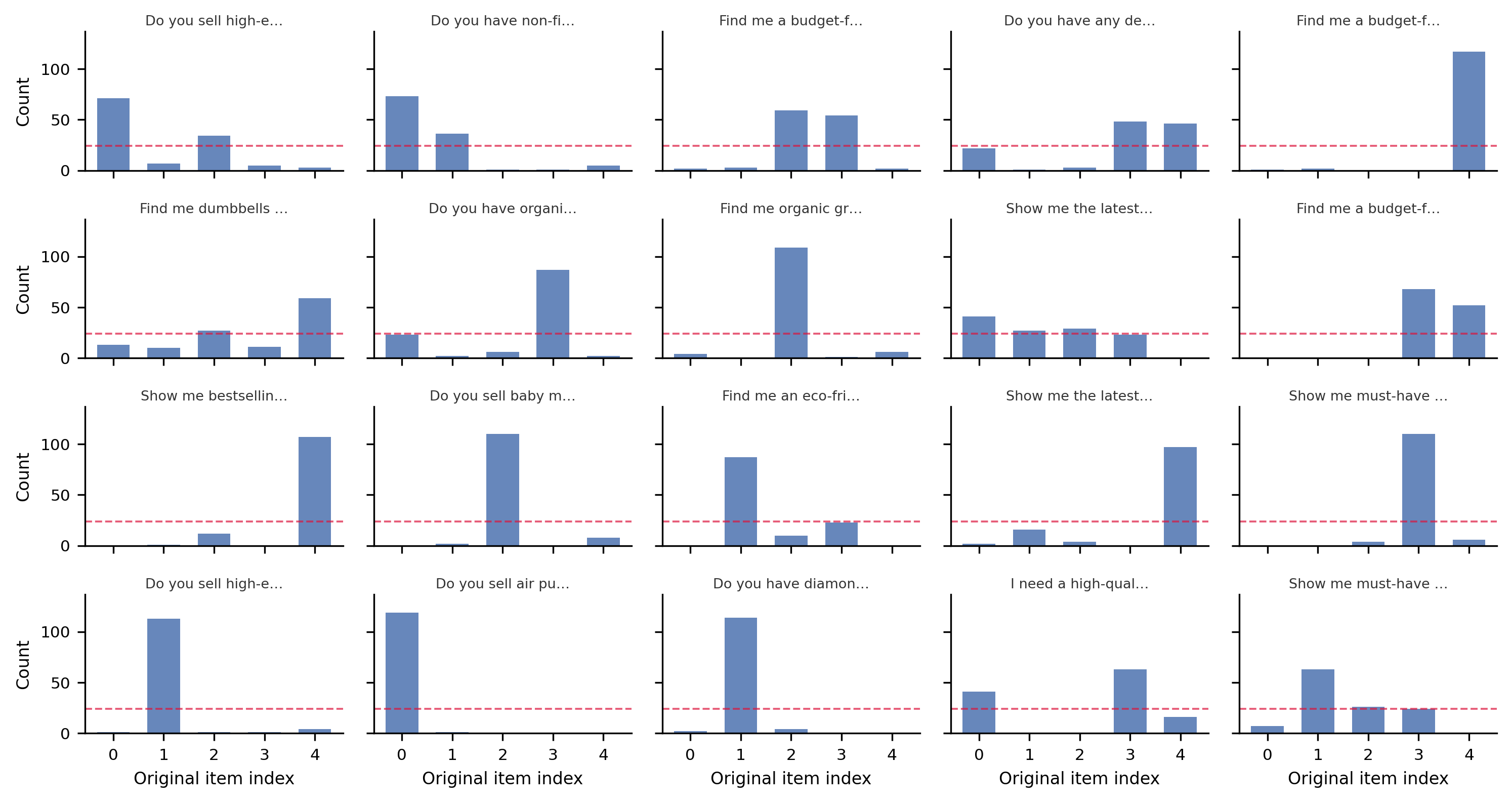}
    \caption{Selection counts of each candidate item, by original (pre-shuffling) position, for twenty random queries of the exhaustive configuration, using LLaMA 3.1 8B. Each panel shows, for one query, how many of the 120 permutations recommended each item.}
    \label{fig:relevance_set}
\end{figure}

\autoref{tab:metrics_topk} reports PS, PB, and sPB across $k \in \{5, 10, 15, 20\}$ for all evaluated models. Several patterns emerge. First, permutation sensitivity is substantial and grows monotonically with $k$ for every model, confirming that prompt sensitivity is a genuine and increasingly impactful phenomenon as the candidate set grows: e.g., DeepSeek R1 7B rises from $34.0\%$ at $k=5$ to $58.1\%$ at $k=20$, meaning the recommended item changes on more than half of all reorderings once twenty candidates are presented. Second, positional bias (PB) generally increases with $k$ (the only exception being a slight decrease for DeepSeek R1 7B from $k=15$ to $k=20$), though the magnitude of this growth varies considerably across model families. LLaMA 4 Scout exhibits by far the steepest increase, with PB rising from $2.9\%$ at $k=5$ to $33.8\%$ at $k=20$ (more than a tenfold increase), consistent with the pronounced $k$-dependence already observed in the probability distributions in \autoref{fig:raw_probability_distribution}. In contrast, GPT-OSS (20B and 120B) and LLaMA 3.1 70B show comparatively modest growth in PB (roughly a $2$--$2.5\times$ increase from $k=5$ to $k=20$), mirroring their stable first-position preference across $k$ noted previously. Third, strong positional bias (sPB) increases monotonically with $k$ for every model, growing between roughly 3 and 11 from $k=5$ to $k=20$; this suggests that as the candidate set grows, a progressively larger share of permutation-induced changes involve items outside the relevance set, at a model-dependent rate ranging from DeepSeek R1 8B ($6.7\%$ to $20.6\%$) to LLaMA 4 Scout ($3.2\%$ to $36.6\%$). Taken together, these results indicate that larger candidate sets do not merely amplify generic instability, but specifically exacerbate the structural, position-driven component of that instability, with the magnitude of this effect strongly modulated by model identity.

\begin{table}
    \centering
    \begin{tabular}{c|cccc|cccc|cccc}
        \toprule
        \multirow{2}{*}{Model name} & \multicolumn{4}{c}{PS} & \multicolumn{4}{c}{PB} & \multicolumn{4}{c}{sPB} \\
         & 5 & 10 & 15 & 20 & 5 & 10 & 15 & 20 & 5 & 10 & 15 & 20 \\
        \midrule
        DeepSeek R1 32B & 21.9 & 29.4 & 34.1 & 37.2 & 2.6 & 5.0 & 7.1 & 9.4 & 3.1 & 6.9 & 10.7 & 13.7 \\
        DeepSeek R1 7B & 34.0 & 45.1 & 52.7 & 58.1 & 5.6 & 8.9 & 13.5 & 13.0 & 9.8 & 20.1 & 26.6 & 33.2 \\
        DeepSeek R1 8B & 29.8 & 38.6 & 43.2 & 46.4 & 2.4 & 4.5 & 7.1 & 10.1 & 6.7 & 14.8 & 19.4 & 20.6 \\
        LLaMA 3.1 8B & 30.4 & 41.9 & 46.9 & 49.6 & 4.2 & 5.3 & 6.8 & 7.0 & 6.7 & 15.7 & 20.7 & 24.1 \\
        LLaMA 4 Scout & 22.4 & 38.3 & 51.5 & 61.6 & 2.9 & 10.5 & 24.3 & 33.8 & 3.2 & 12.9 & 24.5 & 36.6 \\
        GPT OSS 120B & 15.5 & 21.9 & 24.9 & 27.3 & 4.4 & 5.4 & 7.5 & 10.2 & 1.0 & 3.4 & 5.9 & 6.7 \\
        GPT-OSS 20B & 17.7 & 24.3 & 28.3 & 30.9 & 5.9 & 7.8 & 10.2 & 11.6 & 2.1 & 5.6 & 6.7 & 9.4 \\
        LLaMA 3.1 70B & 22.7 & 31.7 & 36.7 & 40.2 & 4.6 & 5.3 & 7.3 & 10.7 & 3.1 & 9.5 & 12.2 & 14.9 \\
        \bottomrule
    \end{tabular}
    \caption{Permutation sensitivity (PS), positional bias (PB), and strong positional bias (sPB), all expressed as percentages, across $k \in \{5, 10, 15, 20\}$. Confidence intervals of sPB metric reported in \autoref{app:spb_ci}.}
    \label{tab:metrics_topk}
\end{table}

\autoref{tab:metrics_temperature} reports the same metrics across decoding temperatures $t \in \{0.0, 0.2, 0.4, 0.6\}$ at top-$5$. In contrast to the $k$-dependence above, all three metrics remain remarkably stable as temperature increases: PS grows only mildly (e.g., LLaMA 3.1 8B rises from $30.4\%$ at $t=0.0$ to $33.0\%$ at $t=0.6$), and sPB follows a similarly modest upward trend, consistent with the added decoding stochasticity introducing a small amount of additional, content-driven variability in the recommended item. Notably, PB does not exhibit a consistent upward trend with temperature; for several models it remains essentially flat or even decreases slightly (e.g., LLaMA 3.1 8B decreases from $4.2\%$ at $t=0.0$ to $3.0\%$ at $t=0.6$, and DeepSeek R1 32B fluctuates within a narrow band of $2.5$--$2.8\%$). This indicates that the structural, position-driven component of the bias captured by PB is largely insensitive to sampling temperature, while the mild increase observed in PS and sPB is better attributed to the additional stochasticity in item selection introduced by non-zero temperature itself, rather than to a strengthening of positional bias. This observation reinforces the conclusion already drawn from the POB-UNBP results in \autoref{tab:pob_unbp_temp}: positional bias is a structural property of the model that persists, largely unchanged in magnitude, under stochastic decoding.

\begin{table}
    \centering
    \begin{tabular}{c|cccc|cccc|cccc}
        \toprule
        \multirow{2}{*}{Model name} & \multicolumn{4}{c}{PS} & \multicolumn{4}{c}{PB} & \multicolumn{4}{c}{sPB} \\
         & 0.0 & 0.2 & 0.4 & 0.6 & 0.0 & 0.2 & 0.4 & 0.6 & 0.0 & 0.2 & 0.4 & 0.6 \\
        \midrule
        DeepSeek R1 32B & 21.9 & 22.7 & 22.9 & 23.4 & 2.6 & 2.6 & 2.5 & 2.8 & 3.1 & 3.3 & 3.4 & 3.5 \\
        DeepSeek R1 7B & 34.0 & 34.6 & 35.3 & 35.9 & 5.6 & 5.9 & 5.9 & 6.0 & 9.8 & 10.0 & 10.4 & 10.4 \\
        DeepSeek R1 8B & 29.8 & 30.9 & 31.5 & 32.6 & 2.4 & 2.1 & 2.7 & 2.8 & 6.7 & 6.9 & 8.7 & 9.9 \\
        LLaMA 3.1 8B & 30.4 & 31.4 & 32.3 & 33.0 & 4.2 & 4.0 & 3.6 & 3.0 & 6.7 & 7.1 & 9.3 & 9.4 \\
        LLaMA 4 Scout & 22.4 & 22.6 & 22.6 & 22.9 & 2.9 & 3.1 & 3.1 & 3.0 & 3.2 & 3.2 & 3.2 & 3.3 \\
        GPT OSS 120B & 15.5 & 15.6 & 15.7 & 16.0 & 4.4 & 4.4 & 4.6 & 4.4 & 1.0 & 1.0 & 1.2 & 1.1 \\
        GPT-OSS 20B & 17.7 & 17.9 & 18.6 & 18.6 & 5.9 & 6.5 & 6.3 & 6.2 & 2.1 & 2.5 & 2.6 & 2.7 \\
        LLaMA 3.1 70B & 22.7 & 22.8 & 23.7 & 24.2 & 4.6 & 4.7 & 4.9 & 5.6 & 3.1 & 3.4 & 3.5 & 3.5 \\
        \bottomrule
    \end{tabular}
    \caption{Permutation sensitivity (PS), positional bias (PB), and strong positional bias (sPB), all expressed as percentages, for top-$5$ across decoding temperatures $t \in \{0.0, 0.2, 0.4, 0.6\}$. Confidence intervals of sPB metric reported in \autoref{app:spb_ci}.}
    \label{tab:metrics_temperature}
\end{table}

Since the relevance set is estimated from the selection frequencies of the model itself, sPB also captures recommendations outside $\mathcal{R}(X)$ that are not caused by the order of the items. These arise from stochastic decoding, that is, the randomness of token sampling at non-zero temperature, which can change the recommended item even if the prompt is identical, and from the estimation of the relevance set from a limited number of permutations. Together, they define a noise floor, the level of sPB observed without any change in item order. \autoref{app:sampling_vs_order} estimates this floor by repeating the generation for a fixed order at $t = 0.6$: it lies between $0.9\%$ and $9.1\%$, comparable to the sPB obtained with permutations, and permutations add between $0.3$ and $1.1$ percentage points to it for six of the eight models. Hence, absolute sPB values should be read as an upper bound on the strong component of positional bias, while its trends across $k$ and across models are more informative.

% \textcolor{red}{ANDRES TODO: Hacer una tabla con los p-valores de las dos metodologías, para cada uno de los LLMs. Quiero ver qué p-valor sale para deep seek. Draft the exact p-value table }

\section{Conclusion and future work}
\label{sec:conclusion_futurework}

This paper investigated positional bias in top-1 and single-turn LLM-based RS with the goal of determining whether, how, and to what extent the order of candidate items influences recommendation outcomes. Addressing \textbf{RQ1}, our empirical analysis provides statistical evidence that the output selection of LLM-based RS depends on the position at which candidate items are presented in the input prompt. Through extensive permutation-based evaluations, we show that recommendation behavior is not invariant to item ordering, even when item relevance and retrieval are held constant. Addressing \textbf{RQ2}, we demonstrate that this dependence is statistically significant under formal hypothesis testing, with all evaluated models rejecting the null hypotheses of positional independence and uniformity using two complementary methodologies. Finally, with respect to \textbf{RQ3}, our results indicate that the magnitude and structure of positional bias vary across models, even within the same family, such as LLaMA 3.1 8B and 70B, suggesting that positional bias is influenced by design and training choices rather than being purely incidental. By increasing the number of presented items in the prompt, the positional bias is observed to generally increase across models. This idea is supported by how LLMs attend to different regions of the input prompt, where information appearing at the beginning or end of the prompt is often more influential for the model, regardless of the semantic information distributed across the prompt~\cite{liu-etal-2024-lost}. Furthermore, we distinguish two components of positional bias: soft positional bias, where the position determines which of several comparably relevant candidates is recommended, and strong positional bias, where the position displaces the recommendation towards clearly less relevant items. While the proposed statistical tests detect the overall positional bias, which comprises both components, the proposed metrics quantify the overall positional bias and its strong component, and our counterfactual analyses show that strong positional bias persists when controlling for the identifiers assigned to the items, and for six of the eight evaluated models when controlling for stochastic decoding.

These findings should be interpreted within the scope of our experimental setting. Our evaluation covers single-turn product recommendation, combining user queries generated from templates with real components: a real product catalog (Amazon ML Challenge 2023), a pre-trained embedding-based retriever, and a single fixed prompt template. The heterogeneity of the real product titles and descriptions makes the task non-trivial; for instance, the titles generated by the models do not always match the original ones exactly. However, the template-generated queries encode a limited number of user preferences, and other domains, query formulations or prompt designs may yield different magnitudes and patterns of positional bias. Moreover, we focus on top-1 recommendation, a particularly demanding setting for positional bias analysis, since the relevance the model assigns to the candidates can only be inferred from a single selected item rather than from a full ranking. Within this scope, the main contribution of this work is methodological: a statistically grounded framework to detect overall positional bias in top-1 LLM-based RS, together with metrics to quantify it and the distinction between soft and strong positional bias, quantifying the latter, which is applicable beyond the evaluated setting, as discussed below.

Beyond its methodological contributions, this work has direct implications for the design, evaluation, and governance of RS based on LLM in practice. In operational settings, positional bias can undermine the reliability of A/B testing, distort offline evaluation metrics, and introduce unintended variability in user-facing recommendations. From a governance perspective, such biases complicate auditability and accountability, as system behavior may change solely due to prompt construction choices rather than genuine differences in item relevance. For practitioners deploying LLM-based RS, our findings highlight the importance of treating prompt order as a first-class design variable rather than an implementation detail.

Based on these findings, we outline several practical recommendations for organizations deploying LLM-based RS. First, evaluation pipelines should explicitly incorporate randomized or counterbalanced item ordering to detect and monitor positional bias. Second, prompt construction should be standardized or controlled to ensure consistency between deployments and experiments. Third, bias audits should include permutation-based tests, such as those proposed in this work, as part of routine model validation and monitoring processes. Finally, when positional bias is detected, mitigation strategies such as order randomization, ensemble prompting, or post-generation re-ranking should be considered to reduce its impact on end-user outcomes. Auditing positional bias is inexpensive: in our simulation study, POB-UNBP reaches $80\%$ power with fewer than $200$ queries and $20$ permutations per query when the item placed first is selected with an extra probability of $0.03$, which corresponds to $\mathrm{PB} \approx 2.4$ for $k=5$, comparable to the weakest bias among the evaluated models (\Cref{tab:min_samples}). Since the simulation ignores the variability of relevance across queries, real audits may require more queries. Thus, audits can be repeated whenever the model or the prompt template changes, also to verify position-agnostic templates that avoid ordinal cues in the item identifiers. At inference time, randomizing the order of the candidates prevents items from being systematically favored by their retrieval rank at no additional cost, and, when the model exposes probabilities over the candidate identifiers, the estimated distribution $p(\cdot)$ can serve as a prior for post-hoc recalibration~\cite{zheng_large_2024}.

\noindent Regarding the extensibility of the proposed framework and future research lines, both POB-ICM and POB-UNBP are designed to operate on any black-box system that (i) receives an indexed list of candidate products and (ii) returns a recommendation expressed as a product position within that list. This abstraction is intentionally general and does not depend on the internal architecture of the LLM, the retrieval mechanism, or the number of conversational turns. Concretely:
\begin{itemize}
    \item \textbf{Multi-turn conversational RS.} The framework can be applied to multi-turn settings provided that the recommendation at each turn is expressed as a positional selection from a fixed candidate set. The permutation-based evaluation can then be applied independently at each turn, treating each turn as a self-contained Top-1 instance, although confirming that our empirical findings carry over to multi-turn interactions requires a new set of experiments, which is outside the scope of this work.
    \item \textbf{Top-$N$ recommendation and re-ranking.} These settings are less straightforward. When the system returns an ordered list of $N > 1$ items, or reorders the full candidate set, the notion of positional bias must be generalized to account for how the relative ranking of selected items is affected by input ordering. This requires adapting the contingency-based and uniformity-based null models to joint distributions over ranked outputs, which we leave as a direction for future work.
    \item \textbf{Prompt sensitivity analysis.} We use a single fixed prompt template, not derived from a systematic prompt engineering phase. This does not affect the validity of the proposed detection methodologies, but the magnitude and pattern of the measured bias may depend on the prompt design, so evaluating how different prompt formulations modulate positional bias remains an important direction for future work.
    \item \textbf{Other domains and richer user queries.} Our experiments are restricted to product recommendation, although the proposed methodologies apply without modification to other domains, such as news, music or long-document recommendation. Since the position preferences are stable across decoding temperatures and, for several models, across candidate-set sizes, we expect the differences between models to persist, although their magnitude may vary: longer item descriptions may accentuate effects such as lost in the middle~\cite{liu-etal-2024-lost}, and queries encoding more user preferences may shift the balance between soft and strong positional bias.
    \item \textbf{Independent detection of soft and strong positional bias.} The proposed statistical tests detect the overall positional bias, while strong positional bias is currently quantified through a metric based on relevance sets estimated from the selection frequencies of the model itself. A natural research line is to develop statistical methodologies that detect each component independently, with formal guarantees and robust to the noise introduced by stochastic decoding and by the estimation of the relevance sets.
    \item \textbf{Bias mitigation guided by detection metrics.} The proposed metrics can guide mitigation, either by adding PB as a regularization term when fine-tuning the recommender, as related work does in reranking settings~\cite{debiasFirst}, or, without access to the model weights, by optimizing the prompt template with a genetic algorithm whose fitness is the negated sPB on random batches of queries. Such an optimization should preserve the recommendation task, e.g., by also penalizing PB, since a purely position-driven prompt could also lower sPB, and be validated on held-out queries. More broadly, our framework provides a principled protocol to benchmark any mitigation strategy.
\end{itemize}

Future work should therefore further extend positional bias analysis to top-$N$, multi-turn conversational and other recommendation domains, incorporate explicit quality and utility metrics, and explore the interaction between positional bias and other structural biases such as popularity or exposure bias. Finally, the development of principled mitigation techniques and the integration of positional bias diagnostics into production-scale evaluation frameworks remain important directions to advance the robustness and trustworthiness of LLM-based information systems.

\subsubsection*{Acknowledgments} This publication is part of the Project “Ethical, Responsible and General Purpose Artificial Intelligence: Applications In Risk Scenarios” (IAFER) Exp.:TSI-100927-2023-1 funded through the Creation of university-industry research programs (Enia Programs), aimed at the research and development of artificial intelligence, for its dissemination and education within the framework of the Recovery, Transformation and Resilience Plan from the European Union Next Generation EU through the Ministry for Digital Transformation and the Civil Service.

This publication used computational resources from the project "Thematic Center on Mountain Ecosystem \& Remote sensing, Deep learning-AI e-Services University of Granada-Sierra Nevada" (LifeWatch-2019-10-UGR-01), which has been co-funded by the Ministry of Science and Innovation through the FEDER funds from the Spanish Pluriregional Operational Program 2014-2020 (POPE), LifeWatch-ERIC action line.

\bibliographystyle{splncs04}
\bibliography{biblatex}

@book{agresti2013categorical,
  title     = {Categorical Data Analysis},
  author    = {Agresti, Alan},
  edition   = {3},
  year      = {2013},
  publisher = {Wiley},
  address   = {Hoboken, NJ, USA}
}

@article{chen_bias_2023,
	title = {Bias and {Debias} in {Recommender} {System}: {A} {Survey} and {Future} {Directions}},
	volume = {41},
	issn = {1046-8188},
	
	number = {3},
	journal = {ACM Trans. Inf. Syst.},
	author = {Chen, Jiawei and Dong, Hande and Wang, Xiang and Feng, Fuli and Wang, Meng and He, Xiangnan},
	month = feb,
	year = {2023},
	note = {Place: New York, NY, USA
Publisher: Association for Computing Machinery},
}

@inproceedings{collins_position_2018,
	address = {Cham},
	title = {Position {Bias} in {Recommender} {Systems} for {Digital} {Libraries}},
	isbn = {978-3-319-78105-1},
	booktitle = {Transforming {Digital} {Worlds}},
	publisher = {Springer International Publishing},
	author = {Collins, Andrew and Tkaczyk, Dominika and Aizawa, Akiko and Beel, Joeran},
	editor = {Chowdhury, Gobinda and McLeod, Julie and Gillet, Val and Willett, Peter},
	year = {2018},
	pages = {335--344},
}

@article{deepseekr1,
  author  = {Guo, Daya and Yang, Dejian and Zhang, Hao and et al.},
  title   = {DeepSeek-R1 Incentivizes Reasoning in LLMs through Reinforcement Learning},
  journal = {Nature},
  year    = {2025},
  volume  = {645},
  pages   = {633--638},
  doi     = {10.1038/s41586-025-09422-z}
}

@inproceedings{deldjoo2024review,
  title={A review of modern recommender systems using generative models (gen-recsys)},
  author={Deldjoo, Yashar and He, Zhankui and McAuley, Julian and Korikov, Anton and Sanner, Scott and Ramisa, Arnau and Vidal, Ren{\'e} and Sathiamoorthy, Maheswaran and Kasirzadeh, Atoosa and Milano, Silvia},
  booktitle={Proceedings of the 30th ACM SIGKDD conference on Knowledge Discovery and Data Mining},
  pages={6448--6458},
  year={2024}
}

@article{gallegos_bias_2024,
    title = "Bias and Fairness in Large Language Models: A Survey",
    author = "Gallegos, Isabel O.  and
      Rossi, Ryan A.  and
      Barrow, Joe  and
      Tanjim, Md Mehrab  and
      Kim, Sungchul  and
      Dernoncourt, Franck  and
      Yu, Tong  and
      Zhang, Ruiyi  and
      Ahmed, Nesreen K.",
    journal = "Computational Linguistics",
    volume = "50",
    number = "3",
    month = sep,
    year = "2024",
    address = "Cambridge, MA",
    publisher = "MIT Press",
    url = "https://aclanthology.org/2024.cl-3.8/",
    doi = "10.1162/coli_a_00524",
    pages = "1097--1179"
}

@book{good2013permutation,
  title     = {Permutation Tests: A Practical Guide to Resampling Methods for Testing Hypotheses},
  author    = {Good, Phillip},
  edition   = {3},
  year      = {2013},
  publisher = {Springer},
  address   = {New York, NY, USA}
}

@misc{gptoss,
      title={gpt-oss-120b \& gpt-oss-20b Model Card}, 
      author={OpenAI},
      year={2025},
      eprint={2508.10925},
      archivePrefix={arXiv},
      primaryClass={cs.CL},
      url={https://arxiv.org/abs/2508.10925}, 
}

@misc{grattafiori_llama_2024,
	title = {The {Llama} 3 {Herd} of {Models}},
	
	urldate = {2025-07-17},
	publisher = {arXiv},
	author = {Grattafiori, Aaron and et al.},
	month = nov,
	year = {2024},
	note = {arXiv:2407.21783 [cs]},
}

@inproceedings{hou_large_2024,
	address = {Berlin, Heidelberg},
	title = {Large {Language} {Models} are {Zero}-{Shot} {Rankers} for {Recommender} {Systems}},
	isbn = {978-3-031-56059-0},
	booktitle = {Advances in {Information} {Retrieval}: 46th {European} {Conference} on {Information} {Retrieval}, {ECIR} 2024, {Glasgow}, {UK}, {March} 24–28, 2024, {Proceedings}, {Part} {II}},
	publisher = {Springer-Verlag},
	author = {Hou, Yupeng and Zhang, Junjie and Lin, Zihan and Lu, Hongyu and Xie, Ruobing and McAuley, Julian and Zhao, Wayne Xin},
	month = mar,
	year = {2024},
	pages = {364--381},
}

@inproceedings{li-etal-2024-large,
    title = "Large Language Models for Generative Recommendation: A Survey and Visionary Discussions",
    author = "Li, Lei  and
      Zhang, Yongfeng  and
      Liu, Dugang  and
      Chen, Li",
    editor = "Calzolari, Nicoletta  and
      Kan, Min-Yen  and
      Hoste, Veronique  and
      Lenci, Alessandro  and
      Sakti, Sakriani  and
      Xue, Nianwen",
    booktitle = "Proceedings of the 2024 Joint International Conference on Computational Linguistics, Language Resources and Evaluation (LREC-COLING 2024)",
    month = may,
    year = "2024",
    address = "Torino, Italia",
    publisher = "ELRA and ICCL",
    pages = "10146--10159"
}

@misc{
    meta2024llama4, 
    title={Introducing LLaMA 4: Advancing Multimodal Intelligence}, 
    author={Meta AI}, 
    year={2024}, 
    url={https://ai.meta.com/blog/llama-4-multimodal-intelligence/}
}

@article{navigli2023biases,
  title={Biases in large language models: origins, inventory, and discussion},
  author={Navigli, Roberto and Conia, Simone and Ross, Bj{\"o}rn},
  journal={ACM Journal of Data and Information Quality},
  volume={15},
  number={2},
  pages={1--21},
  year={2023},
  publisher={ACM New York, NY}
}

@inproceedings{pezeshkpour-2024,
    title = "Large Language Models Sensitivity to The Order of Options in Multiple-Choice Questions",
    author = "Pezeshkpour, Pouya  and
      Hruschka, Estevam",
    editor = "Duh, Kevin  and
      Gomez, Helena  and
      Bethard, Steven",
    booktitle = "Findings of the Association for Computational Linguistics: NAACL 2024",
    month = jun,
    year = "2024",
    address = "Mexico City, Mexico",
    publisher = "Association for Computational Linguistics",
    pages = "2006--2017",
}

@article{raiaan_review_2024,
	title = {A {Review} on {Large} {Language} {Models}: {Architectures}, {Applications}, {Taxonomies}, {Open} {Issues} and {Challenges}},
	volume = {12},
	issn = {2169-3536},
	shorttitle = {A {Review} on {Large} {Language} {Models}},
	
	journal = {IEEE Access},
	author = {Raiaan, Mohaimenul Azam Khan and Mukta, Md. Saddam Hossain and Fatema, Kaniz and Fahad, Nur Mohammad and Sakib, Sadman and Mim, Most Marufatul Jannat and Ahmad, Jubaer and Ali, Mohammed Eunus and Azam, Sami},
	year = {2024},
	pages = {26839--26874},
}

@article{storey2025genai,
  title   = {Generative Artificial Intelligence: Evolving Technology, Growing Societal Impact, and Opportunities for Information Systems Research},
  author  = {Storey, Veda C. and Yue, Wei T. and Zhao, J. Leon and Lukyanenko, Roman and others},
  journal = {Information Systems Frontiers},
  volume  = {27},
  pages   = {2081--2102},
  year    = {2025},
  doi     = {10.1007/s10796-025-10581-7}
}

@misc{wang_zero-shot_2023,
	title = {Zero-{Shot} {Next}-{Item} {Recommendation} using {Large} {Pretrained} {Language} {Models}},
	
	urldate = {2025-06-30},
	publisher = {arXiv},
	author = {Wang, Lei and Lim, Ee-Peng},
	month = apr,
	year = {2023},
	note = {arXiv:2304.03153 [cs]},
}

@article{wu2024survey,
  title={A survey on large language models for recommendation},
  author={Wu, Likang and Zheng, Zhi and Qiu, Zhaopeng and Wang, Hao and Gu, Hongchao and Shen, Tingjia and Qin, Chuan and Zhu, Chen and Zhu, Hengshu and Liu, Qi and others},
  journal={World Wide Web},
  volume={27},
  number={5},
  pages={60},
  year={2024},
  publisher={Springer}
}

@inproceedings{xiao_c-pack_2024,
	address = {New York, NY, USA},
	series = {{SIGIR} '24},
	title = {C-{Pack}: {Packed} {Resources} {For} {General} {Chinese} {Embeddings}},
	isbn = {979-8-4007-0431-4},
	shorttitle = {C-{Pack}},
	booktitle = {Proceedings of the 47th {International} {ACM} {SIGIR} {Conference} on {Research} and {Development} in {Information} {Retrieval}},
	publisher = {Association for Computing Machinery},
	author = {Xiao, Shitao and Liu, Zheng and Zhang, Peitian and Muennighoff, Niklas and Lian, Defu and Nie, Jian-Yun},
	month = jul,
	year = {2024},
	pages = {641--649},
}

@inproceedings{yang-etal-2025-option,
    title = "Option Symbol Matters: Investigating and Mitigating Multiple-Choice Option Symbol Bias of Large Language Models",
    author = "Yang, Zhen  and
      Jian, Ping  and
      Li, Chengzhi",
    editor = "Chiruzzo, Luis  and
      Ritter, Alan  and
      Wang, Lu",
    booktitle = "Proceedings of the 2025 Conference of the Nations of the Americas Chapter of the Association for Computational Linguistics: Human Language Technologies (Volume 1: Long Papers)",
    month = apr,
    year = "2025",
    address = "Albuquerque, New Mexico",
    publisher = "Association for Computational Linguistics",
    pages = "1902--1917",
    ISBN = "979-8-89176-189-6"
}

@inproceedings{zhang_can_2024,
	address = {Miami, Florida, USA},
	title = {Can {We} {Instruct} {LLMs} to {Compensate} for {Position} {Bias}?},
	
	urldate = {2025-07-17},
	booktitle = {Findings of the {Association} for {Computational} {Linguistics}: {EMNLP} 2024},
	publisher = {Association for Computational Linguistics},
	author = {Zhang, Meiru and Meng, Zaiqiao and Collier, Nigel},
	editor = {Al-Onaizan, Yaser and Bansal, Mohit and Chen, Yun-Nung},
	month = nov,
	year = {2024},
	pages = {12545--12556},
}

@misc{zheng_large_2024,
	title = {Large {Language} {Models} {Are} {Not} {Robust} {Multiple} {Choice} {Selectors}},
	booktitle = {The {Twelfth} {International} {Conference} on {Learning} {Representations}},
	author = {Zheng, Chujie and Zhou, Hao and Meng, Fandong and Zhou, Jie and Huang, Minlie},
	year = {2024}
}

@article{liu-etal-2024-lost,
    title = "Lost in the Middle: How Language Models Use Long Contexts",
    author = "Liu, Nelson F.  and
      Lin, Kevin  and
      Hewitt, John  and
      Paranjape, Ashwin  and
      Bevilacqua, Michele  and
      Petroni, Fabio  and
      Liang, Percy",
    journal = "Transactions of the Association for Computational Linguistics",
    volume = "12",
    year = "2024",
    address = "Cambridge, MA",
    publisher = "MIT Press",
    url = "https://aclanthology.org/2024.tacl-1.9/",
    doi = "10.1162/tacl_a_00638",
    pages = "157--173"
}

@inproceedings{hsieh-etal-2024-found,
    title = "Found in the middle: Calibrating Positional Attention Bias Improves Long Context Utilization",
    author = "Hsieh, Cheng-Yu  and
      Chuang, Yung-Sung  and
      Li, Chun-Liang  and
      Wang, Zifeng  and
      Le, Long  and
      Kumar, Abhishek  and
      Glass, James  and
      Ratner, Alexander  and
      Lee, Chen-Yu  and
      Krishna, Ranjay  and
      Pfister, Tomas",
    editor = "Ku, Lun-Wei  and
      Martins, Andre  and
      Srikumar, Vivek",
    booktitle = "Findings of the Association for Computational Linguistics: ACL 2024",
    month = aug,
    year = "2024",
    address = "Bangkok, Thailand",
    publisher = "Association for Computational Linguistics",
    url = "https://aclanthology.org/2024.findings-acl.890/",
    doi = "10.18653/v1/2024.findings-acl.890",
    pages = "14982--14995"
}

@inproceedings{zheng-etal-2023-judging,
author = {Zheng, Lianmin and Chiang, Wei-Lin and Sheng, Ying and Zhuang, Siyuan and Wu, Zhanghao and Zhuang, Yonghao and Lin, Zi and Li, Zhuohan and Li, Dacheng and Xing, Eric P. and Zhang, Hao and Gonzalez, Joseph E. and Stoica, Ion},
title = {Judging LLM-as-a-judge with MT-bench and Chatbot Arena},
year = {2023},
publisher = {Curran Associates Inc.},
address = {Red Hook, NY, USA},
booktitle = {Proceedings of the 37th International Conference on Neural Information Processing Systems},
articleno = {2020},
numpages = {29},
location = {New Orleans, LA, USA},
series = {NIPS '23}
}

@inproceedings{shi-etal-2025-judging,
    title = "Judging the Judges: A Systematic Study of Position Bias in {LLM}-as-a-Judge",
    author = "Shi, Lin  and
      Ma, Chiyu  and
      Liang, Wenhua  and
      Diao, Xingjian  and
      Ma, Weicheng  and
      Vosoughi, Soroush",
    editor = "Inui, Kentaro  and
      Sakti, Sakriani  and
      Wang, Haofen  and
      Wong, Derek F.  and
      Bhattacharyya, Pushpak  and
      Banerjee, Biplab  and
      Ekbal, Asif  and
      Chakraborty, Tanmoy  and
      Singh, Dhirendra Pratap",
    booktitle = "Proceedings of the 14th International Joint Conference on Natural Language Processing and the 4th Conference of the Asia-Pacific Chapter of the Association for Computational Linguistics",
    month = dec,
    year = "2025",
    address = "Mumbai, India",
    publisher = "The Asian Federation of Natural Language Processing and The Association for Computational Linguistics",
    url = "https://aclanthology.org/2025.ijcnlp-long.18/",
    doi = "10.18653/v1/2025.ijcnlp-long.18",
    pages = "292--314",
    ISBN = "979-8-89176-298-5",
}

@inproceedings{liusie-etal-2024-teacher,
    title = "Teacher-Student Training for Debiasing: General Permutation Debiasing for Large Language Models",
    author = "Liusie, Adian  and
      Fathullah, Yassir  and
      Gales, Mark",
    editor = "Ku, Lun-Wei  and
      Martins, Andre  and
      Srikumar, Vivek",
    booktitle = "Findings of the Association for Computational Linguistics: ACL 2024",
    month = aug,
    year = "2024",
    address = "Bangkok, Thailand",
    publisher = "Association for Computational Linguistics",
    url = "https://aclanthology.org/2024.findings-acl.81/",
    doi = "10.18653/v1/2024.findings-acl.81",
    pages = "1376--1387",
}

@inproceedings{wei-etal-2024-unveiling,
    title = "Unveiling Selection Biases: Exploring Order and Token Sensitivity in Large Language Models",
    author = "Wei, Sheng-Lun  and
      Wu, Cheng-Kuang  and
      Huang, Hen-Hsen  and
      Chen, Hsin-Hsi",
    editor = "Ku, Lun-Wei  and
      Martins, Andre  and
      Srikumar, Vivek",
    booktitle = "Findings of the Association for Computational Linguistics: ACL 2024",
    month = aug,
    year = "2024",
    address = "Bangkok, Thailand",
    publisher = "Association for Computational Linguistics",
    url = "https://aclanthology.org/2024.findings-acl.333/",
    doi = "10.18653/v1/2024.findings-acl.333",
    pages = "5598--5621"
}

@InProceedings{debiasFirst,
    author="Qiao, Jingfen
    and Huang, Jin
    and Ma, Xinyu
    and Wang, Shuaiqiang
    and Yin, Dawei
    and Kanoulas, Evangelos
    and Yates, Andrew",
    editor="Campos, Ricardo
    and Jatowt, Adam
    and Lan, Yanyan
    and Aliannejadi, Mohammad
    and Bauer, Christine
    and MacAvaney, Sean
    and Anand, Avishek
    and Ren, Zhaochun
    and Verberne, Suzan
    and Bai, Nan
    and Mansoury, Masoud",
    title="LLM-Based Listwise Reranking Under the Effect of Positional Bias",
    booktitle="Advances in Information Retrieval",
    year="2026",
    publisher="Springer Nature Switzerland",
    address="Cham",
    pages="131--146",
    isbn="978-3-032-21289-4"
}

% I was making a mistake, computing the POB-ICM over the contingency matrix of raw predicted positions. It should be applied over the item positions before shuffling (I didn't apply \sigma)

\appendix

\section{Proof of \autoref{thm:stochastic} and \autoref{prop:lower_bound}}
\label{app:proof_theorem}

\begin{customthm}{\ref{thm:stochastic}}
    Under $\varepsilon$-absence of positional bias, we have $\mathrm{TV}(p(\cdot), U_k) \le \varepsilon$, where $U_k$ denotes the uniform distribution over $\{1, 2, \ldots, k\}$. In fact, for all instances $x$, $\mathrm{TV}\!\left(p(\cdot \mid x),\; U_k\right) \leq \varepsilon$.
\end{customthm}

\begin{proof}[Proof of \autoref{thm:stochastic}]
    Fix an instance $X = x$ and a reference permutation $\varphi_0 \in S_k$ (e.g., the identity). Write $r_\sigma(j) = \mathrm{Pr}(P_j \mid x, \sigma)$ for the distribution over items under the ordering $\sigma$. Recall that if the item at position $i$ in the prompt is picked under ordering $\sigma$, then the item selected is $P_{\sigma(i)}$. Thus, for any $\sigma \in S_k$, we have
    \begin{equation*}
        p(i \mid x) = \mathbb{E}_{\sigma \sim U(S_k)} \left[ r_\sigma(\sigma(i)) \right].
    \end{equation*}
    Since $\sigma(i)$ is uniformly distributed over $\{1, \ldots, k\}$ when $\sigma \sim U(S_k)$, for the fixed permutation $\varphi_0$ we obtain 
    \begin{equation*}
        \mathbb{E}_{\sigma \sim U(S_k)}\left[r_{\varphi_0}(\sigma(i))\right] = \frac{1}{k} \sum_{j=1}^k r_{\varphi_0}(j) = \frac{1}{k}.
    \end{equation*}
    Therefore, combining the two equations above, we find that
    \begin{equation*}
        p(i \mid x) - \frac{1}{k} = \mathbb{E}_{\sigma \sim U(S_k)} \left[ r_\sigma(\sigma(i)) \right] - \mathbb{E}_{\sigma \sim U(S_k)}\left[r_{\varphi_0}(\sigma(i))\right] = \mathbb{E}_{\sigma} \left[ r_\sigma(\sigma(i)) - r_{\varphi_0}(\sigma(i)) \right],
    \end{equation*}
    and we have
    \begin{equation*}
    \begin{split}
        \sum_{i=1}^k \left| p(i \mid x) - \frac{1}{k} \right|
        &= \sum_{i=1}^k \left| \mathbb{E}_{\sigma} \left[ r_\sigma(\sigma(i)) - r_{\varphi_0}(\sigma(i)) \right] \right| \\
        & \leq \sum_{i=1}^k \mathbb{E}_{\sigma} \left| r_\sigma(\sigma(i)) - r_{\varphi_0}(\sigma(i)) \right|,
    \end{split}
    \end{equation*}
    where we have used the inequality $\lvert \mathbb{E}[Y] \rvert \le \mathbb{E}[ \lvert Y \rvert]$. By using the definition of expectation and reordering the sums, it follows
    \begin{equation*}
    \begin{split}
        \sum_{i=1}^k \mathbb{E}_{\sigma} \left| r_\sigma(\sigma(i)) - r_{\varphi_0}(\sigma(i)) \right|
        & = \sum_{i=1}^k \sum_{j = 1}^k \sum_{\sigma \in S_k : \sigma(i) = j} \frac{1}{k!}  \left| r_\sigma(j) - r_{\varphi_0}(j) \right| \\
        & = \sum_{j = 1}^k \sum_{\sigma \in S_k} \frac{1}{k!}  \left| r_\sigma(j) - r_{\varphi_0}(j) \right| \\
        &= \sum_{\sigma \in S_k} \frac{1}{k!} \sum_{j=1}^k \left| r_\sigma(j) - r_{\varphi_0}(j) \right| \\
        &= \mathbb{E}_{\sigma \sim U(S_k)} \left[ \sum_{j=1}^k \left| r_\sigma(j) - r_{\varphi_0}(j) \right| \right] \\
        &= 2\, \mathbb{E}_{\sigma}\left[\mathrm{TV}(r_\sigma, r_{\varphi_0})\right] \\
        &\leq 2\varepsilon,
    \end{split}
    \end{equation*}
    where the last inequality follows from \autoref{def:eps-absence}. Dividing both sides by $2$ yields 
    \begin{equation*}
        \mathrm{TV}(p(\cdot \vert x),U_k) =  \frac{1}{2}\sum_{i=1}^k \left| p(i \mid x) - \frac{1}{k} \right| \le \frac{1}{2} \sum_{i=1}^k \mathbb{E}_{\sigma} \left| r_\sigma(\sigma(i)) - r_{\varphi_0}(\sigma(i)) \right| \le \varepsilon,
    \end{equation*}
    as we wanted to prove.
    Now we tackle the unconditional distribution $p(i) = \mathbb{E}_X[p(i \mid X)]$. Indeed, by convexity of total variation distance and Jensen's inequality,
    \begin{equation*}
        \mathrm{TV}(p(\cdot), U_k) = \mathrm{TV}\!\left(\mathbb{E}_X[p(\cdot \mid X)],\; U_k\right) \leq \mathbb{E}_X\left[\mathrm{TV}(p(\cdot \mid X), U_k)\right] \leq \varepsilon. \qedhere
    \end{equation*}
\end{proof}

\begin{customprop}{\ref{prop:lower_bound}}
    Fix $k \ge 2$, $\ell \ge 1$ and $\alpha \in (0,1)$, and let $L$ be the bound defined in~\autoref{eq:tv_lower_bound}, computed from $n$ independent instances. Then
    \begin{equation*}
        \liminf_{n \to \infty} \; \Pr\!\left[\, L \le \mathrm{TV}(p(\cdot), U_k) \,\right] \;\ge\; 1 - \alpha .
    \end{equation*}
\end{customprop}
\begin{proof}
    We first show that $\limsup_{n \to \infty} \Pr[L_i > \mu_i] \le \alpha/k$ for every $i \in \{1, 2, \ldots, k\}$. Fix $i$. The variables $Z_1(i), Z_2(i), \ldots$ are independent and identically distributed with mean $\mu_i$, and they are bounded since $\lvert Z_j(i) \rvert \le 1$. We distinguish two cases.
    \begin{itemize}
        \item Case $\operatorname{Var}(Z_1(i)) = 0$. Then $Z_j(i) = \mu_i$ almost surely for every $j$, so that $s_i = 0$ and $L_i = \mu_i$. Hence $\Pr[L_i > \mu_i] = 0$.
        \item Case $\operatorname{Var}(Z_1(i)) > 0$. Let $T_n(i) = \sqrt{n}\,(\bar{Z}(i) - \mu_i)/s_i$, with the convention $T_n(i) = 0$ when $s_i = 0$. By the strong law of large numbers, $s_i$ converges almost surely to $\operatorname{Var}(Z_1(i))^{1/2} > 0$, so $\Pr[s_i = 0] \to 0$. Moreover, by the central limit theorem and Slutsky's theorem, $T_n(i)$ converges in distribution to a standard normal random variable. On the event $s_i > 0$, we have $L_i > \mu_i$ if and only if $T_n(i) > t_{1-\alpha/k,\,n-1}$. Since $t_{1-\alpha/k,\,n-1} \to z_{1-\alpha/k}$ as $n \to \infty$, where $z_\beta$ denotes the $\beta$-quantile of the standard normal distribution, we obtain
        \begin{equation*}
            \lim_{n \to \infty} \Pr\!\left[ L_i > \mu_i \right]
            = \lim_{n \to \infty} \Pr\!\left[ T_n(i) > t_{1-\alpha/k,\,n-1} \right]
            = \frac{\alpha}{k} .
        \end{equation*}
    \end{itemize}
    Now let $E_n$ be the event ``$L_i \le \mu_i$ for all $i$'', and let $E_n^{C}$ be its complement. By the union bound,
    \begin{equation*}
        \limsup_{n \to \infty} \Pr[E_n^{C}] \le \sum_{i=1}^{k} \limsup_{n \to \infty} \Pr[L_i > \mu_i] \le \alpha,
    \end{equation*}
    so $\liminf_{n \to \infty} \Pr[E_n] \ge 1 - \alpha$. If $E_n$ holds, the monotonicity of $x \mapsto x_+$ gives $(L_i)_+ \le (\mu_i)_+$ for every $i$, and therefore
    \begin{equation*}
        L = \sum_{i=1}^{k} (L_i)_+ \;\le\; \sum_{i=1}^{k} (\mu_i)_+ \;=\; \mathrm{TV}(p(\cdot), U_k),
    \end{equation*}
    where the last equality is~\autoref{eq:tv_positive_part}. That is, $E_n$ implies $L \le \mathrm{TV}(p(\cdot), U_k)$, and we conclude that
    \begin{equation*}
        1 - \alpha \le \liminf_{n \to \infty} \Pr[E_n] \le \liminf_{n \to \infty} \Pr\!\left[ L \le \mathrm{TV}(p(\cdot), U_k) \right]. \qedhere
    \end{equation*}
\end{proof}

\section{Simulation study: Type-I error and statistical power}
\label{sec:simulation_study}

To assess the statistical validity and sample efficiency of POB-ICM and POB-UNBP independently of any specific LLM, we run a simulation study on synthetic data with known ground-truth bias. We measure (i) empirical Type-I error under the null hypothesis of no positional bias, and (ii) statistical power as a function of both bias strength and sample size.

For each synthetic query\footnote{In this context, a synthetic query is a numerical index with no associated text, unlike the queries used in \autoref{sec:datasets}.}, we draw a uniformly random, non-empty relevance set $\mathcal{R}(X) \subseteq \{1,\ldots,k\}$. For each sampled permutation $\sigma$, the simulated recommender selects an item uniformly from $\mathcal{R}(X)$ with probability $1-p_{\mathrm{bias}}$, and otherwise (probability $p_{\mathrm{bias}}$) selects the item placed first by $\sigma$. Thus $p_{\mathrm{bias}}=0$ is an exact null model of no positional bias, and $p_{\mathrm{bias}}>0$ injects a known amount of strong positional bias; all experiments below are repeated for $k \in \{5,10,15,20\}$.

\emph{Experiment 1: Type-I error.} We estimate the empirical Type-I error (false positive rate) of both tests, where synthetic datasets with no induced bias are sampled and tested. We fix $p_{\mathrm{bias}}=0$, $n_{\mathrm{queries}}=2000$, $n_\sigma=20$ (matching \autoref{sec:datasets}), and estimate the rejection rate at $\alpha=0.05$ over $1000$ replicates. A well calibrated $p$-value should produce a type-I error equal to the nominal level $\alpha$ on expectation. \Cref{tab:type_i_error} shows POB-UNBP stays well calibrated at every $k$ ($0.039$--$0.053$), while POB-ICM is calibrated at $k=5,10$ but becomes conservative for $k\geq15$ ($0.016$--$0.020$), since most columns of its contingency matrix contain a single observation, which makes the Chi-square statistic nearly constant across the permutation replicates and produces many ties with the observed value.

\emph{Experiment 2: power vs.\ bias strength.} We measure power as a function of bias magnitude. Fixing $n_{\mathrm{queries}}=500$, $n_\sigma=20$, we sweep $p_{\mathrm{bias}} \in (0,0.1)$ and estimate power over $200$ replicates per point. \Cref{fig:power_vs_bias} shows POB-UNBP saturates by $p_{\mathrm{bias}}\approx0.02$ at $k=5$, and at even smaller $p_{\mathrm{bias}}$ for larger $k$, while POB-ICM needs roughly $p_{\mathrm{bias}}\in(0.02,0.06)$ at $k=5$ and fails to gain power at all for $k\geq10$.

\emph{Experiment 3: minimum sample size.} We estimate the minimum number of queries needed for $0.8$ power. Fixing $p_{\mathrm{bias}}=0.03$, we grid-search $n_{\mathrm{queries}}\in[25,2500]$, $n_\sigma\in[2,80]$ (60 replicates/cell), interpolating the smallest $n_{\mathrm{queries}}$ reaching $80\%$ power per $n_\sigma$. \Cref{tab:min_samples} shows POB-UNBP requires approximately six times fewer queries than POB-ICM at $k=5$ (e.g., $175$ vs.\ $1{,}068$ at $n_\sigma=20$), and this requirement decreases further with $k$ ($65$ queries at $k=20$); POB-ICM never reaches $80\%$ power for $k\geq10$, so it is reported only at $k=5$.

Taken together, both tests are valid (Type-I error near nominal), and POB-UNBP is consistently more sample-efficient than POB-ICM, an advantage that grows rather than shrinks with $k$. This motivates relying on POB-UNBP for the candidate-set-size analysis in \autoref{sec:results_unbp} and reporting POB-ICM only at $k=5$ (\autoref{tab:tests_results}); it also confirms that the $n_{\mathrm{queries}}=2000$, $n_\sigma=20$ operating point used throughout \autoref{sec:results} affords high power for both tests.

\begin{figure}
\centering

\begin{minipage}[b]{0.48\columnwidth}
    \centering
    
    \begin{tabular}{c|cccc}
        \toprule
        Method & $k=5$ & $k=10$ & $k=15$ & $k=20$ \\
        \midrule
        POB-ICM & 0.055 & 0.050 & 0.016 & 0.020 \\
        POB-UNBP & 0.053 & 0.044 & 0.043 & 0.039 \\
        \bottomrule
    \end{tabular}
    \captionof{table}{Type-I error (false positive rate) in the simulation, $1000$ synthetic dataset samples with no positional bias induced. % COMENTARIO: Explicar más claramente qué es en este contexto el type I error rate en este experimento
    }
    \label{tab:type_i_error}

    \vspace{12mm}

    \resizebox{\linewidth}{!}{%
    \begin{tabular}{r|r|rrrr}
        \toprule
        \multirow{2}{*}{$n_\sigma$} & POB-ICM & \multicolumn{4}{c}{POB-UNBP} \\
         & $k=5$ & $k=5$ & $k=10$ & $k=15$ & $k=20$ \\
        \midrule    2 & $\ge$2500 & 1393 & 707 & 688 & 562 \\
        5 & $\ge$2500 & 659 & 368 & 338 & 225 \\
        10 & 1983 & 344 & 214 & 181 & 168 \\
        20 & 1068 & 175 & 98 & 72 & 65 \\
        40 & 589 & 88 & 46 & 44 & 39 \\
        80 & 246 & 40 & $\le$25 & $\le$25 & $\le$25 \\
        \bottomrule
    \end{tabular}}
    \captionof{table}{Required number of queries to reach a power of $0.8$, conditioned on the number of permutations per query, $n_\sigma$, used in the experiment.}
    \label{tab:min_samples}
\end{minipage}
\hfill
\begin{minipage}[b]{0.48\columnwidth}
    \centering
    \includegraphics[width=\linewidth]{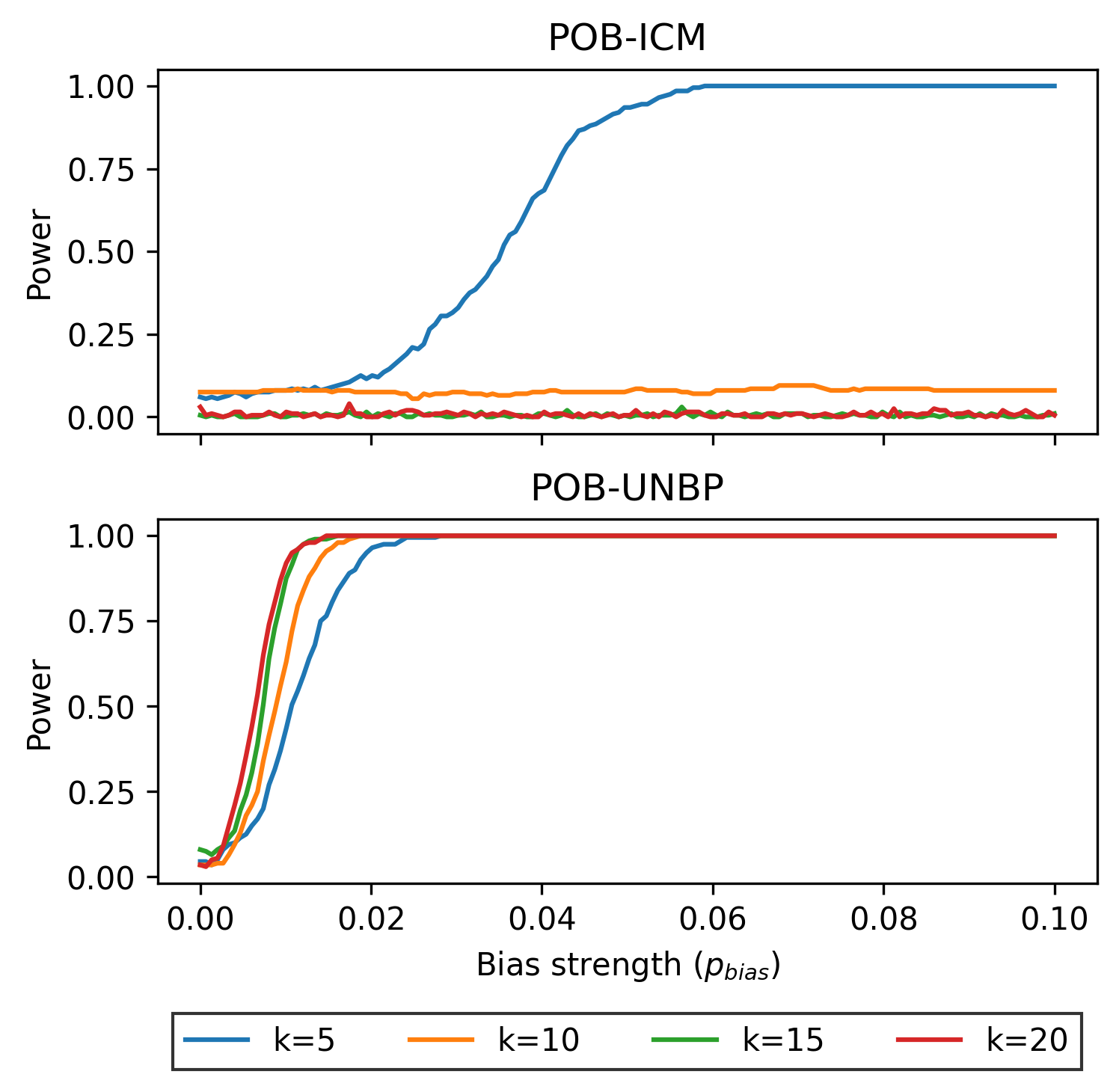}
    \captionof{figure}{Power depending on the bias strength, measured as the probability of switch the decision to the first prompt position.}
    \label{fig:power_vs_bias}
\end{minipage}

\end{figure}

\section{Separating stochastic decoding variability from positional bias}
\label{app:sampling_vs_order}

Under non-zero temperature, the recommended item may change between two evaluations of the same instance either because the candidate items were permuted or because of the randomness of the decoding itself. To disentangle both mechanisms, we evaluate the top-$5$ setting under three counterfactual configurations: (i) \emph{decoding only}, with temperature $t = 0.6$ and the identity permutation for every evaluation; (ii) \emph{permutation only}, with $t = 0$ and uniformly sampled permutations; and (iii) \emph{both}, with $t = 0.6$ and uniformly sampled permutations. Configurations (ii) and (iii) are the evaluations reported in \autoref{sec:results_analysis}.

For every configuration, we compute three quantities on $20$ evaluations per instance. The first is the instance-level strong positional bias, that is, the fraction of evaluations whose recommended item lies outside the relevance set $\mathcal{R}(X)$ (\autoref{eq:strong_positional_bias_metric}). The second is the lower bound $L$ of POB-UNBP (\autoref{eq:tv_lower_bound}), in percentage points, which quantifies the overall positional bias; it is reported only for configurations (ii) and (iii), since under configuration (i) the order of the items is fixed. The third is the intersection over union (IoU) between the relevance set of the instance and the one obtained under configuration (i), averaged over instances. Relevance sets are computed as described in \autoref{sec:metrics}, with the folds and thresholds $\tau$ of configuration (ii) shared by the three configurations, so that the relevance set is defined identically across them.\footnote{Hence, the sPB values of configuration (ii) coincide with \autoref{tab:metrics_topk}, whereas those of configuration (iii) may differ from \autoref{tab:metrics_temperature}, where $\tau$ is calibrated on the evaluations at $t = 0.6$; the differences are largest for DeepSeek R1 8B and LLaMA 3.1 8B, due to the discreteness of the selection frequencies discussed in \autoref{app:spb_ci}.} Confidence intervals for sPB are obtained by percentile bootstrap at the query level, keeping the folds, the thresholds $\tau$ and hence the relevance sets fixed. They therefore reflect the sampling variability of the queries conditionally on the relevance sets, which is the relevant uncertainty when comparing configurations that share them, and are narrower than those of \autoref{app:spb_ci}, where $\tau$ is recalibrated in every resample.

To test whether permutations increase strong positional bias beyond the effect of stochastic decoding, we compare the instance-level sPB of configurations (iii) and (i), which share the same temperature and differ only in the permutations. Since both configurations evaluate the same $2000$ instances, we apply a two-sided paired Wilcoxon signed-rank test, discarding instances with identical values in both configurations, and adjust the resulting $p$-values for the eight models with the Holm procedure, which controls the family-wise error rate at $0.05$.

\begin{table}
    \centering
    
    \resizebox{\linewidth}{!}{%
    \begin{tabular}{c|ccc|cc|cc|c}
        \toprule
        \multirow{2}{*}{Model name}
        & \multicolumn{3}{c|}{sPB}
        & \multicolumn{2}{c|}{$L$}
        & \multicolumn{2}{c|}{IoU of $\mathcal{R}(X)$}
        & \multirow{2}{*}{$p$-value} \\
        & (i) & (ii) & (iii) & (ii) & (iii) & (i), (ii) & (i), (iii) & \\
        \midrule
        DeepSeek R1 32B & 2.9 (2.8, 3.1) & 3.1 (2.9, 3.2) & 3.4 (3.3, 3.6) & 1.0 & 1.1 & 78.9 & 79.8 & \textbf{2.2e-08} \\
        DeepSeek R1 7B & 9.1 (8.8, 9.4) & 9.8 (9.5, 10.1) & 10.2 (9.9, 10.5) & 4.2 & 4.4 & 70.2 & 70.6 & \textbf{7.4e-09} \\
        DeepSeek R1 8B & 6.6 (6.3, 6.8) & 6.7 (6.5, 6.9) & 7.5 (7.3, 7.8) & 0.8 & 1.6 & 73.6 & 73.4 & \textbf{4.3e-11} \\
        LLaMA 3.1 8B & 6.9 (6.7, 7.2) & 6.7 (6.4, 6.9) & 7.5 (7.3, 7.7) & 3.0 & 2.0 & 74.3 & 74.5 & \textbf{1.6e-08} \\
        LLaMA 4 Scout & 2.3 (2.1, 2.5) & 3.2 (3.0, 3.4) & 3.3 (3.1, 3.5) & 0.9 & 1.0 & 76.4 & 76.6 & \textbf{5.1e-19} \\
        GPT OSS 120B & 0.9 (0.9, 1.0) & 1.0 (1.0, 1.1) & 1.1 (1.0, 1.1) & 3.1 & 3.2 & 74.8 & 75.1 & 0.104 \\
        GPT-OSS 20B & 2.5 (2.3, 2.6) & 2.1 (2.0, 2.2) & 2.3 (2.1, 2.4) & 4.7 & 4.9 & 72.8 & 74.2 & 0.194 \\
        LLaMA 3.1 70B & 3.2 (3.0, 3.3) & 3.1 (3.0, 3.3) & 3.5 (3.3, 3.7) & 3.4 & 4.5 & 78.5 & 79.6 & \textbf{0.003} \\
        \bottomrule
    \end{tabular}}
    \caption{Strong positional bias (sPB), POB-UNBP lower bound $L$ of $\mathrm{TV}(p(\cdot), U_k)$ for the configurations with permutations, and mean intersection over union (IoU) of the relevance sets with configuration (i), for top-$5$ under configurations (i) stochastic decoding, (ii) random permutation and (iii) both. sPB and IoU are percentages and $L$ is in percentage points, with $95\%$ confidence intervals for sPB. The last column reports the Holm-adjusted $p$-value comparing the sPB of configurations (iii) and (i); bold values are significant at the $0.05$ level.}
    \label{tab:spb_mechanisms}
\end{table}

\paragraph{Overall positional bias.} Stochastic decoding does not remove the overall positional bias: $L$ remains of similar magnitude with and without temperature, between $0.8$ and $4.7$ percentage points under configuration (ii) and between $1.0$ and $4.9$ under configuration (iii), consistently with \autoref{tab:pob_unbp_temp}.

\paragraph{Strong positional bias beyond decoding noise.} Recommendations outside the relevance set are rare in all configurations, with sPB between $0.9\%$ and $10.2\%$, but they depend on the source of variability. Adding permutations to stochastic decoding, configuration (iii) versus (i), significantly increases sPB for six of the eight models, by $0.3$ to $1.1$ percentage points, the exceptions being the two GPT-OSS models, whose sPB is the lowest among the evaluated models. Permutations alone, configuration (ii), yield an sPB comparable to that of stochastic decoding alone for most models, and notably larger for LLaMA 4 Scout ($3.2\%$ against $2.3\%$) and DeepSeek R1 7B ($9.8\%$ against $9.1\%$). Therefore, part of the strong positional bias observed under stochastic decoding is attributable to the permutation of the candidate items and is not explained by temperature alone.

The GPT-OSS models are a notable exception. Their sPB under permutations does not exceed the level observed under decoding only, and the $95\%$ confidence interval of the paired difference between configurations (iii) and (i) bounds any increase of sPB due to permutations to at most $0.2$ percentage points. Since sPB also captures recommendations outside the relevance set caused by decoding noise and by the estimation of the relevance set from a limited number of evaluations, we find no evidence of strong positional bias in these models beyond this noise level, although the identity permutation used in configuration (i) prevents a fully clean comparison for them. Combined with their significant overall positional bias (\autoref{tab:tests_results}), this indicates that the positional bias of the GPT-OSS models is essentially soft: their marked preference for the first prompt positions acts as a tie-breaker among comparably relevant candidates rather than displacing the recommendation towards clearly less relevant items.

\paragraph{Stability of the relevance set.} The mean intersection over union (IoU) between the relevance sets obtained under decoding only and under permutations ranges from $70\%$ to $80\%$, with almost identical values whether or not temperature is added to the permutations. Hence, the relevance sets largely overlap across sources of variability, and their differences typically involve a single item. Agreement is lowest for DeepSeek R1 7B, the model with the highest sPB, whose recommendations are more dispersed across candidate items, and highest for DeepSeek R1 32B and LLaMA 3.1 70B. Since every relevance set is estimated from only $20$ evaluations, part of these differences reflects estimation noise in the items closest to the threshold.

One caveat applies: configuration (i) fixes the identity permutation, so the product ranked first by the semantic matcher always occupies the first prompt position, which the GPT-OSS models strongly favor (\autoref{fig:bias_analysis}); this couples positional preference with retrieval rank and may explain why the GPT-OSS models show no significant differences.

\section{Separating item identifiers from positional bias}
\label{app:labels_vs_order}

In the prompt of \autoref{fig:recommender_prompt}, each item is identified by its prompt position (``Title 1'', ``Description 1'', etc.), and the model returns this identifier in the field \texttt{article\_number}. Hence, the position of an item and its identifier are coupled, and part of the observed effect could reflect a preference for particular identifiers rather than for prompt positions. To disentangle both, we evaluate the top-$5$ setting with random labels: the identifiers $1, \ldots, k$ are assigned to the prompt positions by a random permutation, drawn independently of the query and of the item permutation, and the label returned by the model is mapped back to its prompt position before the analysis. As in \autoref{app:sampling_vs_order}, we consider three configurations at $t = 0$: (i) \emph{random labels only}, with the identity permutation; (ii) \emph{permutation only}, with positional labels and uniformly sampled permutations, i.e., the evaluation reported in \autoref{sec:results_analysis}; and (iii) \emph{both}, with random labels and uniformly sampled permutations. sPB, $L$ and the relevance sets are computed as in \autoref{app:sampling_vs_order}, with the folds and thresholds of configuration (ii), and configurations are compared with the same paired Wilcoxon signed-rank tests and Holm correction.

\begin{table}
    \centering
    
    \resizebox{\linewidth}{!}{%
    \begin{tabular}{c|ccc|cc|cc|c}
        \toprule
        \multirow{2}{*}{Model name}
        & \multicolumn{3}{c|}{sPB}
        & \multicolumn{2}{c|}{$L$}
        & \multicolumn{2}{c|}{IoU of $\mathcal{R}(X)$}
        & \multirow{2}{*}{$p$-value} \\
        & (i) & (ii) & (iii) & (ii) & (iii) & (i), (ii) & (i), (iii) & \\
        \midrule
        DeepSeek R1 32B & 2.6 (2.4, 2.8) & 3.1 (2.9, 3.2) & 3.5 (3.3, 3.6) & 1.0 & 2.6 & 78.6 & 78.9 & \textbf{3.0e-18} \\
        DeepSeek R1 7B & 9.5 (9.2, 9.8) & 9.8 (9.5, 10.1) & 10.4 (10.1, 10.7) & 4.2 & 3.9 & 71.0 & 71.3 & \textbf{2.4e-07} \\
        DeepSeek R1 8B & 6.0 (5.8, 6.3) & 6.7 (6.5, 6.9) & 7.2 (6.9, 7.4) & 0.8 & 3.7 & 72.4 & 72.4 & \textbf{1.3e-14} \\
        LLaMA 3.1 8B & 6.0 (5.7, 6.2) & 6.7 (6.4, 6.9) & 7.3 (7.0, 7.5) & 3.0 & 0.3 & 73.9 & 74.0 & \textbf{7.8e-21} \\
        LLaMA 4 Scout & 2.7 (2.5, 2.9) & 3.2 (3.0, 3.4) & 3.2 (3.1, 3.4) & 0.9 & 0.9 & 79.2 & 80.6 & \textbf{2.4e-07} \\
        GPT OSS 120B & 1.0 (0.9, 1.1) & 1.0 (1.0, 1.1) & 1.5 (1.4, 1.6) & 3.1 & 4.3 & 77.3 & 79.9 & \textbf{1.0e-14} \\
        GPT-OSS 20B & 2.3 (2.2, 2.5) & 2.1 (2.0, 2.2) & 3.0 (2.8, 3.1) & 4.7 & 7.5 & 73.9 & 76.7 & \textbf{5.7e-13} \\
        LLaMA 3.1 70B & 2.4 (2.3, 2.6) & 3.1 (3.0, 3.3) & 3.2 (3.1, 3.4) & 3.4 & 5.3 & 76.2 & 77.5 & \textbf{2.4e-13} \\
        \bottomrule
    \end{tabular}}
    \caption{Strong positional bias (sPB), POB-UNBP lower bound $L$ of $\mathrm{TV}(p(\cdot), U_k)$ for the configurations with permutations, and mean intersection over union (IoU) of the relevance sets with configuration (i), for top-$5$ under configurations (i) random labels, (ii) random permutation and (iii) both. sPB and IoU are percentages and $L$ is in percentage points, with $95\%$ confidence intervals for sPB. The last column reports the Holm-adjusted $p$-value comparing the sPB of configurations (iii) and (i); bold values are significant at the $0.05$ level.}
    \label{tab:spb_labels}
\end{table}

\paragraph{Identifiers alone.} With a fixed order, random labels already change the recommended item: the sPB of configuration (i), between $1.0\%$ and $9.5\%$, is comparable to the noise floor of stochastic decoding (\autoref{app:sampling_vs_order}). Hence, the identifiers act as a perturbation of the prompt, similar to decoding noise.

\paragraph{Overall positional bias beyond identifiers.} With random labels, configuration (iii), POB-UNBP rejects the null hypothesis of uniformity for all eight models ($p \le 3.5 \cdot 10^{-4}$), so the overall positional bias is not explained by the identifiers either. Its magnitude, however, depends on them: $L$ increases for five models, most markedly for GPT-OSS 20B ($4.7$ to $7.5$ percentage points), remains similar for DeepSeek R1 7B and LLaMA 4 Scout, and drops for LLaMA 3.1 8B ($3.0$ to $0.3$), whose preference for the last prompt positions under positional labels thus largely reflects a preference for the corresponding identifiers.

\paragraph{Strong positional bias beyond identifiers.} Adding permutations to random labels, configuration (iii) versus (i), significantly increases sPB for all eight models, by $0.5$ to $1.3$ percentage points. Moreover, decoupling identifiers from positions never reduces sPB: compared with positional labels, configuration (iii) versus (ii), sPB increases significantly for six models and is unchanged for LLaMA 4 Scout and LLaMA 3.1 70B. Therefore, the strong positional bias is not explained by the identifiers assigned to the items.

The increase of sPB from (ii) to (iii) is largest for the GPT-OSS models ($1.0\%$ to $1.5\%$ and $2.1\%$ to $3.0\%$), whose overall positional bias also increases ($L$ from $3.1$ to $4.3$ and from $4.7$ to $7.5$ percentage points). Hence, their preference for the first prompt positions is not a preference for the identifier ``1''. A possible explanation is that, when identifiers no longer follow the order of the prompt, these models rely more on the position of the items, which also displaces part of their recommendations outside the relevance set. In any case, the essentially soft positional bias of the GPT-OSS models observed in \autoref{app:sampling_vs_order} holds under positional labels, whereas their sPB increases when identifiers and positions are decoupled.

\paragraph{Stability of the relevance set.} The mean IoU between the relevance sets obtained with random labels only and under permutations ranges from $71\%$ to $81\%$, similar to the agreement observed for stochastic decoding (\autoref{app:sampling_vs_order}). The caveat of \autoref{app:sampling_vs_order} also applies: configuration (i) uses the identity permutation, so the product ranked first by the retriever always occupies the first prompt position.

\section{Confidence intervals for permutation sensitivity and strong positional bias}
\label{app:spb_ci}

This appendix reports $95\%$ confidence intervals for the permutation sensitivity (PS) and strong positional bias (sPB) values of \autoref{tab:metrics_topk} and \autoref{tab:metrics_temperature}, obtained by percentile bootstrap at the query level with $B = 2000$ resamples (\autoref{tab:spb_ci_topk} and \autoref{tab:spb_ci_temperature}). It also reports the thresholds $\tau$ across folds computed to estimate the sPB metric.

\begin{table}
    \centering
    
    \begin{tabular}{c|cccc}
        \toprule
        \multirow{2}{*}{Model name} & \multicolumn{4}{c}{Top-$k$} \\
         & $k = 5$ & $k = 10$ & $k = 15$ & $k = 20$ \\
        \midrule
        \multicolumn{5}{c}{PS} \\
        \midrule
        DeepSeek R1 32B & 21.9 (21.0, 22.7) & 29.4 (28.5, 30.4) & 34.1 (33.1, 35.2) & 37.2 (36.2, 38.3) \\
        DeepSeek R1 7B & 34.0 (33.1, 34.9) & 45.1 (44.2, 45.9) & 52.7 (51.8, 53.6) & 58.1 (57.3, 59.0) \\
        DeepSeek R1 8B & 29.8 (28.8, 30.7) & 38.6 (37.7, 39.5) & 43.2 (42.3, 44.2) & 46.4 (45.5, 47.4) \\
        LLaMA 3.1 8B & 30.4 (29.5, 31.3) & 41.9 (41.0, 42.8) & 46.9 (46.0, 47.8) & 49.6 (48.6, 50.5) \\
        LLaMA 4 Scout & 22.4 (21.5, 23.3) & 38.3 (37.3, 39.2) & 51.5 (50.6, 52.3) & 61.6 (60.9, 62.4) \\
        GPT OSS 120B & 15.5 (14.7, 16.3) & 21.9 (21.0, 22.9) & 24.9 (23.9, 25.9) & 27.3 (26.3, 28.3) \\
        GPT-OSS 20B & 17.7 (16.8, 18.5) & 24.3 (23.3, 25.3) & 28.3 (27.2, 29.3) & 30.9 (29.9, 32.0) \\
        LLaMA 3.1 70B & 22.7 (21.8, 23.6) & 31.7 (30.7, 32.8) & 36.7 (35.7, 37.7) & 40.2 (39.2, 41.3) \\
        \midrule
        \multicolumn{5}{c}{sPB} \\
        \midrule
        DeepSeek R1 32B & 3.1 (2.9, 3.2) & 6.9 (6.7, 8.5) & 10.7 (10.4, 11.1) & 13.7 (11.3, 14.2) \\
        DeepSeek R1 7B & 9.8 (9.1, 10.2) & 20.1 (19.2, 20.5) & 26.6 (25.8, 28.2) & 33.2 (32.1, 35.1) \\
        DeepSeek R1 8B & 6.7 (6.5, 7.0) & 14.8 (14.4, 15.1) & 19.4 (17.3, 19.8) & 20.6 (20.2, 21.1) \\
        LLaMA 3.1 8B & 6.7 (6.4, 6.9) & 15.7 (15.3, 16.1) & 20.7 (20.0, 21.2) & 24.1 (23.2, 25.0) \\
        LLaMA 4 Scout & 3.2 (3.0, 3.4) & 12.9 (10.9, 14.4) & 24.5 (23.9, 25.4) & 36.6 (35.6, 37.6) \\
        GPT OSS 120B & 1.0 (1.0, 1.3) & 3.4 (3.2, 3.6) & 5.9 (5.3, 6.2) & 6.7 (6.4, 6.9) \\
        GPT-OSS 20B & 2.1 (1.2, 2.5) & 5.6 (4.0, 5.9) & 6.7 (6.5, 7.6) & 9.4 (8.5, 9.9) \\
        LLaMA 3.1 70B & 3.1 (3.0, 3.3) & 9.5 (8.0, 9.8) & 12.2 (10.9, 14.3) & 14.9 (14.5, 15.3) \\
        \bottomrule
    \end{tabular}
    \caption{Permutation sensitivity (PS) and strong positional bias (sPB) across candidate-set sizes $k$ at temperature $t = 0$, with $95\%$ confidence intervals obtained by percentile bootstrap at the query level. Values are percentages.}
    \label{tab:spb_ci_topk}
\end{table}

\begin{table}
    \centering
    
    \begin{tabular}{c|cc|cc|cc|cc}
        \toprule
        \multirow{3}{*}{Model name} & \multicolumn{8}{c}{Top-$k$} \\
         & \multicolumn{2}{c|}{$k = 5$} & \multicolumn{2}{c|}{$k = 10$} & \multicolumn{2}{c|}{$k = 15$} & \multicolumn{2}{c}{$k = 20$} \\
          & Min. & Max. & Min. & Max. & Min. & Max. & Min. & Max. \\
        \midrule
        DeepSeek R1 32B & 86.8 & 87.4 & 80.6 & 80.8 & 76.7 & 77.1 & 74.0 & 74.7 \\
        DeepSeek R1 7B & 78.4 & 78.8 & 67.3 & 68.3 & 60.6 & 61.1 & 55.0 & 55.7 \\
        DeepSeek R1 8B & 80.8 & 81.6 & 72.9 & 73.5 & 68.8 & 69.5 & 65.9 & 66.9 \\
        LLaMA 3.1 8B & 81.7 & 82.1 & 71.4 & 71.9 & 66.4 & 67.0 & 63.4 & 64.1 \\
        LLaMA 4 Scout & 86.6 & 86.9 & 74.8 & 75.1 & 63.1 & 64.0 & 52.5 & 53.0 \\
        GPT OSS 120B & 90.5 & 91.1 & 85.8 & 86.5 & 83.7 & 84.1 & 81.6 & 82.0 \\
        GPT-OSS 20B & 89.6 & 90.0 & 84.2 & 84.7 & 80.7 & 81.2 & 78.6 & 79.1 \\
        LLaMA 3.1 70B & 86.2 & 86.7 & 78.9 & 79.6 & 74.7 & 75.2 & 71.9 & 72.3 \\
        \bottomrule
    \end{tabular}
    \caption{Minimum and maximum of the threshold $\tau$ of sPB across the five folds, for candidate-set sizes $k$ at temperature $t = 0$. Values are percentages.}
    \label{tab:tau_thresholds_k}
\end{table}

\paragraph{Permutation sensitivity.} PS is an average of per-query values, so each bootstrap resample recomputes this average over the resampled queries. The intervals are narrow, about $\pm 1$ percentage point. The increase of PS with $k$ is robust, as the intervals for consecutive values of $k$ never overlap. The increase with temperature is mild: the intervals at $t = 0$ and $t = 0.6$ are disjoint only for DeepSeek R1 8B and LLaMA 3.1 8B.

\paragraph{Strong positional bias.} Since sPB depends on a threshold $\tau$ calibrated from the data (\autoref{sec:metrics}), each bootstrap resample repeats the whole estimation procedure: the partition into five folds, the calibration of $\tau$ on the training folds, and the computation of the relevance sets and sPB on the held-out fold. Copies of the same instance are always assigned to the same fold. Most intervals are narrow, but some are asymmetric, e.g., DeepSeek R1 32B at $k = 10$, because with $20$ evaluations per instance the item frequencies lie on a grid of step $0.05$, so small variations of $\tau$ across resamples can move many instances between relevance sets of different sizes at once. The increase of sPB with $k$ is robust, as the intervals at $k = 5$ and $k = 20$ are far apart for every model, whereas its differences across temperatures are of the same order as the interval widths, so sPB is essentially stable with respect to temperature, in agreement with \autoref{app:sampling_vs_order}.

The calibrated thresholds are stable across folds: their minimum and maximum differ by at most $1$ percentage point in every setting (\autoref{tab:tau_thresholds_k} and \autoref{tab:tau_thresholds_t}). They decrease with $k$, as the recommendations of each instance spread over more candidates, decrease slightly with temperature for most models, and are highest for the GPT-OSS models, whose recommendations are the most concentrated.

\begin{table}
    \centering
    
    \begin{tabular}{c|cccc}
        \toprule
        \multirow{2}{*}{Model name} & \multicolumn{4}{c}{Temperature $t$} \\
         & $t = 0.0$ & $t = 0.2$ & $t = 0.4$ & $t = 0.6$ \\
        \midrule
        \multicolumn{5}{c}{PS} \\
        \midrule
        DeepSeek R1 32B & 21.9 (21.0, 22.7) & 22.7 (21.8, 23.6) & 22.9 (22.0, 23.8) & 23.4 (22.5, 24.3) \\
        DeepSeek R1 7B & 34.0 (33.1, 34.9) & 34.6 (33.8, 35.5) & 35.3 (34.4, 36.2) & 35.9 (35.0, 36.8) \\
        DeepSeek R1 8B & 29.8 (28.8, 30.7) & 30.9 (30.0, 31.8) & 31.5 (30.5, 32.3) & 32.6 (31.7, 33.4) \\
        LLaMA 3.1 8B & 30.4 (29.5, 31.3) & 31.4 (30.5, 32.3) & 32.3 (31.4, 33.1) & 33.0 (32.0, 33.8) \\
        LLaMA 4 Scout & 22.4 (21.5, 23.3) & 22.6 (21.7, 23.5) & 22.6 (21.8, 23.5) & 22.9 (21.9, 23.7) \\
        GPT OSS 120B & 15.5 (14.7, 16.3) & 15.6 (14.8, 16.4) & 15.7 (15.0, 16.5) & 16.0 (15.2, 16.8) \\
        GPT-OSS 20B & 17.7 (16.8, 18.5) & 17.9 (17.0, 18.8) & 18.6 (17.7, 19.5) & 18.6 (17.7, 19.4) \\
        LLaMA 3.1 70B & 22.7 (21.8, 23.6) & 22.8 (21.9, 23.8) & 23.7 (22.8, 24.6) & 24.2 (23.3, 25.1) \\
        \midrule
        \multicolumn{5}{c}{sPB} \\
        \midrule
        DeepSeek R1 32B & 3.1 (2.9, 3.2) & 3.3 (3.1, 3.5) & 3.4 (3.2, 4.0) & 3.5 (3.3, 5.0) \\
        DeepSeek R1 7B & 9.8 (9.1, 10.2) & 10.0 (9.6, 10.5) & 10.4 (9.9, 10.8) & 10.4 (10.0, 10.8) \\
        DeepSeek R1 8B & 6.7 (6.5, 7.0) & 6.9 (6.7, 9.1) & 8.7 (7.0, 9.9) & 9.9 (9.4, 10.2) \\
        LLaMA 3.1 8B & 6.7 (6.4, 6.9) & 7.1 (6.8, 8.7) & 9.3 (7.0, 9.8) & 9.4 (8.1, 10.1) \\
        LLaMA 4 Scout & 3.2 (3.0, 3.4) & 3.2 (3.0, 3.4) & 3.2 (3.1, 3.4) & 3.3 (3.1, 3.5) \\
        GPT OSS 120B & 1.0 (1.0, 1.3) & 1.0 (1.0, 1.4) & 1.2 (1.1, 1.7) & 1.1 (1.0, 1.6) \\
        GPT-OSS 20B & 2.1 (1.2, 2.5) & 2.5 (2.1, 2.7) & 2.6 (1.7, 2.8) & 2.7 (2.4, 2.8) \\
        LLaMA 3.1 70B & 3.1 (3.0, 3.3) & 3.4 (3.3, 3.6) & 3.5 (3.3, 4.0) & 3.5 (3.3, 4.0) \\
        \bottomrule
    \end{tabular}
    \caption{Permutation sensitivity (PS) and strong positional bias (sPB) across decoding temperatures $t$ at top-$5$, with $95\%$ confidence intervals obtained by percentile bootstrap at the query level. Values are percentages.}
    \label{tab:spb_ci_temperature}
\end{table}

\begin{table}
    \centering
    
    \begin{tabular}{c|cc|cc|cc|cc}
        \toprule
        \multirow{3}{*}{Model name} & \multicolumn{8}{c}{Temperature $t$} \\
         & \multicolumn{2}{c|}{$t = 0.0$} & \multicolumn{2}{c|}{$t = 0.2$} & \multicolumn{2}{c|}{$t = 0.4$} & \multicolumn{2}{c}{$t = 0.6$} \\
          & Min. & Max. & Min. & Max. & Min. & Max. & Min. & Max. \\
        \midrule
        DeepSeek R1 32B & 86.8 & 87.4 & 86.1 & 86.2 & 85.6 & 86.0 & 85.1 & 85.3 \\
        DeepSeek R1 7B & 78.4 & 78.8 & 77.9 & 78.2 & 77.7 & 78.6 & 76.9 & 77.6 \\
        DeepSeek R1 8B & 80.8 & 81.6 & 80.1 & 80.8 & 79.7 & 80.1 & 78.5 & 79.1 \\
        LLaMA 3.1 8B & 81.7 & 82.1 & 80.2 & 80.9 & 79.4 & 79.7 & 79.0 & 79.6 \\
        LLaMA 4 Scout & 86.6 & 86.9 & 86.0 & 86.7 & 86.4 & 86.8 & 86.4 & 86.6 \\
        GPT OSS 120B & 90.5 & 91.1 & 90.6 & 90.9 & 90.6 & 90.8 & 90.6 & 91.0 \\
        GPT-OSS 20B & 89.6 & 90.0 & 89.1 & 89.4 & 89.1 & 89.7 & 88.8 & 89.2 \\
        LLaMA 3.1 70B & 86.2 & 86.7 & 86.0 & 86.6 & 85.6 & 86.0 & 85.4 & 86.2 \\
        \bottomrule
    \end{tabular}
    \caption{Minimum and maximum of the threshold $\tau$ of sPB across the five folds, for decoding temperatures $t$ at top-$5$. Values are percentages.}
    \label{tab:tau_thresholds_t}
\end{table}

\section{Effect of retrieval score dispersion on strong positional bias}
\label{app:retrieval_dispersion}

Although the random permutations prevent retrieval relevance from favoring any prompt position (\autoref{sec:results_icm}), the relevance structure of the candidates may still modulate the magnitude of positional bias. We measure it through the retrieval score dispersion $\delta(X)$, the standard deviation of the cosine similarities between the query and its top-$k$ candidates, standardized within each setting. Low values correspond to candidates that the retriever barely distinguishes.

For each model and setting, we fit a weighted linear regression of the instance-level strong positional bias, that is, the percentage of valid evaluations of $X$ recommending an item outside $\mathcal{R}(X)$, on $\delta(X)$, weighting each instance by its number of valid evaluations. \autoref{tab:dispersion_topk} and \autoref{tab:dispersion_temperature} report the slope, in percentage points of sPB per standard deviation of $\delta(X)$, with $95\%$ confidence intervals based on heteroscedasticity-robust (HC1) standard errors, and the same change relative to the sPB of the model in that setting. We test the null hypothesis that the slope is zero, i.e., that the strong positional bias of a query does not depend linearly on its retrieval score dispersion, with a two-sided Wald test based on the same HC1 standard errors. The resulting $p$-values are Holm-adjusted jointly across the $32$ tests of each table, that is, across the eight models and the four values of $k$ (\autoref{tab:dispersion_topk}) or of $t$ (\autoref{tab:dispersion_temperature}).

The slope is negative in every configuration: queries whose candidates are less separated by the retriever show slightly more strong positional bias. The effect is small, at most about $1$ pp per standard deviation, or a few percent of the strong positional bias of each model, and it is significant mainly for $k \ge 10$, while at $k = 5$ it is significant only for a few models and shows no trend with temperature. Hence, retrieval score dispersion does not explain the positional bias detected in this work, and only marginally modulates its strong component.

\begin{table}
    \centering
    \resizebox{\linewidth}{!}{%
    \begin{tabular}{l|cc|cc|cc|cc}
        \toprule
         & \multicolumn{2}{c|}{$k = 5$} & \multicolumn{2}{c|}{$k = 10$} & \multicolumn{2}{c|}{$k = 15$} & \multicolumn{2}{c}{$k = 20$} \\
        Model name & $\Delta$sPB/SD (95\% CI) & Rel. & $\Delta$sPB/SD (95\% CI) & Rel. & $\Delta$sPB/SD (95\% CI) & Rel. & $\Delta$sPB/SD (95\% CI) & Rel. \\
        \midrule
        DeepSeek R1 32B & -0.08 (-0.25, 0.09) & -2.6\% & \textbf{-0.41 (-0.66, -0.16)} & \textbf{-5.8\%} & \textbf{-0.58 (-0.89, -0.27)} & \textbf{-5.4\%} & \textbf{-0.66 (-1.05, -0.26)} & \textbf{-4.8\%} \\
        DeepSeek R1 7B & -0.26 (-0.58, 0.05) & -2.7\% & \textbf{-0.60 (-0.96, -0.23)} & \textbf{-3.0\%} & \textbf{-1.09 (-1.49, -0.68)} & \textbf{-4.1\%} & \textbf{-0.95 (-1.33, -0.56)} & \textbf{-2.9\%} \\
        DeepSeek R1 8B & -0.30 (-0.54, -0.05) & -4.4\% & \textbf{-0.84 (-1.19, -0.49)} & \textbf{-5.7\%} & \textbf{-0.59 (-0.98, -0.20)} & \textbf{-3.0\%} & \textbf{-0.94 (-1.32, -0.55)} & \textbf{-4.5\%} \\
        LLaMA 3.1 8B & -0.22 (-0.47, 0.03) & -3.3\% & \textbf{-0.81 (-1.17, -0.46)} & \textbf{-5.2\%} & \textbf{-0.90 (-1.29, -0.50)} & \textbf{-4.3\%} & -0.60 (-1.02, -0.17) & -2.5\% \\
        LLaMA 4 Scout & \textbf{-0.30 (-0.48, -0.12)} & \textbf{-9.3\%} & \textbf{-0.85 (-1.22, -0.48)} & \textbf{-6.1\%} & \textbf{-0.89 (-1.27, -0.51)} & \textbf{-3.7\%} & \textbf{-1.08 (-1.44, -0.72)} & \textbf{-3.0\%} \\
        GPT OSS 120B & -0.01 (-0.10, 0.08) & -1.0\% & -0.20 (-0.38, -0.02) & -5.9\% & \textbf{-0.38 (-0.63, -0.13)} & \textbf{-6.4\%} & \textbf{-0.48 (-0.73, -0.23)} & \textbf{-7.2\%} \\
        GPT-OSS 20B & -0.16 (-0.32, 0.01) & -6.8\% & \textbf{-0.39 (-0.64, -0.15)} & \textbf{-7.2\%} & \textbf{-0.54 (-0.80, -0.28)} & \textbf{-8.0\%} & \textbf{-0.49 (-0.82, -0.16)} & \textbf{-5.1\%} \\
        LLaMA 3.1 70B & -0.07 (-0.24, 0.11) & -2.1\% & -0.38 (-0.69, -0.07) & -4.0\% & -0.35 (-0.74, 0.03) & -2.5\% & \textbf{-0.63 (-1.03, -0.22)} & \textbf{-4.2\%} \\
        \bottomrule
    \end{tabular}}
    \caption{Change of sPB, in percentage points, per standard deviation of retrieval score dispersion (95\% CI), and relative to the sPB of each model (Rel.), across candidate-set sizes at $t = 0$. Bold: significant at the $0.05$ level after Holm correction.}
    \label{tab:dispersion_topk}
\end{table}

\begin{table}
    \centering
    \resizebox{\linewidth}{!}{%
    \begin{tabular}{l|cc|cc|cc|cc}
        \toprule
         & \multicolumn{2}{c|}{$t = 0.0$} & \multicolumn{2}{c|}{$t = 0.2$} & \multicolumn{2}{c|}{$t = 0.4$} & \multicolumn{2}{c}{$t = 0.6$} \\
        Model name & $\Delta$sPB/SD (95\% CI) & Rel. & $\Delta$sPB/SD (95\% CI) & Rel. & $\Delta$sPB/SD (95\% CI) & Rel. & $\Delta$sPB/SD (95\% CI) & Rel. \\
        \midrule
        DeepSeek R1 32B & -0.08 (-0.25, 0.09) & -2.6\% & -0.07 (-0.25, 0.11) & -2.2\% & -0.19 (-0.37, -0.02) & -5.7\% & -0.05 (-0.23, 0.13) & -1.5\% \\
        DeepSeek R1 7B & -0.26 (-0.58, 0.05) & -2.7\% & \textbf{-0.66 (-0.96, -0.36)} & \textbf{-6.6\%} & \textbf{-0.74 (-1.04, -0.43)} & \textbf{-7.1\%} & -0.38 (-0.68, -0.07) & -3.6\% \\
        DeepSeek R1 8B & -0.30 (-0.54, -0.05) & -4.4\% & -0.16 (-0.40, 0.09) & -2.2\% & -0.42 (-0.72, -0.13) & -4.4\% & \textbf{-0.51 (-0.81, -0.20)} & \textbf{-5.1\%} \\
        LLaMA 3.1 8B & -0.22 (-0.47, 0.03) & -3.3\% & -0.33 (-0.58, -0.08) & -4.7\% & \textbf{-0.55 (-0.85, -0.26)} & \textbf{-5.9\%} & -0.17 (-0.46, 0.12) & -1.8\% \\
        LLaMA 4 Scout & \textbf{-0.30 (-0.48, -0.12)} & \textbf{-9.3\%} & -0.17 (-0.34, 0.00) & -5.3\% & \textbf{-0.34 (-0.51, -0.17)} & \textbf{-10.4\%} & \textbf{-0.29 (-0.46, -0.12)} & \textbf{-8.7\%} \\
        GPT OSS 120B & -0.01 (-0.10, 0.08) & -1.0\% & -0.06 (-0.15, 0.03) & -5.4\% & -0.06 (-0.15, 0.04) & -4.7\% & -0.02 (-0.11, 0.07) & -2.0\% \\
        GPT-OSS 20B & -0.16 (-0.32, 0.01) & -6.8\% & -0.10 (-0.27, 0.06) & -4.1\% & -0.12 (-0.30, 0.05) & -4.8\% & -0.16 (-0.32, 0.01) & -5.9\% \\
        LLaMA 3.1 70B & -0.07 (-0.24, 0.11) & -2.1\% & -0.23 (-0.41, -0.05) & -6.8\% & -0.16 (-0.33, 0.00) & -4.8\% & -0.13 (-0.30, 0.05) & -3.6\% \\
        \bottomrule
    \end{tabular}
    }
    \caption{As \autoref{tab:dispersion_topk}, across decoding temperatures at $k = 5$.}
    \label{tab:dispersion_temperature}
\end{table}

\section{Sensitivity to discarded generations}
\label{app:failures_order}

Generations that cannot be parsed, in which the model refuses to recommend an item, or that return an invalid position are discarded before applying POB-UNBP (\autoref{sec:experimental_setup}). If these failures depended on the ordering of the candidate items, discarding them could bias the estimated position distributions. To assess this, we repeat POB-UNBP restricted to the complete-case queries, i.e., those for which none of the generations was discarded, so that every permutation of each query contributes to the estimate. \autoref{thm:stochastic} holds for any set of queries, so restricting the analysis to a subset of queries does not change the null hypothesis.

\autoref{tab:complete_case} reports, for each model and candidate-set size, the lower bound $L$ of the total variation distance to the uniform distribution and the $p$-value of the POB-UNBP test, computed on all queries (All) and on the complete-case queries (CC). The complete-case subset retains between $44\%$ and $98\%$ of the queries, depending on the model and candidate-set size, which is well above the sample sizes required for high power (\autoref{sec:simulation_study}). The test rejects the null hypothesis of absence of positional bias in all configurations on the complete-case queries, with the largest $p$-value being $0.002$ (DeepSeek R1 32B, top-5). The lower bound $L$ also remains stable: for most configurations it changes by less than $0.01$, and the conclusions drawn in \autoref{sec:results} hold for both sets of queries. The largest relative decrease corresponds to DeepSeek R1 32B at top-5 ($0.010$ to $0.006$), which already exhibited the weakest bias, and the largest absolute decrease to LLaMA 4 Scout at top-15 and top-20 ($0.214$ to $0.194$ and $0.302$ to $0.281$), which remains the most biased model by a wide margin. Conversely, for DeepSeek R1 7B, the model with the highest fraction of discarded generations, $L$ increases on the complete-case queries for top-5 and top-10 ($0.042$ to $0.053$ and $0.058$ to $0.074$), so discarding failures does not inflate the measured bias. Hence, the exclusion of failed generations does not explain our findings.

Note that this analysis does not determine whether failures depend on the ordering of the candidate items; rather, it shows that, whether or not they do, their exclusion does not alter the estimated positional bias.

\begin{table}[t]
    \centering
    
    \resizebox{\textwidth}{!}{%
    \begin{tabular}{cc|cc|cc|cc|cc}
        \toprule
        \multirow{2}{*}{Model name} & \multirow{2}{*}{Metric} & \multicolumn{2}{c|}{Top-5} & \multicolumn{2}{c|}{Top-10} & \multicolumn{2}{c|}{Top-15} & \multicolumn{2}{c}{Top-20} \\
         & & All & CC & All & CC & All & CC & All & CC \\
        \midrule
        \multirow{2}{*}{DeepSeek R1 32B} & $L$ & 0.970{\%} & 0.563{\%} & 2.827{\%} & 2.291{\%} & 4.291{\%} & 4.114{\%} & 6.775{\%} & 6.618{\%} \\
         & $p$-value & 7.456e-7 & 0.002 & 5.130e-17 & 7.082e-17 & 2.202e-47 & 2.072e-38 & 2.098e-48 & 2.859e-40 \\
        \midrule
        \multirow{2}{*}{DeepSeek R1 7B} & $L$ & 4.151{\%} & 5.279{\%} & 5.786{\%} & 7.430{\%} & 9.506{\%} & 9.432{\%} & 9.894{\%} & 9.969{\%} \\
         & $p$-value & 6.197e-60 & 2.169e-64 & 9.362e-55 & 1.299e-45 & 5.966e-79 & 2.045e-47 & 4.449e-70 & 2.075e-33 \\
        \midrule
        \multirow{2}{*}{DeepSeek R1 8B} & $L$ & 0.788{\%} & 0.785{\%} & 1.959{\%} & 2.093{\%} & 4.303{\%} & 4.387{\%} & 7.238{\%} & 7.256{\%} \\
         & $p$-value & 5.875e-5 & 9.146e-7 & 1.150e-11 & 2.483e-12 & 2.253e-46 & 2.989e-46 & 2.111e-71 & 3.821e-70 \\
        \midrule
        \multirow{2}{*}{LLaMA 3.1 8B} & $L$ & 3.009{\%} & 2.923{\%} & 3.131{\%} & 2.960{\%} & 4.078{\%} & 3.993{\%} & 3.768{\%} & 3.986{\%} \\
         & $p$-value & 4.457e-26 & 4.831e-24 & 4.928e-17 & 3.058e-15 & 3.379e-16 & 9.291e-15 & 1.465e-26 & 1.886e-15 \\
        \midrule
        \multirow{2}{*}{LLaMA 4 Scout} & $L$ & 0.872{\%} & 0.792{\%} & 8.298{\%} & 7.223{\%} & 21.430{\%} & 19.403{\%} & 30.189{\%} & 28.097{\%} \\
         & $p$-value & 1.188e-4 & 8.446e-5 & 1.438e-60 & 2.975e-48 & 8.023e-208 & 1.121e-125 & 1.341e-249 & 1.913e-137 \\
        \midrule
        \multirow{2}{*}{GPT OSS 120B} & $L$ & 3.144{\%} & 3.188{\%} & 3.621{\%} & 3.614{\%} & 5.099{\%} & 5.088{\%} & 6.997{\%} & 7.053{\%} \\
         & $p$-value & 8.879e-30 & 3.145e-30 & 2.233e-41 & 3.083e-41 & 6.968e-48 & 2.612e-47 & 1.601e-58 & 9.038e-60 \\
        \midrule
        \multirow{2}{*}{GPT-OSS 20B} & $L$ & 4.733{\%} & 4.670{\%} & 5.690{\%} & 5.603{\%} & 7.779{\%} & 7.187{\%} & 9.034{\%} & 8.474{\%} \\
         & $p$-value & 4.814e-65 & 5.262e-59 & 5.034e-82 & 5.072e-70 & 6.788e-96 & 3.230e-79 & 2.993e-100 & 1.016e-85 \\
        \midrule
        \multirow{2}{*}{LLaMA 3.1 70B} & $L$ & 3.396{\%} & 3.328{\%} & 2.990{\%} & 3.030{\%} & 4.402{\%} & 4.417{\%} & 7.039{\%} & 7.096{\%} \\
         & $p$-value & 7.216e-24 & 3.409e-23 & 4.845e-25 & 9.634e-26 & 2.463e-29 & 1.214e-29 & 1.045e-41 & 3.221e-41 \\
        \bottomrule
    \end{tabular}
    }
    \caption{Sensitivity of POB-UNBP to the discarded generations. For each model, the first row reports the lower bound $L$ of the total variation distance to the uniform distribution and the second row the $p$-value of the POB-UNBP test, computed on all queries (All) and on the complete-case queries (CC), for which no generation was discarded.}
    \label{tab:complete_case}
\end{table}

\end{document}